\documentclass[11pt]{article}

\pdfoutput=1

\let\originalleft\left
\let\originalright\right
\renewcommand{\left}{\mathopen{}\mathclose\bgroup\originalleft}
\renewcommand{\right}{\aftergroup\egroup\originalright}

\usepackage{autobreak}
\usepackage[utf8]{inputenc}
\usepackage{geometry, amsmath, amsthm, amsfonts, mathtools, fullpage, changepage, paralist}
\usepackage[shortlabels]{enumitem}
  \setlist[itemize]{leftmargin=*}
  \setlist[enumerate]{leftmargin=*}
\usepackage{framed, algorithm, algpseudocode}
\usepackage{thmtools, thm-restate}
\usepackage[nottoc,notlot,notlof]{tocbibind}
\usepackage{tcolorbox, multirow, tabularx, makecell}

\usepackage[single]{accents}
\usepackage{scalerel}

\usepackage{xurl}

\usepackage{cellspace}
\usepackage[backend=biber,
            style=trad-alpha,
            backref=true,
            backrefstyle=three,
            isbn=false,
            maxbibnames=9,
            maxcitenames=2]{biblatex}
\bibliography{refs}
\DefineBibliographyStrings{english}{
  backrefpage = {p.\ },
  backrefpages = {pp.\ }
}

\usepackage{hyperref}
\hypersetup{
    breaklinks=true,
    colorlinks=true,
    allcolors=blue,
    linktocpage
}
\usepackage[capitalize]{cleveref}

\allowdisplaybreaks

\theoremstyle{plain}
\newtheorem{thm}{Theorem}[section]
\newtheorem{lem}[thm]{Lemma}
\newtheorem{cor}[thm]{Corollary}

\newtheorem{obs}[thm]{Observation}

\crefname{lem}{Lemma}{Lemmas}
\crefname{thm}{Theorem}{Theorems}
\crefname{cor}{Corollary}{Corollaries}
\crefname{prp}{Proposition}{Propositions}

\theoremstyle{definition}
\newtheorem{dfn}[thm]{Definition}
\newtheorem{xmp}[thm]{Example}
\crefname{xmp}{Example}{Examples}

\newcommand{\eps}{\varepsilon}
\newcommand{\N}{\mathbb{N}}
\newcommand{\R}{\mathbb{R}}
\newcommand{\C}{\mathbb{C}}

\newcommand{\norm}[1]{{\|#1\|}}

\newcommand{\bits}{\{0,1\}}

\DeclareMathOperator*{\E}{\mathbb{E}}

\DeclarePairedDelimiter{\ceil}{\lceil}{\rceil}

\newcommand*\fid[2]{\mathrm{F}\Paren{#1, #2}}
\newcommand*{\Mag}[1]{\left| #1 \right|}
\newcommand*{\adj}[1]{#1^\dagger}
\newcommand*{\ketbra}[2]{|#1\rangle\!\langle #2|}
\newcommand*{\kb}[1]{\ketbra{#1}{#1}}
\newcommand*{\braket}[2]{\langle #1| #2\rangle}
\newcommand*{\Norm}[1]{\left\|#1\right\|}
\newcommand{\Paren}[1]{\left(#1\right)}

\newcommand*\dens[1]{\mathsf{D}(#1)}
\newcommand*\Dens[1]{\mathsf{D}\Paren{#1}}
\newcommand*\super[2]{\mathsf{S}(#1,#2)}
\newcommand*\spr[1]{\super #1 *}
\newcommand*\cptp[2]{\mathsf{C}(#1,#2)}
\newcommand*\sym[1]{\Pi^{\mrm{sym}}_{#1}}
\newcommand*\swap[1]{\mathrm{SWAP}_{#1}}

\newcommand*{\pr}[1]{\mathrm{Pr}(#1)}

\newcommand*{\PR}[1]{\mathrm{Pr}\left(#1\right)}

\newcommand*{\mrm}[1]{\mathrm{#1}}
\newcommand{\Brac}[1]{\left[#1\right]}

\DeclareMathOperator*{\tr}{tr}
\DeclareMathOperator*{\rank}{rank}

\newcommand*{\din}{d_{\mathrm{in}}}
\newcommand*{\dout}{d_{\mathrm{out}}}
\newcommand*{\daux}{d_{\mathrm{anc}}}

\newcommand*{\sphere}[1]{\mathbb S^{#1}}

\newcommand*\reg\mathsf

\newcommand*\chan\mathcal
\newcommand*{\dual}[1]{\overline{\chan #1}}
\newcommand*\bs\boldsymbol
\newcommand*\lsd{J}
\newcommand*\choi[1]{J_{\chan{#1}}}

\newcommand{\idsupop}{\mathcal{I}}

\newcommand*{\bra}[1]{\ensuremath{\langle #1|}}
\newcommand*{\ket}[1]{\ensuremath{|#1\rangle}}

\title{Quantum Channel Testing in Average-Case Distance}
\author{Gregory Rosenthal\footnote{University of Cambridge, University of Warwick. gar52@cam.ac.uk.}
\and Hugo Aaronson\footnote{University of Cambridge. ha406@cam.ac.uk.}
\and Sathyawageeswar Subramanian\footnote{University of Cambridge. ss2310@cam.ac.uk.}
\and Animesh Datta\footnote{University of Warwick. animesh.datta@warwick.ac.uk.}
\and Tom Gur\footnote{University of Cambridge. tom.gur@cl.cam.ac.uk.}}

\date{}

\begin{document}

\maketitle

\begin{abstract}
We study the complexity of testing properties of quantum channels. First, we show that testing identity to \emph{any} channel $\chan N: \C^{\din \times \din} \to \C^{\dout \times \dout}$ in the diamond-norm distance requires $\Omega(\sqrt{\din}/\eps)$ queries, even in the strongest algorithmic model allowing ancillae, coherence, and adaptivity. This hardness arises from the worst-case nature of the distance induced by the diamond norm.

Motivated by this limitation, and by both theoretical and practical considerations, we initiate a systematic study of the trace distance between Choi states of channels, which we argue serves as an average-case analogue of the diamond norm. We refer to this measure as the \emph{average-case imitation diamond (ACID)} norm for brevity. In the weakest algorithmic model—without ancillae, coherence, or adaptivity—we prove that testing identity to certain types of channels in ACID distance can be performed with complexity \emph{independent of the channel dimensions}, whereas for other channel classes the complexity depends on both the input and output dimensions. Building on previous work, we further show that identity to \emph{any} fixed channel can be tested with $\tilde{O}(\din\dout^{3/2}/\eps^2)$ queries in ACID distance and $\tilde{O}(\din^2\dout^{3/2}/\eps^2)$ queries in diamond distance under this model. Finally, we establish tight bounds on the complexity of channel tomography in ACID distance.
\end{abstract}
\newpage

\setcounter{tocdepth}{3}
\thispagestyle{empty}
\newpage
\tableofcontents
\newpage
\pagenumbering{arabic}

\section{Introduction} \label{sec:intro}


Property testing is concerned with the task of efficiently distinguishing whether a large object satisfies a given property or is far from all objects with that property, with respect to a meaningful notion of distance. In the setting of quantum computing, one may seek quantum testers for both classical objects such as Boolean functions, and quantum objects such as states or unitary transformations, as discussed in surveys by Montanaro and de Wolf~\cite{MdW13} and O'Donnell and Wright~\cite{OW21} and in a paper of Linden and de Wolf \cite{LdW21}.

Unlike most previous work, this paper is concerned with testing properties of \emph{quantum channels}, which capture the most general dynamics of quantum systems. The state of a $d$-dimensional quantum system is described by a \emph{density matrix} in $\C^{d \times d}$, meaning a positive semidefinite matrix with unit trace. A quantum channel (henceforth just ``channel") is a \emph{superoperator} or linear transformation from $\C^{\din \times \din}$ to $\C^{\dout \times \dout}$, all of whose trivial extensions are required to map every input density matrix to an output density matrix.
 
Fawzi, Flammarion, Garivier and Oufkir~\cite{FFGO23} considered the problem of testing whether a given blackbox implements a fixed channel $\chan N$ or is $\eps$-far from $\chan N$ in the diamond norm. This task is called \emph{testing identity to $\chan N$}, and also called \emph{channel certification}. In the weakest algorithmic model without ancillae or adaptivity, they proved that $d/\eps^{\Theta(1)}$ queries to the blackbox are necessary and sufficient to test identity to a fixed unitary channel. They also showed that $\tilde\Theta \Paren{\din^2 \dout^{3/2} / \eps^2}$ queries are necessary and sufficient to test identity to the completely depolarizing channel, which maps every $\din$-dimensional input state to the $\dout$-dimensional maximally mixed state.

However, the polynomial dependence on $\din$ and $\dout$ in the complexity of these channel testers is unsatisfactory. The goal of property testing is to obtain ultra-fast algorithms that only probe a tiny portion of their input. Indeed, a property is said to be ``testable'' if it can be tested with complexity that depends only on the proximity parameter $\eps$ and not on the size or dimension of the object. Quantum objects are large, as the dimension of the state space of a collection of $n$ quantum systems scales exponentially in $n$, so it is critical to obtain channel testers that (at worst) query the blackbox channel a number of times polylogarithmic in the dimensions of that channel.

The problem here is that diamond distance is a \emph{worst-case} distance, defined via a maximization over all input states, so two channels can be far apart even if they behave similarly except near a single input state. It is natural that such channels cannot be distinguished by a tester that does not consider the action of the blackbox channel on a large part of its input domain. In contrast, testers for Boolean functions measure distance by the fraction of the domain on which two functions differ, and this notion of statistical distance inherently captures average-case behavior. Property testing algorithms in general capitalize on local-to-global phenomena that typically arise in such average-case settings. 

The above observation underlies the main contributions of this work. We systematically study an average-case analogue of the diamond norm and simultaneously quantify the limitations of channel testing with respect to the diamond norm. This average-case norm is mathematically well-motivated and has properties that make it attractive from a complexity theoretic viewpoint. Indeed, we are able to demonstrate the power of channel testing in this average-case distance for several natural classes of channels. 

\subsection{Hardness of channel testing in diamond distance}
\label{subsec:diamond-hardness}

Our first result is a $\din^{\Omega(1)}/\eps$ lower bound for testing identity to \emph{any} fixed channel in diamond distance, even in the strongest query model that allows ancillae, coherence and adaptivity. (By \emph{coherence} we mean entanglement between subsystems associated with different queries; see \cref{sec:query-models} for formal definitions of the different query models that we consider.) This provides motivation to test with respect to an average-case distance where dimension-independent complexity may be achieved.

To make this precise, recall that the trace norm $\norm{X}_1$ of a matrix $X$ equals the sum of its singular values. The trace distance $\frac12 \norm{\rho-\sigma}_1$ between states $\rho$ and $\sigma$ generalizes the notion of statistical distance between probability distributions. The trace norm for matrices induces a corresponding trace norm for superoperators, defined by $\norm{\chan L}_1 \coloneqq \max_{\norm{X}_1\leq 1} \norm{\chan L(X)}_1$ for a superoperator $\chan L$. The completely bounded trace norm, more commonly known as the \emph{diamond norm}, is defined similarly but with the maximum taken over all trivial extensions of the superoperator: for $\chan L:\C^{\din \times \din}\to\C^{\dout \times \dout}$,
\begin{equation}
\label{eq:diamond}
    \norm{\chan L}_{\diamond}
    \coloneqq \norm{\chan L \otimes \chan I_{\din}}_1
    = \displaystyle\max_{\norm{X}_1\leq 1}\norm{(\chan L \otimes \idsupop_{\din}) \cdot X}_1,
\end{equation}
where $\idsupop_{\din}: \C^{\din \times \din} \to \C^{\din \times \din}$ is the identity map. This modification of the trace norm is particularly appealing, as the distance induced by the diamond norm has a natural operational interpretation, quantifying the distinguishability between two channels when arbitrary input states and measurements are allowed. We prove the following:\footnote{By ``success probability at least $2/3$" in the theorem statement, we mean that the tester accepts with probability at least $2/3$ if $\chan M = \chan N$ and rejects with probability at least $2/3$ if $\norm{\chan M - \chan N}_\diamond \ge \eps$.

Inspection of the proof of \cref{thminf:diamond-LB} reveals that it also holds with the induced trace norm in place of the diamond norm; however, we will focus our discussion on the diamond norm for simplicity.}

\begin{restatable}[Lower bound for channel certification in diamond distance]{thm}{diamondlb}
\label{thminf:diamond-LB}
    For all fixed channels $\chan N: \C^{\din \times \din} \to \C^{\dout \times \dout}$ with $\dout \ge 2$ and all $\eps > 0$, every ancilla-assisted, coherent, adaptive algorithm requires $\Omega(\sqrt{\din} / \eps)$ queries to a channel $\chan M$ to decide whether $\chan M = \chan N$ or $\norm{\chan M - \chan N}_\diamond \ge \eps$ with success probability at least $2/3$.
\end{restatable}

We prove this result by considering the trace distance between the outputs of the fixed channel and the same channel perturbed in a Haar random direction in the input. Exploiting the Haar random nature of the perturbation yields the $\sqrt\din$ factor in the sample complexity lower bound.

\cref{thminf:diamond-LB} generalizes the observation of Montanaro and de Wolf~\cite[Section 5.1.1]{MdW13} that testing identity to a unitary channel in diamond distance requires $\Omega(\sqrt{d})$ queries, by a reduction to the lower bound for unstructured search. We conjecture that the lower bound in \cref{thminf:diamond-LB} can be improved to $\Omega(\din/\eps)$, as we achieve for even the extremely simple channel that always outputs a fixed pure state regardless of its input:

\begin{restatable}[Lower bound for pure state replacement channel certification in diamond distance]{thm}{replacementlb} \label{rmk:ReplacementLB}
    Let $\chan N: \C^{\din \times \din} \to \C^{\dout \times \dout}$ be a pure state replacement channel, i.e.\ $\chan N(X) = \tr(X) \theta$ for some fixed pure state $\theta$ of dimension $\dout \ge 2$, and let $\eps>0$. Then every ancilla-assisted, coherent, adaptive algorithm requires $\Omega(\din/\eps)$ queries to a channel $\chan M$ to decide whether $\chan M = \chan N$ or $\norm{\chan M - \chan N}_\diamond \ge \eps$ with success probability at least $2/3$.
\end{restatable}

\subsection{An average-case analogue of the diamond norm}
\label{subsec:main-results}
Thus motivated, we now study an average-case analogue of the diamond norm. A natural approach is to replace the maximum in the definition \cref{eq:diamond} of the diamond norm with an expectation: for a superoperator $\chan L:\C^{\din \times \din} \to \C^{\dout \times \dout}$, let
\begin{equation*}
    \norm{\chan L}_{\mathrm{avg}} \coloneqq \displaystyle\E_{\bs\psi}\norm{(\chan L \otimes \idsupop_{\din}) \cdot\bs\psi}_1,
\end{equation*}
where $\bs\psi \in \Paren{\C^{\din\times\din}}^{\otimes 2}$ is a Haar random (pure) state.\footnote{Throughout the paper, we will use boldface font to denote random variables.} However, $\norm\cdot_{\mathrm{avg}}$ has the undesirable feature of being sensitive to the dimension of the ancillary register. In the definition \cref{eq:diamond} of the diamond norm, this register may have dimension $\din$ without loss of generality~\cite[Theorem 3.46]{Wat18}, in the sense that if $X$ ranges over $\C^{\din \times \din} \otimes \C^{\daux \times \daux}$ for some $\daux \ge \din$ then
\begin{equation*}
    \norm{\chan L}_\diamond = \max_{\norm{X}_1 \le 1} \Norm{(\chan L \otimes \chan I_{\daux}) \cdot X}_1.
\end{equation*}
If an analogous statement were to fail to hold for $\norm\cdot_{\mrm{avg}}$, then it would not be clear why any one value of $\daux$ should be better motivated than any other. It is also not immediately clear that $\norm{(\chan L \otimes \chan I_{\daux}) \cdot\bs\psi}_1$ is concentrated around its mean, even when $\daux = \din$, and this condition is necessary for $\norm\cdot_{\mrm{avg}}$ to describe the behavior of $\chan L$ on ``typical" inputs (unlike the diamond norm).

Luckily though, for a wide range of values of $\daux$, the quantity $\norm{(\chan L \otimes \chan I_{\daux}) \cdot\bs\psi}_1$ \emph{does} concentrate around its mean, and furthermore its mean is \emph{independent} of $\daux$ (up to a universal constant factor). To state this result more precisely, let
\begin{equation*}
    \Phi_d = \frac1d \sum_{i,j=1}^d \ketbra{ii}{jj}
\end{equation*}
denote the maximally entangled state, and let
\begin{equation*}
    \choi L \coloneqq \Paren{\chan L \otimes \idsupop_{\din}} \cdot \Phi_{\din}
\end{equation*}
denote the \emph{Choi operator} of a superoperator $\chan L: \C^{\din \times \din} \to \C^{\dout \times \dout}$. (The $J$ notation alludes to the \emph{Choi--Jamiołkowski isomorphism} between $\chan L$ and $\choi L$.) In \cref{sec:E-Phi-def} we prove that $\Norm{(\chan L \otimes \chan I_{\daux}) \otimes \bs\psi}_1$ is concentrated around $\Norm{\choi L}_1$ for $\daux = \Omega(\din)$:

\begin{thm}[Informal compilation of \cref{cor:haar-e,thm:levy-app,thm:up-tail-haar}] \label{thminf:compil}
     Let $\chan L: \C^{\din \times \din} \to \C^{\dout \times \dout}$ be a superoperator, let $\daux = \Omega(\din)$, and let $\bs\psi \in \C^{\din \times \din} \otimes \C^{\daux \times \daux}$ be a Haar random state. Then $\E \norm{(\chan L \otimes \chan I_{\daux}) \cdot \bs\psi}_1 = \Theta(\norm{\choi L}_1)$, with high probability $\norm{(\chan L \otimes \chan I_{\daux}) \cdot \bs\psi}_1 \le O(\norm{\choi L}_1)$, and (under a slightly stronger assumption\footnote{Specifically, assuming that $\norm{\chan L}_{\diamond} \le o(\din \norm{\choi L}_1)$, which holds in almost all cases by \cref{thm:diamond-acid-relat}. Or alternatively, assuming $\daux \ge \omega(\din)$ rather than just $\daux = \Omega(\din)$.}) with high probability $\norm{(\chan L \otimes \chan I_{\daux}) \cdot \bs\psi}_1 = \Omega(\norm{\choi L}_1)$, where the asymptotic notation hides universal multiplicative constants.
\end{thm}

For the first part, we show that every superoperator admits a density operator for which the $\rho$-norm is bounded above, for all input states $\rho$, by the ACID norm multiplied by the input dimension and a trace term. The expectation bound follows immediately by averaging this inequality, while the lower bound is obtained by applying it to the spectral decomposition of $\rho$. We then derive the concentration bounds via L\'evy's lemma, proving that the $\rho$-norm is Lipschitz with Lipschitz constant determined by the diamond norm of the superoperator. Finally, we bound the lower-tail event using the tail distribution of the ratio of two $\chi^2$ random variables.

For sufficiently large values of $\daux$, a Haar random state will be close to maximally entangled and therefore \cref{thminf:compil} will follow immediately from the triangle inequality, but the threshold $\daux = \Omega(\din)$ is far too low for such an argument to go through (as we show in Appendix \ref{app:S5-nontrivial}) so \cref{thminf:compil} is nontrivial. The lack of explicit averaging in $\norm{\choi L}_1$ makes it a more convenient quantity to work with than $\norm{\chan L}_{\mrm{avg}}$, and with $\norm{\choi L}_1$ there is no ambiguity regarding the dimension of the ancillary register, so we take $\norm{\choi L}_1$ as our \emph{definition} of the average-case norm:

\begin{dfn}[ACID norm] \label{dfn:acid}
    The \emph{average-case imitation diamond (ACID) norm} of a superoperator $\chan L: \C^{\din \times \din} \to \C^{\dout \times \dout}$ is the quantity
    \begin{equation*}
        \norm{\chan L}_{\lsd} \coloneqq 
        \norm{\choi L}_1 =
        \norm{(\chan L \otimes \chan I_{\din}) \cdot \Phi_{\din}}_1.
    \end{equation*}
\end{dfn}

Montanaro and de Wolf~\cite[Section 5.2]{MdW13} briefly proposed property testing of arbitrary channels in the ACID norm as well, albeit not by this name and without the motivations we give. That the ACID norm is indeed a norm follows from the fact that the trace norm is a norm. The ACID norm is defined similarly to the diamond norm, except that instead of maximizing over all bipartite input states, the input is fixed to the maximally entangled state. However this does \emph{not} mean that optimal channel testing in ACID distance is as simple as optimal state testing in trace distance for the corresponding property of the Choi state, as we will see in \cref{subsec:acid-upper-bounds}.

In Appendix \ref{sec:acid-definitions} we relate the ACID norm to other quantities of interest. We show that it generalizes average-case distances used in property testing of Boolean functions (i.e.\ statistical distance) and in property testing of unitary transformations~\cite{Low2009Avg,MO10,Wang2011Unitary,MdW13,Chen2023Junta,Zhao+23}. This further motivates our definition of the ACID norm, especially since the ACID norm already has the ``right" multiplicative constant for some of these generalizations (unlike $\norm\cdot_{\mrm{avg}}$). Additionally, Montanaro and de Wolf~\cite[Lemma 25]{MdW13} proved that ACID distance is quadratically related to a distance used by Wang~\cite{wang2012} in POVM testing. We also compare the ACID norm to the diamond norm and to the ``average-case induced trace norm" $\E\norm{\chan L(\bs\psi)}_1$ of a superoperator $\chan L$. The latter quantity is also an ``average-case norm", but it seems to lack most of the other motivations that we give for the ACID norm. Finally, we observe that the ACID norm shares certain convenient mathematical properties with the diamond norm, and discuss the prospect of proving a version of the quantum fault-tolerance theorem with the ACID norm in place of the diamond norm.

Besides \cref{thminf:compil}, another sense in which the ACID norm is ``average-case" is that the reduced state on the first register of $\Phi_{\din}$ is maximally mixed, and this is the input to $\chan L$ in the definition of $\choi L$. One can also define variants of the ACID norm with an arbitrary bipartite pure state $\psi$ in place of $\Phi$, i.e.\ the quantity $\norm{(\chan L \otimes \chan I) \cdot \psi}_1$, and each possible reduced state on the first register of $\psi$ can be thought of as specifying a different average-case problem~\cite[top of page 23]{Bos+23}. In this sense \cref{thminf:compil} says that the ACID norm is the ``\emph{average} average-case norm".

There are also practical motivations for channel testing in the ACID norm. A primary application of channel testing is to determine whether a quantum device built in a laboratory or supplied by a third party actually implements the target channel it was allegedly designed to implement. In some applications the device will always take as input half of a maximally entangled state---examples include nonlocal games~\cite{CHSH69}, quantum teleportation~\cite[Sec.\ 6.2.4]{Wil19}, the encoding scheme in superdense coding~\cite[Sec.\ 6.2.3]{Wil19}, entanglement dilution~\cite[Sec.\ 19]{Wil19}, and various protocols for quantum communication over a noisy channel~\cite[Part VI]{Wil19}---and in these cases ACID distance describes the trace distance between the actual and desired states of the bipartite system arising from the faultiness of the quantum device.

Note also that diamond distance is always at least as large as ACID distance, so a property tester with respect to diamond distance implies one for ACID distance.

Finally, Maciejewski, Puchała, and Oszmaniec \cite{MPO23} also study average-case distances between channels. In that paper, the distance  between two channels $\chan N$ and $\chan M$ is defined in terms of the Hilbert-Schmidt distance between the corresponding Choi states and the squared trace of the difference between each channel applied to the maximally mixed state. They prove that this is an average-case measure by showing that this distance is approximately the total variation distance between the output distribution of a Haar random measurement on the output of each channel for Haar random input.

\subsection{Channel certification and tomography in ACID distance}
\label{subsec:acid-upper-bounds}
A \emph{channel tester} is an algorithm that makes queries to a channel $\chan M$ and tries to decide whether $\chan M$ satisfies or is far from some property. We consider three resources which a channel tester may or may not have access to, given the tendency for quantum systems to decohere over time and lose their quantum properties such as entanglement and superposition. First, \emph{ancillae}: does the tester have access to a system of arbitrarily large dimension, or only to a $\din$-dimensional system, barely large enough to apply $\chan M$ to (and which is reset after measuring the output of $\chan M$)? Second, \emph{coherence}: if the tester does have ancillae, can it apply $\chan M$ on different subsystems of an entangled input state and then perform an entangled measurement on the entire output? Or must the tester partition its system as the tensor product of always-unentangled subsystems with only one query to $\chan M$ made within any given subsystem? And third, \emph{adaptivity}: can the input to subsequent queries depend on the output of previous queries, or must all queries be made in parallel?

We now present a series of results on channel testing in ACID distance, which we prove in \cref{sec:acid-certification} and which we summarize and compare to previous work in \cref{tab:Results}. Consider the task of testing identity to a fixed channel $\chan N$. Bădescu, O’Donnell and Wright~\cite{BOW19} proved that for all states $\sigma \in \C^{d \times d}$, there is an algorithm that performs an entangled measurement on $O(d/\eps^2)$ copies of an unknown state $\rho \in \C^{d \times d}$ and decides whether $\rho = \sigma$ or $\norm{\rho - \sigma}_1 \ge \eps$ with success probability at least $2/3$.\footnote{In fact, they proved the stronger statement that given $O(d/\eps^2)$ copies of \emph{two} unknown states $\rho$ and $\sigma$, an entangled measurement can decide whether $\rho = \sigma$ or $\norm{\rho - \sigma}_1 \ge \eps$ with success probability at least $2/3$. Thus \cref{thm:trivial} generalizes to testing equality between two unknown channels given query access to both of them.} Since $\norm{\chan M - \chan N}_{\lsd} = \norm{\choi M - \choi N}_1$ by definition and since $\choi M$ can be constructed using one query to $\chan M$, the following is immediate by applying the above algorithm with $\sigma = \choi N$ and $\rho = \choi M$ (and $d = \din \dout$):

\begin{thm}[Coherent channel certification] \label{thm:trivial}
    For all fixed channels $\chan N: \C^{\din \times \din} \to \C^{\dout \times \dout}$ and $\eps>0$, there is an ancilla-assisted, coherent, non-adaptive algorithm that makes $O(\din \dout / \eps^2)$ queries to a channel $\chan M$, and decides whether $\chan M = \chan N$ or $\norm{\chan M - \chan N}_\lsd \ge \eps$ with success probability at least $2/3$.
\end{thm}

More generally, the query complexity of testing identity to a channel $\chan N$ in this model is at most the sample complexity of testing identity to $\choi N$, which may be $o(\din \dout / \eps^2)$ depending on $\chan N$. However, even \emph{with} coherence, this blackbox reduction to state certification may be far from optimal for channel certification. For example, consider the channel $\chan N: \C^{d \times d} \to \C^{1 \times 1}$ that traces out its entire input, i.e.\ $\chan N(X) = \tr(X)$. Since $\chan N$ is the \emph{only} channel of these dimensions, testing identity to $\chan N$ trivially requires zero queries, whereas its Choi state is maximally mixed and so the blackbox reduction to state testing would require $\Omega(d/\eps^2)$ queries~\cite{OW15}. The key observation is that regardless of the dimensions of a channel $\chan M$, the reduced state on the second subsystem of $\choi M$ is guaranteed to be maximally mixed, a fact which the blackbox reduction to state certification does not take advantage of. Furthermore, channel certification algorithms may query $\chan M$ in ways besides constructing $\choi M$, analogously to how classical property testing algorithms may be allowed to query a function on explicitly chosen inputs rather than random inputs; we leave it as an open problem whether there exists a channel $\chan N$ for which an optimal certification algorithm must query $\chan M$ in ways besides constructing its Choi state.

For all states $\sigma \in \C^{d \times d}$, there is also an algorithm that performs \emph{unentangled}, non-adaptive measurements on $O\Paren{d^{3/2}/\eps^2}$ copies of an unknown state $\rho \in \C^{d \times d}$, and decides whether $\rho = \sigma$ or $\norm{\rho - \sigma}_1 \ge \eps$ with success probability at least $2/3$~\cite{BCL20,CLO22}. Similarly to the above, this implies an $O\Paren{\din^{3/2} \dout^{3/2} \big/ \eps^2}$ upper bound for testing identity to an arbitrary channel in ACID distance in the ancilla-assisted, incoherent, non-adaptive setting. We nontrivially improve on this upper bound by a $\din^{1/2}$ factor, even \emph{without} ancillae:

\begin{restatable}[Ancilla-free channel certification in ACID distance]{thm}{generalub} \label{thminf:ub-general}
    For all fixed channels $\chan N: \C^{\din \times \din} \to \C^{\dout \times \dout}$ and $\eps>0$, there is an ancilla-free, non-adaptive algorithm that makes $\tilde{O}\Paren{\din \dout^{3/2} \big/ \eps^2}$ queries to a channel $\chan M$, and decides whether $\chan M = \chan N$ or $\norm{\chan M - \chan N}_{\lsd} \ge \eps$ with success probability at least $2/3$.
\end{restatable}

Our proof of \cref{thminf:ub-general} goes through an analogous statement where the distance between channels $\chan M$ and $\chan N$ is measured by the $\ell_2$ distance between their Choi states, i.e.\ the quantity $\norm{\choi M - \choi N}_2$. This quantity is related to the ACID distance between $\chan M$ and $\chan N$ by Cauchy-Schwarz, and so \cref{thminf:ub-general} follows as a corollary. Fawzi et al.~\cite{FFGO23} related the $\ell_2$ distance between Choi states to the diamond distance between the corresponding channels, so we also obtain an analogue of \cref{thminf:ub-general} with respect to the diamond norm:

\begin{restatable}[Ancilla-free channel certification in diamond distance]{thm}{generalubdiamond}
\label{thm:ub-general-diamond}
    For all fixed channels $\chan N: \C^{\din \times \din} \to \C^{\dout \times \dout}$ and $\eps>0$, there is an ancilla-free, non-adaptive algorithm that makes $\tilde{O}\Paren{\din^{2} \dout^{3/2} \big/ \eps^2}$ queries to a channel $\chan M$, and decides whether $\chan M = \chan N$ or $\norm{\chan M - \chan N}_\diamond \ge \eps$ with success probability at least $2/3$.
\end{restatable}

\cref{thm:ub-general-diamond} generalizes a result of Fawzi et al.~\cite{FFGO23}, who proved the same upper bound without log factors in the case where $\chan N$ is the completely depolarizing channel. We also remove the log factors from \cref{thminf:ub-general} when $\chan N$ is the completely depolarizing channel.

We also give \emph{dimension-independent} upper bounds for testing identity to certain channels:

\begin{restatable}[Erasure, unitary, and pure state replacement channel certification]{thm}{erasureub}
\label{thminf:erasure-etc-upper-bounds}
    Let $\chan N$ be any of the following types of channels:
    \begin{itemize}[leftmargin=*]
        \item an \emph{erasure channel}, i.e.\ $\chan N(X \otimes Y) = X \tr(Y)$ for all $X \in \C^{\dout \times \dout}, Y \in \C^{\din/\dout \times \din/\dout}$, with the definition extended to arbitrary inputs by linearity;
        \item a \emph{unitary channel}, i.e.\ $\chan N(X) = U X \adj U$ for all $X \in \C^{d \times d}$, for some unitary $U \in \C^{d \times d}$ (independent of $X$);
        \item a \emph{pure state replacement channel}, i.e.\ $\chan N(X) = \tr(X) \psi$ for all $X \in \C^{\din \times \din}$, for some pure state $\psi \in \C^{\dout \times \dout}$ (independent of $X$).
    \end{itemize}
    Then there is an ancilla-free, non-adaptive algorithm that makes $O(1/\eps^2)$ queries to a channel $\chan M$, accepts with probability $1$ if $\chan M = \chan N$, and accepts with probability at most $1/2$ if $\norm{\chan M - \chan N}_{\lsd} \ge \eps$.
\end{restatable}

We prove this theorem by first showing that whenever a fixed channel acts as the identity on a subspace of the input space, the complexity of testing it is essentially no greater than the complexity of testing its restriction to the complementary subspace. We then apply this reduction to each of the channel classes above, reducing their certification problems to certification of simpler channels.

For comparison, recall that channel certification in diamond distance requires $\Omega(\sqrt\din/\eps)$ queries for erasure channels (\cref{thminf:diamond-LB}), $d/\eps^{\Theta(1)}$ queries for unitary channels~\cite{FFGO23}, and $\Omega(\din/\eps)$ queries for pure state replacement channels (\cref{rmk:ReplacementLB}). Along the way to proving \cref{thminf:erasure-etc-upper-bounds}, we also show that for \emph{every} channel $\chan N$, testing identity to $\chan I \otimes \chan N$ in ACID distance efficiently reduces to testing identity to $\chan N$ in ACID distance (\cref{thm:N-red-IN}); we consider this observation to be of independent interest as progress toward instance optimality (see \cref{subsec:discussion}). We also remark that Montanaro and de Wolf~\cite[Section 5.2.1]{MdW13} gave an $O(1/\eps^2)$ bound upper bound for testing whether a channel $\chan M$ satisfies the \emph{property} of being unitary or is far from that property in ACID distance, by a blackbox reduction to purity testing on $\choi M$.

The case of \cref{thminf:erasure-etc-upper-bounds} where $\chan N$ is the identity channel on $\C^{d \times d}$ is particularly interesting. By the Fuchs--van de Graaf inequalities, the ACID distance $\frac12 \norm{\chan M - \chan N}_{\lsd} = \frac12 \norm{\choi M - \Phi_d}_1$ is quadratically related to the \emph{entanglement fidelity}~\cite[Definition 9.5.1]{Wil19} $\tr\Paren{\choi M \Phi_d}$ between $\chan M$ and the identity channel with respect to the maximally entangled state. Fawzi et al.~\cite[Lemma A.1]{FFGO23} proved that if $d$ is large, then $\tr\Paren{\choi M \Phi_d}$ is a close approximation of $\E\Brac{\tr(\chan M(\bs\psi) \bs\psi)}$ (where $\bs\psi$ is Haar random), a quantity which is a standard measure for quantifying errors in physical implementations of quantum gates~\cite[Eq.\ 1]{KLDF16}.

\cref{thminf:erasure-etc-upper-bounds} does not generalize to arbitrary channels $\chan N$ however. For example, let $\chan N$ be the channel that replaces its input with a known state $\sigma$ (i.e.\ $\chan N(X) = \tr(X) \sigma$), and suppose that $\chan M$ is promised to replace its input with some unknown state $\rho$ (i.e.\ $\chan M(X) = \tr(X) \rho$). It is straightforward to verify that $\norm{\chan M - \chan N}_{\lsd} = \norm{\rho - \sigma}_1$, and query access to $\chan M$ is equivalent to sample access to $\rho$, so testing identity to $\chan N$ in ACID distance is no easier than testing identity to $\sigma$ in trace distance. If $\sigma$ is the maximally mixed state for example, i.e.\ if $\chan N$ is the completely depolarizing channel, then this requires $\Omega\Paren{\dout^{3/2} / \eps^2}$ queries in the ancilla-free, adaptive model~\cite[Theorem 6.1]{CLHL22}. This is why we specifically considered \emph{pure} state replacement channels in \cref{thminf:erasure-etc-upper-bounds}.

The above discussion shows that a dependence on the \emph{output} dimension is sometimes unavoidable. We also prove that a dependence on the \emph{input} dimension is sometimes unavoidable, again in the case of the completely depolarizing channel, and even for $\dout$ as small as $2$:

\begin{restatable}[Lower bound for the completely depolarizing channel]{thm}{lbdepol} \label{thm:lb-depol}
    Let $\chan N: \C^{\din \times \din} \to \C^{\dout \times \dout}$ be the completely depolarizing channel, i.e.\ $\chan N(X) = \tr(X) I/\dout$, and assume for simplicity that $\din$ and $\dout$ are even. Then every ancilla-free, non-adaptive channel tester requires $\Omega(\din/\eps^2)$ queries to decide whether $\chan M = \chan N$ or $\norm{\chan M - \chan N}_{\lsd} \ge \eps$ with success probability at least $2/3$.
\end{restatable}

We prove this result by first showing that the total variation distance between two product distributions can be bounded in terms of the total variation distances between their constituent distributions. We then construct a hard instance for certification of the completely depolarizing channel by perturbing it in distinct directions on two complementary subspaces of the input. This reduces the certification problem to distinguishing between two carefully chosen product distributions, to which the previous result applies. Unlike the cases considered in \cref{thminf:erasure-etc-upper-bounds}, this approach does not yield a dimension-independent upper bound, since there is no analogous reduction to a simpler channel for the completely depolarizing channel.

The $\Omega(\din/\eps^2)$ lower bound from \cref{thm:lb-depol} matches the dependence on $\din$ and $\eps$ from \cref{thminf:ub-general} in the same query model, and along with the above discussion implies an $\Omega\Paren{\Paren{\din + \dout^{3/2}} \big/ \eps^2}$ lower bound for testing identity to the completely depolarizing channel in this model. We conjecture that this lower bound can be improved to $\Omega\Paren{\din \dout^{3/2} \big/ \eps^2}$, which would match our upper bound.

We also briefly consider a nonstandard query model, where it turns out that channel certification \emph{can} always be done with complexity independent of the input dimension. King, Wan and McClean~\cite{KWM24} proposed a model of quantum state testing with sample access to both $\rho$ and $\rho^\top$, and gave several examples~\cite[Appendix D]{KWM24} where this may be a physically realistic assumption. Analogously, for a channel $\chan M$ we define $\dual M(X) \coloneqq M\Paren{X^\top}^\top$. The fact that $\dual M$ is a channel is most easily seen by considering its Kraus decomposition (see \cref{eq:kraus}), which also illustrates that $\dual M$ is the element-wise complex conjugate of $\chan M$. For example, if $\chan M$ is defined by evolving a real-valued Hamiltonian forward in time, then $\dual M$ is defined by evolving that same Hamiltonian backward in time. If $\chan M$ is implemented by a quantum circuit over the gate set $\{H, T, \mrm{Toffoli}\}$, then $\dual M$ can be implemented by substituting $\adj T$ for $T$ throughout that circuit.\footnote{However, if our motivation is to test whether an alleged circuit implementation of $\chan N$ is accurate, then there is no guarantee that faulty implementations of $\chan N$ and $\dual N$ would be $\chan M$ and $\dual M$ respectively for the \emph{same} channel $\chan M$.} We prove the following:

\begin{restatable}[Channel certification using $\chan M$ and $\dual M$]{thm}{ubdual} \label{thm:ub-dual}
    For all fixed channels $\chan N: \C^{\din \times \din} \to \C^{\dout \times \dout}$ and $\eps > 0$, there is an ancilla-assisted, coherent, non-adaptive algorithm that makes $O\Paren{\dout^4 / \eps^4}$ queries to channels $\chan M$ and $\dual M$, and decides whether $\chan M = \chan N$ or $\norm{\chan M - \chan N}_{\lsd} \ge \eps$ with success probability at least $2/3$.
\end{restatable}

Finally we consider the complexity of channel tomography in ACID distance, as a benchmark against which to compare our results about channel testing (as testing trivially reduces to tomography). We prove the following by a blackbox reduction to state tomography on $\choi M$, followed by post-processing to ensure that the output is a channel:

\begin{restatable}[Upper bound for coherent channel tomography]{thm}{UBtom} \label{thm:UBCoherentTom}
    There is an ancilla-assisted, coherent, non-adaptive algorithm that makes $O\Paren{\din^2 \dout^2 / \eps^2}$ queries to a channel $\chan M: \C^{\din \times \din} \to \C^{\dout \times \dout}$, and with probability at least $2/3$ outputs the description of a channel $\chan N$ such that $\norm{\chan M - \chan N}_{\lsd} \le \eps$.
\end{restatable}

We also prove a nontrivial matching lower bound for fixed $\eps$, even for adaptive algorithms:

\begin{restatable}[Lower bound for coherent channel tomography]{thm}{CoherentTomAcid}\label{thm:CoherentTomAcid}
   For all $\din \ge 1$ and $\dout \ge 4$, every ancilla-assisted, coherent, adaptive algorithm requires $\Omega\Paren{\din^2 \dout^2 / \log(\din \dout)}$ queries to a channel $\chan M: \C^{\din \times \din} \to \C^{\dout \times \dout}$ to output the description of a channel $\chan N$ such that $\norm{\chan M - \chan N}_{\lsd} < 1/16$ with probability at least $2/3$.
\end{restatable}

Since the ACID norm is trivially at most the diamond norm, results of Oufkir~\cite{O23} resolve the complexity of incoherent, non-adaptive channel tomography in both ACID and diamond distances:

\begin{thm}[{Incoherent channel tomography~\cite[Theorems 3.3 and 2.1\protect\footnote{The lower bound is stated in terms of diamond distance, but inspection of the proof reveals that it holds for ACID distance.}]{O23}}] \label{thm:DiamondTomTight}
    There is an ancilla-free, non-adaptive algorithm that makes $\tilde O\Paren{\din^3 \dout^3 \big/ \eps^2}$ queries to a channel $\chan M: \C^{\din \times \din} \to \C^{\dout \times \dout}$, and outputs the description of a channel $\chan N$ such that $\norm{\chan M - \chan N}_\diamond \le \eps$ with probability at least $2/3$. Furthermore $\Omega\Paren{\din^3 \dout^3 \big/ \eps^2}$ queries are necessary for this task when $\dout \ge 4$, even using ancillae (but not coherence or adaptivity) and with the ACID norm in place of the diamond norm.
\end{thm}

Similarly, the following upper bound of Haah, Kothari, O’Donnell and Tang~\cite[Theorem 1.1]{HKOT23} and lower bound of Zhao, Lewis, Kannan, Quek, Huang and Caro~\cite[$G=d^2$ case of Theorem 4]{Zhao+23}\footnote{In \cref{subsec:avg-unitary} we explain why Zhao et al.'s distance is equivalent to ACID distance.} resolve the complexity of \emph{unitary} tomography in both ACID and diamond distances:

\begin{thm}[{Unitary tomography~\cite{HKOT23,Zhao+23}}]
    There is an ancilla-free, adaptive algorithm that makes $O(d^2/\eps)$ queries to a unitary channel $\chan M: \C^{d \times d} \to \C^{d \times d}$, and outputs the description of a unitary channel $\chan N$ such that $\norm{\chan M - \chan N}_\diamond \le \eps$ with probability at least $2/3$. Furthermore $\Omega(d^2/\eps)$ queries are necessary for this task, even using ancillae and coherence, and even with the ACID norm in place of the diamond norm.
\end{thm}

\begin{table}
    \centering
    \begin{tabular}{|Sc|Sc||Sc|Sc|Sc|}
    \cline{3-5}
    \multicolumn2{c||}{}
    & Ancilla-free
    & \makecell{Ancilla-assisted, \\ Incoherent}
    & \makecell{Coherent}
    \\ \hline\hline
    \multirow2*{\makecell{Generic \\ channel \\ certification}}
    & $\diamond$
    & \makecell{$\bs{\tilde O \Paren{\din^2 \dout^{3/2} / \eps^2}}$ \\ \cref{thm:ub-general-diamond}}
    & 
    &  \makecell{$O\Paren{\din^2 \dout / \eps^2}$ \\ \cite[Thm.\ 1.4]{BOW19} \\ \cite[Lem.\ C.1]{FFGO23} \\[+1.5mm] *$\bs{\Omega\Paren{\din^{1/2} / \eps}}$ \\ \cref{thminf:diamond-LB}}
    \\ \cline{2-5}
    & $J$
    & \makecell{$\bs{\tilde O \Paren{\din \dout^{3/2} / \eps^2}}$ \\ \cref{thminf:ub-general}}
    & 
    & \makecell{$O \Paren{\din \dout / \eps^2}$ \\ \cref{thm:trivial} \\[+1.5mm] $\bs{O\Paren{\dout^4 / \eps^4}}$ with $\chan M, \dual M$ \\ \cref{thm:ub-dual}}
    \\ \hline
    \multirow2*{\makecell{Completely \\ depolarizing \\ channel}}
    & $\diamond$
    & \makecell{$O \Paren{\din^2 \dout^{3/2} / \eps^2}$ \\ \cite[Thm.\ 4.4]{FFGO23}}
    & \makecell{$\tilde\Omega\Paren{\din^2 \dout^{3/2} / \eps^2}$ \\ \cite[Thm.\ 4.5]{FFGO23}}
    & 
    \\ \cline{2-5}
    & $J$
    & \makecell{$\bs{O \Paren{\din \dout^{3/2} / \eps^2}}$ \\ \cref{thm:ub-depol} \\[+1.5mm] $\bs{\Omega\Paren{\din/\eps^2}}$ \\ \cref{thm:lb-depol} \\[+1.5mm] *$\Omega\Paren{\dout^{3/2} / \eps^2}$ \\ \cite[Thm.\ 6.1]{CLHL22}}
    & 
    & 
    \\ \hline
    \multirow2*{\makecell{Unitary \\ channel}}
    & $\diamond$
    & \makecell{$O(d/\eps^4)$ \\ \cite[Thm.\ 3.1]{FFGO23}}
    & \makecell{*$\Omega\Paren{d/\eps^2}$ \\ \cite[Thm.\ 3.1]{FFGO23}}
    &
    \\ \cline{2-5}
    & $J$
    & \makecell{$\bs{O(1/\eps^2)}$ \\ \cref{thminf:erasure-etc-upper-bounds}}
    & 
    & 
    \\ \hline
    \multirow2*{\makecell{Pure state \\ replacement \\ channel \\ }}
    & $\diamond$
    & 
    & 
    & \makecell{*$\bs{\Omega\Paren{\din/\eps}}$ \\ \cref{rmk:ReplacementLB}}
    \\ \cline{2-5}
    & $J$
    & \makecell{$\bs{O(1/\eps^2)}$ \\ \cref{thminf:erasure-etc-upper-bounds}}
    & 
    & 
    \\ \hline
    \multirow2*{\makecell{\\[-1.5mm] Erasure \\ channel}}
    & $\diamond$
    & 
    & 
    & 
    \\ \cline{2-5}
    & $J$
    & \makecell{$\bs{O(1/\eps^2)}$ \\ \cref{thminf:erasure-etc-upper-bounds}}
    & 
    & 
    \\ \hline
    \multirow2*{Tomography}
    & $\diamond$
    & \makecell{$\tilde O\Paren{\din^3 \dout^3 \big/ \eps^2}$ \\ \cref{thm:DiamondTomTight}}
    & 
    & 
    \\ \cline{2-5}
    & $J$
    & 
    & \makecell{$\Omega\Paren{\din^3 \dout^3 \big/ \eps^2}$ \\ \cref{thm:DiamondTomTight}}
    & \makecell{$\bs{O\Paren{\din^2 \dout^2 / \eps^2}}$ \\ \cref{thm:UBCoherentTom} \\[+1.5mm] *$\bs{\tilde\Omega\Paren{\din^2 \dout^2}}$ \\ \cref{thm:CoherentTomAcid}}
    \\ \hline
    \end{tabular}
    \caption{Query complexity of channel certification and tomography in both diamond ($\diamond$) and ACID ($J$) distances. A star denotes adaptivity. Nontrivial results from this paper (i.e.\ excluding direct reductions to state certification and state tomography) are in bold font.}
    \label{tab:Results}
\end{table}

\subsection{Open problems}
\label{subsec:discussion}

\paragraph*{Instance optimality} The sample complexity of testing identity to a fixed state $\sigma \in \C^{d \times d}$ using unentangled measurements is roughly $d^{3/2}/\eps^2$ times the (square) fidelity of $\sigma$ with the maximally mixed state~\cite{CLO22,CLHL22}. Analogously, what is the query complexity of testing identity to a fixed channel $\chan N$ in any of the query models that we have discussed? Similarly, what is the power of adaptivity in this setting? One may approach this question by trying to close some of the gaps between the upper and lower bounds in \cref{tab:Results}. What if we consider \emph{tolerant testing}, where the goal is to decide whether $\norm{\chan M - \chan N}_{\lsd} \le \delta$ or $\norm{\chan M - \chan N}_{\lsd} \ge \eps$? What if we also require our protocols to be \emph{computationally} efficient, for example by sampling states from a locally scrambled ensemble~\cite[Definition 1]{Zhao+23} instead of the Haar measure?

\paragraph*{Testing and tomography of channels with bounded gate complexity}

Zhao et al.~\cite[Theorem 4]{Zhao+23} proved that $\tilde O\Paren{G/\eps \cdot \min\Paren{1/\eps, \sqrt{d}}}$ queries suffice and $\Omega(G/\eps)$ queries are necessary to learn in ACID distance a $d$-dimensional unitary channel comprised of $G$ two-qubit gates. Does a similar statement hold for arbitrary channels? What about for testing rather than tomography?

\paragraph*{Junta testing and tomography}
A \emph{$k$-junta} is a channel from $\Paren{\C^{2 \times 2}}^{\otimes n}$ to $\Paren{\C^{2 \times 2}}^{\otimes n}$ that acts nontrivially on at most $k$ qubits. Chen, Nadimpalli and Yuen~\cite{Chen2023Junta} proved that $\tilde\Theta(\sqrt k)$ queries to a unitary channel are necessary and sufficient to test whether it is a $k$-junta or far from all $k$-juntas, and that $\tilde\Theta(4^k)$ queries are necessary and sufficient to learn a unitary $k$-junta. (We have suppressed the dependence on $\eps$ for simplicity.) In \cref{subsec:avg-unitary} we show that their distance is proportional to ACID distance, so it is natural to ask whether their results generalize to the case where the blackbox channel is not necessarily unitary, with distance measured in the ACID norm. Bao and Yao~\cite{Bao23Junta} proved similar results (except with only an $\tilde O(k)$ upper bound for testing) when the blackbox channel is not necessarily unitary, but they measured distance between channels by the $\ell_2$ distance between their Choi states, a quantity which is only loosely related to the ACID norm via Cauchy-Schwarz and the fact that the 2-norm is at most the 1-norm.

\paragraph*{Fault-tolerance}
The \emph{quantum fault-tolerance theorem} (also called the \emph{threshold theorem}) says that if each gate in a quantum circuit introduces limited error, then under certain physically realistic assumptions it is possible to design quantum circuits that achieve low error overall~\cite[Section 10.6]{NC10}. Here, errors in individual gates and in the overall circuit are measured in the diamond norm. Does the same statement hold with respect to the ACID norm? For individual gates that act on a constant number of qubits each, the ACID and diamond norms are equivalent ways of measuring error up to a constant factor (see \cref{thm:diamond-acid-relat} for precise bounds), but this constant factor can still make a difference in practice. Furthermore, scaling a general-purpose quantum computer to millions of physical qubits will require partitioning it into modules of tens or hundreds of qubits each where good control has been achieved~\cite{Ach+24}, and one may wish to verify the accuracy of the overall quantum computer by certifying each module individually and then applying a version of the fault-tolerance theorem where the ``gates" are these large modules. We discuss this question further in \cref{subsec:fault-tolerance}.

\section{Preliminaries}
\label{sec:prelims}

We write $\pr\cdot$ to denote probability, $\E[\cdot]$ to denote expected value, $\tr(\cdot)$ to denote trace, and $[n]$ to denote the set $\{1, 2, \dotsc, n\}$ for $n\in \N$. Logarithms in this paper are base 2. We write random variables in boldface font. A statement about a random variable $\bs X$ holds \emph{pointwise} if it holds for all fixed values in the support of $\bs X$.

\subsection{Quantum states and transformations}

We denote the identity matrix in $\C^{d \times d}$ by $I_d$, or just $I$ when $d$ is implicit. The \emph{maximally entangled state} in $\C^d \otimes \C^d$ is the state $\ket{\Phi_d} \coloneqq \frac1{\sqrt d} \sum_{i=1}^d \ket{ii}$, or just $\ket\Phi$ when $d$ is implicit. We also write
\begin{equation*}
    \Phi = \Phi_d \coloneqq \kb{\Phi_d} = \frac1d \sum_{i,j=1}^d \ketbra{ii}{jj}.
\end{equation*}
For a matrix $A$, let $A^*$ denote its element-wise complex conjugate. It is well known that for all matrices $A \in \C^{m \times n}$,
\begin{equation} \label{eq:Phi-trans}
    \sqrt{n} (A \otimes I_n) \ket{\Phi_n} = \sqrt{m} \Paren{I_m \otimes A^\top} \ket{\Phi_m}.
\end{equation}

A matrix is \emph{positive semidefinite} (PSD) if it is Hermitian and its eigenvalues are all nonnegative. A \emph{density matrix} is a PSD matrix whose trace is 1. We denote the set of density matrices in $\C^{d \times d}$ by $\dens d$. A \emph{positive operator-valued measure} (POVM) is a tuple of PSD matrices summing to the identity; if $\rho$ is a density matrix and $(P_1, \dotsc, P_n)$ is a POVM, then $(\tr(P_1 \rho), \dotsc, \tr(P_n \rho))$ is a probability distribution that can physically be sampled from given a copy of $\rho$. A \emph{projection-valued measure} (PVM) is a POVM whose elements are projections onto orthogonal subspaces.

For a pure state $\ket\psi$ we write $\psi = \kb\psi$, for example to denote the rank-1 density matrix or PVM element corresponding to $\ket\psi$. Often we will not need to refer to $\ket\psi$ at all except as part of $\kb\psi$, and in these cases we may \emph{define} $\psi$ to be a pure state, with the lack of a ket symbol indicating that $\psi$ is a rank-1 density matrix rather than a column vector. In particular, a ``Haar random state $\bs\psi$" means $\kb{\bs\psi}$ for a Haar random state $\ket{\bs\psi}$. We will often implicitly use the fact that if $\bs\psi \in \dens d$ is Haar random then $\E[\bs\psi] = I/d$.

A \emph{superoperator} is a linear transformation from $\C^{\din \times \din}$ to $\C^{\dout \times \dout}$. We denote the set of superoperators of these dimensions by $\super \din \dout$, and also define
\begin{equation*}
    \spr \din \coloneqq \bigcup_{\mathclap{\dout \in \N}} \super \din \dout.
\end{equation*}
We denote a superoperator $\chan L$ applied to an
input $X$ by any of $\chan L(X)$ or $\chan L \cdot X$ or $\chan L X$.\footnote{No relation to the Pauli $X$ matrix.} If we define a superoperator $\chan L \in \spr d$ by its action on an unspecified matrix $X$, then $X$ implicitly ranges over all matrices in $\C^{d \times d}$. We write superoperators in mathcal font.

We denote the identity superoperator in $\super d d$ by $\chan I_d$, or just $\chan I$ when $d$ is implicit. The \emph{Choi operator} of a superoperator $\chan L \in \spr d$ is the matrix
\begin{equation*}
    \choi L
    \coloneqq \Paren{\chan L \otimes \chan I_d} \Phi_d
    = \frac1d \sum_{i,j=1}^d \chan L\biggl(\ketbra{i}j\biggr) \otimes \ketbra{i}j.
\end{equation*}

A \emph{channel} is a superoperator $\chan N$ that is completely positive and trace-preserving. \emph{Completely positive} means that $\chan N \otimes \chan I_d$ maps every PSD input to a PSD output for all $d$, or equivalently that $\choi N$ is PSD~\cite[Theorem 2.22]{Wat18}. \emph{Trace-preserving} means that $\tr(\chan N(X)) = \tr(X)$ for all $X$. We denote the set of channels from $\C^{\din \times \din}$ to $\C^{\dout \times \dout}$ by $\cptp \din \dout$. The Choi operator of a channel is called a \emph{Choi state}.

We write $\tr_d \in \cptp d 1$ to denote the channel that traces out its entire $d$-dimensional input, i.e.\ $\tr_d(X) = \tr(X)$. (This is the exception to our criterion that superoperators are written in mathcal font.) Thus $\chan I \otimes \tr_d$ denotes a partial trace.

A superoperator is called \emph{Hermitian-preserving} if it maps every Hermitian input to a Hermitian output. For example, a channel is Hermitian-preserving, as is the difference between two channels. Every Hermitian-preserving superoperator $\chan L$ can be expressed as
\begin{equation} \label{eq:kraus}
    \chan L(X) = \sum_j \pm A_j X \adj A_j
\end{equation}
for some matrices $A_j$~\cite[{Theorems 2.22 and 2.25\protect\footnote{Specifically, Theorem 2.25 says that every Hermitian-preserving superoperator can be expressed as the difference between two completely positive superoperators, and Theorem 2.22 says that every completely positive superoperator can be expressed as in \cref{eq:kraus} without the plus-or-minus signs.}}]{Wat18}.

A \emph{register} is a finite-dimensional complex Hilbert space. We write $\reg{AB}$ to denote the tensor product of registers $\reg A$ and $\reg B$, and $\dens{\reg A}$ to denote the set of density matrices in a register $\reg A$. We also write $\dens{d_1 \otimes d_2}$ to denote the set of density matrices in $\C^{d_1 \times d_1} \otimes \C^{d_2 \times d_2}$, and similarly for $\super \cdot \cdot$ and $\cptp \cdot \cdot$.

\begin{lem} \label{lem:tr-choi-prod}
    Define superoperators $\chan K, \chan L \in \spr d$ by $\chan K(X) = A X \adj A$ and $\chan L(X) = B X \adj B$ for some matrices $A,B$. Then $\tr(\choi K \choi L) = \Mag{\tr\Paren{\adj A B}}^2 / d^2$.
\end{lem}
\begin{proof}
    We have
    \begin{align*}
        \tr(\choi K \choi L)
        &= \tr\Paren{(A \otimes I) \Phi \Paren{\adj A B \otimes I} \Phi \Paren{\adj B \otimes I}} \\
        &= \Mag{\bra\Phi \Paren{\adj A B \otimes I} \ket\Phi}^2 \\
        &= \Mag{\frac1d \sum_{j,k=1}^d \bra{jj} \Paren{\adj A B \otimes I} \ket{kk}}^2 \\
        &= \Mag{\frac1d \sum_{j=1}^d \bra{j} \adj A B \ket{j}}^2 \\
        &= \frac1{d^2} \Mag{\tr(\adj A B)}^2. \qedhere
    \end{align*}
\end{proof}

For $d \in \N$ let $\swap d = \sum_{i,j=1}^d \ketbra{ij}{ji}$. This matrix is Hermitian and unitary, so its eigenvalues are all $\pm1$. The +1 eigenspace of $\swap d$ is known as the \emph{symmetric subspace} of $\C^d \otimes \C^d$. Let $\sym d$ denote the projection onto this subspace.

It follows immediately that $\swap d = 2 \sym d - I$. Furthermore $\sym d = d(d+1)/2 \cdot \E\Brac{\bs\psi^{\otimes 2}}$ for Haar random $\bs\psi \in \dens d$~\cite[Proposition 6]{Har13}, and combining these equations yields
\begin{equation} \label{lem:sym}
    \E\Brac{\bs\psi^{\otimes 2}}
    = \frac2{d(d+1)} \sym d
    = \frac1{d(d+1)} \Paren{I + \swap d}
    = \frac1{d(d+1)} \Paren{I + \sum_{i,j=1}^d \ketbra{ij}{ji}},
\end{equation}
one consequence of which is the fact~\cite[Lemma 4.2.4]{HP00} that for all $i,j \in [d]$,
\begin{equation} \label{lem:haar-moment}
    \E\Brac{\Mag{\braket{i}{\bs\psi}}^2 \cdot \Mag{\braket{j}{\bs\psi}}^2}
    = \begin{cases}
        2/d(d+1) &\text{if } i=j,\\
        1/d(d+1) &\text{if } i\neq j.
    \end{cases}
\end{equation}

\subsection{Matrix norms and fidelity}

For $1 \le p \le \infty$, the \emph{Schatten $p$-norm} of a matrix $A$ is the $p$-norm of the vector of singular values of $A$, and is denoted $\norm A_p$. In particular, we use that $\norm{A}_\infty$ equals the largest singular value of $A$ and that $\norm{A}_2^2 = \tr\Paren{A \adj A}$. The quantity $\norm{A}_1$ is called the \emph{trace norm} of $A$, and has the equivalent definition~\cite[Eq.\ 1.173]{Wat18}
\begin{equation} \label{eq:trace-norm-var-def}
    \norm{A}_1 = \max_{\mathclap{\norm{B}_\infty = 1}} |\tr(AB)|,
\end{equation}
with the maximum achieved by a Hermitian matrix $B$ when $A$ is Hermitian. We use the fact~\cite[Eq.\ 1.186]{Wat18} that for all pure states $\ket\psi$ and $\ket\phi$,
\begin{equation} \label{eq:pure-state-td}
    \norm{\psi - \phi}_1
    = 2\sqrt{1 - \tr(\psi \phi)}.
\end{equation}
It follows from \cref{eq:pure-state-td} that
\begin{equation} \label{eq:td2d}
    \norm{\psi - \phi}_1 \le 2 \norm{\ket\psi - \ket\phi}_2,
\end{equation}
since $1 - \tr(\psi\phi) = (1 + \Mag{\braket\psi\phi}) (1 - \Mag{\braket\psi\phi}) \le 2(1 - \mrm{Re}[\braket\psi\phi]) = \norm{\ket\psi - \ket\phi}_2^2$.

The (square) \emph{fidelity} of density matrices $\rho$ and $\sigma$ is the quantity $\fid \rho \sigma \coloneqq \norm{\sqrt\rho \sqrt\sigma}_1^2$. In particular, if $\rho \in \dens d$ is an arbitrary density matrix then $\fid {\rho} {I/d} = \tr(\sqrt\rho)^2 / d$, and if furthermore $\psi \in \dens d$ is a pure state then $\fid \rho \psi = \tr(\rho \psi)$. We also use the following half of the Fuchs–-van de Graaf inequalities: for all density matrices $\rho$ and $\sigma$,
\begin{equation} \label{eq:fvdg}
    \frac12 \norm{\rho - \sigma}_1 \le \sqrt{1 - \fid \rho \sigma}.
\end{equation}

Finally, recall from \cref{subsec:diamond-hardness,subsec:main-results} that the \emph{induced trace norm}, \emph{diamond norm}, and \emph{ACID norm} of a superoperator $\chan L \in \spr d$ are respectively defined by
\begin{align*}
    \norm{\chan L}_1
    &\coloneqq \max_{\norm{X}_1 = 1} \norm{\chan L(X)}_1, \\
    \norm{\chan L}_\diamond
    &\coloneqq \norm{\chan L \otimes \chan I_d}_1
    = \max_{\norm{X}_1 = 1} \norm{(\chan L \otimes \chan I_d) X}_1, \\
    \norm{\chan L}_{\lsd}
    &\coloneqq \Norm{\choi L}_1
    = \norm{(\chan L \otimes \chan I_d) \Phi_d}_1.
\end{align*}
When $\chan L$ is Hermitian-preserving, the maxima in the definitions of the induced trace norm and diamond norm are achieved when $X$ is Hermitian, and therefore (by convexity) when $X$ is a pure state. It is well known that $\norm{\chan L}_1 \le \norm{\chan L}_\diamond$ for all superoperators $\chan L$, and also that $\norm{\chan N}_1 = \norm{\chan N}_\diamond = 1$ for all channels $\chan N$~\cite[Corollary 3.40]{Wat18}.

\subsection{Query models for channel testers} \label{sec:query-models}

We now formally define the models of channel testers that we consider. The following definitions describe what we call \emph{deterministic} channel testers; a \emph{randomized} channel tester is a convex combination of deterministic ones. For tomography algorithms we replace the set $\{\mrm{Accept}, \mrm{Reject}\}$ with an arbitrarily large finite set of descriptions of channels in the following definitions.

\begin{dfn}[Ancilla-free, non-adaptive channel tester] \label{dfn:tester000}
    A (deterministic) \emph{ancilla-free, non-adaptive channel tester} making $n$ queries to a channel $\chan M \in \cptp \din \dout$ consists of the following:
    \begin{itemize}
        \item pure states $\psi_1, \dotsc, \psi_n \in \dens{\din}$;
        \item POVMs $P^{(1)}, \dotsc, P^{(n)}$ on $\C^{\dout \times \dout}$, where the elements of each $P^{(j)}$ are denoted $P^{(j)}_1, \dotsc, P^{(j)}_{m_j}$;
        \item a function $f: [m_1] \times \dotsb \times [m_n] \to \{\mrm{Accept}, \mrm{Reject}\}$.
    \end{itemize}
    The tester performs $P^{(j)}$ on $\chan M(\psi_j)$ for all $j \in [n]$, yielding a string $\bs x$ of measurement outcomes, and then outputs $f(\bs x)$.
\end{dfn}

The requirement that the input states $\psi_j$ be pure is without loss of generality, because a randomized channel tester can simulate the action of $\chan M$ on a mixed state $\rho$ by writing $\rho$ as a convex combination of pure states.

\begin{dfn}[Ancilla-assisted, incoherent, non-adaptive channel tester] \label{dfn:tester100}
    A (deterministic) \emph{ancilla-assisted, incoherent, non-adaptive channel tester} making $n$ queries to a channel $\chan M \in \cptp \din \dout$ consists of the following:
    \begin{itemize}
        \item pure states $\psi_1, \dotsc, \psi_n \in \dens{\din \otimes \daux}$, for some $\daux \in \N$;
        \item POVMs $P^{(1)}, \dotsc, P^{(n)}$ on $\C^{\dout \times \dout} \otimes \C^{\daux \times \daux}$, where the elements of each $P^{(j)}$ are denoted $P^{(j)}_1, \dotsc, P^{(j)}_{m_j}$;
        \item a function $f: [m_1] \times \dotsb \times [m_n] \to \{\mrm{Accept}, \mrm{Reject}\}$.
    \end{itemize}
    The tester performs $P^{(j)}$ on $(\chan M \otimes \chan I)\psi_j$ for all $j \in [n]$, yielding a string $\bs x$ of measurement outcomes, and then outputs $f(\bs x)$.
\end{dfn}

It is without loss of generality that all $n$ of the unentangled subsystems have the same dimension $\din \daux$, because operations on a larger system can always simulate operations on a smaller one.

\begin{dfn}[Ancilla-assisted, coherent, non-adaptive channel tester] \label{dfn:tester110}
    A (deterministic) \emph{ancilla-assisted, coherent, non-adaptive channel tester} making $n$ queries to a channel $\chan M \in \cptp \din \dout$ consists of the following:
    \begin{itemize}
        \item a pure state $\psi \in \Dens{\din^{\otimes n} \otimes \daux}$, for some $\daux \in \N$;
        \item a two-outcome POVM $P = (P_{\mrm{accept}}, P_{\mrm{reject}})$ on $(\C^{\dout \times \dout})^{\otimes n} \otimes \C^{\daux \times \daux}$.
    \end{itemize}
    The tester performs $P$ on $(\chan M^{\otimes n} \otimes \chan I) \psi$ and accepts or rejects according to the measurement outcome.
\end{dfn}

We do not formally define ancilla-free, adaptive channel testers or ancilla-assisted, incoherent, adaptive channel testers since we do not prove any results in these models. Informally however, they are the same as their non-adaptive counterparts except that the choice of $\psi_j$ and $P^{(j)}$ may depend on the classical information obtained from the previous $j-1$ measurement outcomes.

\begin{dfn}[Ancilla-assisted, coherent, adaptive channel tester] \label{dfn:tester111}
    A (deterministic) \emph{ancilla-assisted, coherent, adaptive channel tester} making $n$ queries to a channel $\chan M \in \cptp \din \dout$ consists of the following:
    \begin{itemize}
        \item channels $\chan V_1, \dotsc, \chan V_n \in \cptp {\dout \otimes \daux} {\din \otimes \daux}$, for some $\daux \in \N$;
        \item a two-outcome POVM $P = (P_{\mrm{accept}}, P_{\mrm{reject}})$ on $\C^{\dout \times \dout} \otimes \C^{\daux \otimes \daux}$.
    \end{itemize}
    The tester performs $P$ on
        $(\chan M \otimes \chan I) \chan V_n
        (\chan M \otimes \chan I) \chan V_{n-1}
        \dotsb
        (\chan M \otimes \chan I) \chan V_1 (\kb0)$
    and accepts or rejects according to the measurement outcome.
\end{dfn}

One may think of a randomized channel tester as a random variable taking values in the space of deterministic channel testers. We make the standard observation that if a randomized channel tester outputs the correct answer with high probability on worst-case channels, then some deterministic channel tester in its support outputs the correct answer with high probability on random channels:

\begin{lem} \label{lem:derand-tester}
    Let $\bs T$ be a randomized channel tester, let $A$ be a set of channels such that $\bs T(\chan A)$ accepts with probability at least $p$ for all $\chan A \in A$, and let $B$ be a set of channels such that $\bs T(\chan B)$ accepts with probability at most $q$ for all $\chan B \in B$, where the probabilities are over both the choice of $\bs T$ and over the randomness of the output measurement. Let $\bs{\chan A}$ and $\bs{\chan B}$ be random channels with support in $A$ and $B$ respectively. Then there exists a deterministic channel tester $T$ in the support of $\bs T$ such that $\PR{T(\bs{\chan A}) \text{ accepts}} - \PR{T(\bs{\chan B}) \text{ accepts}} \ge p-q$, where the probability is over both the choice of $\bs{\chan A}$ and $\bs{\chan B}$ and over the randomness of the output measurement.
\end{lem}
\begin{proof}
    For all fixed channels $\chan A \in A$ and $\chan B \in B$, by definition
    \begin{equation*}
        \PR{\bs T(\chan A) \text{ accepts}} - \PR{\bs T(\chan B) \text{ accepts}} \ge p-q.
    \end{equation*}
    Sampling $\bs T$ independently of $\bs{\chan A}$ and $\bs{\chan B}$, it follows that
    \begin{equation*}
        \PR{\bs T(\bs{\chan A}) \text{ accepts}} - \PR{\bs T(\bs{\chan B}) \text{ accepts}} \ge p-q,
    \end{equation*}
    and the result follows by fixing $T$ appropriately.
\end{proof}

\subsection{Von Neumann entropy}

We will use von Neumann entropy to prove our results about tomography in \cref{sec:tom}.

\begin{dfn}[Von Neumann entropy]
    The \emph{von Neumann entropy} of a density matrix $\rho$ is the quantity $S(\rho) \coloneqq -\tr(\rho \log \rho)$, i.e.\ the Shannon entropy of the spectrum of $\rho$. If $\rho$ is implicit and is in a register $\reg{A}$, then we sometimes refer to this quantity as $S(\reg{A})$. Similarly if $\rho$ is in registers $\reg{AB}$, then $S(\reg A)$ denotes the von Neumann entropy of the reduced state of $\rho$ on $\reg A$. We sometimes write $S_\rho(\cdot)$ to clarify $\rho$.
\end{dfn}

It holds for all density matrices $\rho \in \dens d$ that~\cite[Theorem 11.8(2)]{NC10}
\begin{equation} \label{thm:maxEntropy}
    S(\rho) \le \log d.
\end{equation}
Von Neumann entropy satisfies a property known as subadditivity~\cite[Eq.\ 11.72]{NC10}, i.e.\
\begin{equation} \label{thm:Subadd}
    S(\reg{AB}) \le S(\reg{A}) + S(\reg{B}),
\end{equation}
and a property known as the triangle inequality~\cite[Eq.\ 11.73]{NC10}, i.e.\
\begin{equation} \label{thm:Triangle}
    \Mag{S(\reg A)-S(\reg B)} \le S(\reg{AB}).
\end{equation}
If density matrices $\rho_1, \dotsc, \rho_n \in \dens d$ are supported on orthogonal subspaces, then~\cite[Theorem 11.10]{NC10}
\begin{equation} \label{lem:EntMix}
    S\Paren{\frac1n \sum_{j=1}^n \rho_j}
    = \frac1n \sum_{j=1}^n S\Paren{\rho_j}
    + \log n.
\end{equation}

\begin{dfn}[Conditional von Neumann entropy]\label{dfn:CondEnt}
    The \emph{conditional von Neumann entropy} of a state in registers $\reg{A}$ and $\reg{B}$ is the quantity $S(\reg A|\reg B) \coloneqq S(\reg{AB}) - S(\reg{B})$.
\end{dfn}

If a channel transforms a register $\reg B$ into a register $\reg B^\prime$, leaving another register $\reg A$ untouched, then~\cite[Theorem 11.5(3) and Eq.\ 11.64]{NC10}
\begin{equation} \label{thm:StrongSubadd}
    S(\reg A | \reg B) \le S(\reg A | \reg B^\prime).
\end{equation}

\begin{lem} \label{lem:FAapp}
    Let $\rho, \sigma \in \dens{\reg{AB}}$ be density matrices where $\reg A$ is a $d$-dimensional register and $\reg B$ is an $m$-dimensional register. Then
    \begin{equation*}
        S_\rho(\reg A | \reg B) \le S_\sigma(\reg A | \reg B) + \Norm{\rho - \sigma}_1 \Paren{\frac12 \log(d) + \log(m)} + 2.
    \end{equation*}
\end{lem}
\begin{proof}
    The Fannes–-Audenaert inequality~{\cite[Theorem 1]{A07}} states that if $\rho^\prime, \sigma^\prime \in \Dens{d^\prime}$ are density matrices and $x = \frac12 \Norm{\rho^\prime - \sigma^\prime}_1$, then
    \begin{equation*}
        \Mag{S\Paren{\rho^\prime} - S\Paren{\sigma^\prime}}
        \le x\log\Paren{d^\prime-1} -x\log(x) - (1-x)\log(1-x),
    \end{equation*}
    from which it follows that
    \begin{equation*}
        \Mag{S\Paren{\rho^\prime} - S\Paren{\sigma^\prime}}
        \le x \log\Paren{d^\prime} + 1.
    \end{equation*}
    In particular,
    \begin{equation*}
        S_\rho(\reg{AB}) - S_\sigma(\reg{AB}) \le \frac12 \norm{\rho - \sigma}_1 \log(dm) + 1.
    \end{equation*}
    Similarly, letting $\rho_{\reg B}$ and $\sigma_{\reg B}$ respectively denote the reduced states of $\rho$ and $\sigma$ in $\reg B$, and since tracing out a register cannot increase the trace distance between two states, it holds that
    \begin{equation*}
        S_\sigma (\reg B) - S_\rho (\reg B)
        \le \frac12 \norm{\sigma_{\reg B} - \rho_{\reg B}}_1 \log(m) + 1
        \le \frac12 \norm{\rho - \sigma}_1 \log(m) + 1.
    \end{equation*}
    Therefore
    \begin{align*}
        S_\rho(\reg A | \reg B) - S_\sigma(\reg A | \reg B)
        &= S_\rho(\reg{AB}) - S_\rho(\reg B) - S_\sigma(\reg{AB}) + S_\sigma(\reg B) \\
        &\le \frac12 \norm{\rho - \sigma}_1 (\log(dm) + \log(m)) + 2 \\
        &= \Norm{\rho - \sigma}_1 \Paren{\frac12 \log(d) + \log(m)} + 2. \qedhere
    \end{align*}
\end{proof}

\section{Lower bounds for channel certification in diamond distance}
\label{sec:diamond}

\diamondlb*

\begin{proof}
We define a random channel $\bs{\chan M_\phi} \in \cptp \din \dout$ as follows: let $\bs\phi \in \dens \din$ be Haar random, let $\bs\psi \in \dens \dout$ be the eigenstate corresponding to the smallest eigenvalue\footnote{By making an arbitrarily small perturbation to $\chan N$, it can be guaranteed that $\chan N(\bs\phi)$ has a unique smallest eigenvalue almost surely.} of $\chan N(\bs\phi)$, and let
\begin{equation*}
    \bs{\chan M_\phi}(X) = (1-\eps) \chan N(X) + \eps \tr(\bs\phi X) \bs\psi + \eps \chan N((I - \bs\phi) X (I - \bs\phi)).
\end{equation*}
It is straightforward to verify that $\bs{\chan M_\phi}$ is completely positive and trace-preserving. One may alternatively verify that $\bs{\chan M_\phi}$ is a channel by interpreting it as the following sequence of physical operations: with probability $1-\eps$ apply $\chan N$, and with probability $\eps$ apply the channel that first performs the PVM $(\bs\phi, I-\bs\phi)$ on the input state, and then outputs $\bs\psi$ if the measurement outcome was $\bs\phi$ and outputs $\chan N$ applied to the post-measurement state if the measurement outcome was $I-\bs\phi$. Thus $\bs{\chan M_\phi}$ behaves similarly to $\chan N$ except on inputs near $\bs\phi$.

It follows from the definition of $\bs{\chan M_\phi}$ that
\begin{equation} \label{eq:MN}
(\bs{\chan M_\phi} - \chan N) X
= \eps \Paren{\chan N(\bs\phi X \bs\phi - X \bs\phi - \bs\phi X) + \tr(\bs\phi X) \bs\psi}
\end{equation}
pointwise for all $X$. Consequently,
\begin{align*}
    \norm{\bs{\chan M_\phi} - \chan N}_\diamond
    &\ge \norm{\bs{\chan M_\phi} - \chan N}_1 \\
    &\ge \norm{(\bs{\chan M_\phi} - \chan N) \bs\phi}_1 \\
    &= \eps \norm{\bs\psi - \chan N(\bs\phi)}_1 
    &&\text{\cref{eq:MN}} \\
    &\ge \eps \tr\Paren{(\bs\psi - \chan N(\bs\phi)) \cdot (2\bs\psi - I)} &&\text{\cref{eq:trace-norm-var-def}} \\
    &= 2 \eps \Paren{1 - \tr(\chan N(\bs\phi) \bs\psi)} \\
    &\ge 2 \eps (1 - 1/\dout)
    &&\text{definition of $\bs\psi$} \\
    &\ge \eps
    &&\dout \ge 2.
\end{align*}
Therefore by \cref{lem:derand-tester} it suffices to prove that every deterministic, ancilla-assisted, coherent, adaptive channel tester $T$ requires $\Omega(\sqrt\din / \eps)$ queries in order to satisfy the following inequality:
\begin{equation} \label{eq:diff-acc-prob}
    \PR{T(\chan N) \text{ accepts}} - \PR{T(\bs{\chan M_\phi}) \text{ accepts}} \ge 1/3,
\end{equation}
where the probability is over both the choice of $\bs{\chan M_\phi}$ and the randomness of the output measurement. Recalling \cref{dfn:tester111}, write $T = (\chan V_1,\dotsb, \chan V_n, P)$ where $n$ is the number of queries made by $T$; our goal is to prove that $n = \Omega(\sqrt\din / \eps)$.

Let $\rho_0 = \bs\tau_{0,\phi} = \kb0$, and for $j \in [n]$ let
\begin{align*}
    &\rho_j = (\chan N \otimes \chan I) \chan V_j \rho_{j-1},
    &&\bs\tau_{j,\phi} = (\bs{\chan M_\phi} \otimes \chan I) \chan V_j \bs\tau_{j-1,\phi},
    &&\bs\sigma_{j,\phi} = (\bs{\chan M_\phi} \otimes \chan I) \chan V_j \rho_{j-1}.
\end{align*}
In particular, $\rho_n$ and $\bs\tau_{n,\phi}$ are the pre-measurement states in the executions of $T(\chan N)$ and $T(\bs{\chan M_\phi})$ respectively, so by conditioning on the choice of $\bs{\chan M_\phi}$ it follows from \cref{eq:diff-acc-prob} that
\begin{equation} \label{eq:third}
    1/3 \le \frac12 \E_{\bs{\phi}}\norm{\rho_n - \bs\tau_{n,\phi}}_1.
\end{equation}

For $j \in [n]$, by the triangle inequality
\begin{equation*}
    \E_{\bs{\phi}} \norm{\rho_j - \bs\tau_{j,\phi}}_1
    \leq \E_{\bs{\phi}} \norm{\rho_j - \bs\sigma_{j,\phi}}_1 + \E_{\bs{\phi}} \norm{\bs\sigma_{j,\phi} - \bs\tau_{j,\phi}}_1.
\end{equation*}
We now bound both terms in the latter expression. First, writing $\xi = \chan V_j \rho_{j-1}$, we have that
\begin{align*}
    \E_{\bs{\phi}} \norm{\rho_j - \bs\sigma_{j,\phi}}_1
    &= \E_{\bs{\phi}} \Norm{((\chan N - \bs{\chan M_\phi}) \otimes \chan I) \cdot \xi}_1 \\
    &= \eps \E_{\bs{\phi}} \Norm{\chan N \Paren{(\bs\phi \otimes I) \xi (\bs\phi \otimes I) - \xi (\bs\phi \otimes I) - (\bs\phi \otimes I) \xi} + \tr\Paren{(\bs\phi \otimes I) \xi} \psi}_1 &&\text{\cref{eq:MN}} \\
    &\le \eps \E_{\bs{\phi}} \Brac{\Norm{\chan N \Paren{(\bs\phi \otimes I) \xi (\bs\phi \otimes I) - \xi (\bs\phi \otimes I) - (\bs\phi \otimes I) \xi}}_1 + \tr\Paren{(\bs\phi \otimes I) \xi}}
    &&\text{triangle ineq.} \\
    &\le \eps \Paren{ \E_{\bs{\phi}} \Norm{(\bs\phi \otimes I) \xi (\bs\phi \otimes I) - \xi (\bs\phi \otimes I) - (\bs\phi \otimes I) \xi}_1 + 1/\din}
    &&\norm{\chan N}_1 = 1 \\
    &\le \eps \Paren{3 \E_{\bs{\phi}} \norm{(\bs\phi \otimes I) \xi}_1 + 1/\din}
    &&\text{triangle ineq.}\footnotemark,
\end{align*}
\footnotetext{And also using that by H\"older's inequality, $\norm{(\bs\phi \otimes I) \xi (\bs\phi \otimes I)}_1 \le \norm{(\bs\phi \otimes I) \xi}_1 \norm{\bs\phi \otimes I}_\infty = \norm{(\bs\phi \otimes I) \xi}_1$. }
and if $\xi = \sum_j \lambda_j \kb{\eta_j}$ is an eigendecomposition of $\xi$ then by convexity
\begin{align*}
    \E_{\bs{\phi}}\norm{(\bs\phi \otimes I) \xi}_1
    &\le \sum_j \lambda_j \E_{\bs{\phi}}\norm{(\bs\phi \otimes I) \eta_j}_1 \\
    &= \sum_j \lambda_j \E_{\bs{\phi}} \sqrt{\bra{\eta_j} (\bs\phi \otimes I) \ket{\eta_j}} \\
    &\le \sum_j \lambda_j \sqrt{\E_{\bs{\phi}}\Brac{\bra{\eta_j} (\bs\phi \otimes I) \ket{\eta_j}}} \\
    &= \sum_j \lambda_j \sqrt{\bra{\eta_j} (I/\din \otimes I) \ket{\eta_j}} \\
    &= 1/\sqrt\din,
\end{align*}
so
\begin{equation*}
    \E_{\bs{\phi}} \norm{\rho_j - \bs\sigma_{j,\phi}}_1
    \le \eps\Paren{3/\sqrt\din + 1/\din}
    \le 4\eps/\sqrt\din.
\end{equation*}
Second, since applying a channel to two states cannot increase the trace distance between them,
\begin{equation*}
    \E_{\bs{\phi}}\Norm{\bs\sigma_{j,\phi} - \bs\tau_{j,\phi}}_1
    = \E_{\bs{\phi}}\Norm{(\bs{\chan M_\phi} \otimes \chan I) \chan V_j (\rho_{j-1} - \bs\tau_{j-1,\phi})}_1
    \le \E_{\bs{\phi}}\Norm{\rho_{j-1} - \bs\tau_{j-1,\phi}}_1.
\end{equation*}
Combining the above inequalities yields
\begin{equation*}
    \E_{\bs{\phi}} \norm{\rho_j - \bs\tau_{j,\phi}}_1
    \le 4\eps/\sqrt\din + \E_{\bs{\phi}}\Norm{\rho_{j-1} - \bs\tau_{j-1,\phi}}_1,
\end{equation*}
so by induction $\E_{\bs{\phi}} \norm{\rho_n - \bs\tau_{n,\phi}}_1 \le 4n\eps/\sqrt\din$, and comparing with \cref{eq:third} reveals that $n = \Omega(\sqrt\din / \eps)$ as desired.
\end{proof}

\replacementlb*
\begin{proof}
    The proof is the same as that of \cref{thminf:diamond-LB}, except using the stronger bound $\E_{\bs{\phi}}\norm{\rho_j - \bs\sigma_{j,\phi}} \le 2\eps/\din$ in place of $\E_{\bs{\phi}}\norm{\rho_j - \bs\sigma_{j,\phi}} \le 4\eps/\sqrt\din$. The stronger bound holds because by \cref{eq:MN},
    \begin{equation*}
        (\bs{\chan M_\phi} - \chan N) X = \eps \tr(\bs\phi X) \cdot (\bs\psi - \theta)
    \end{equation*}
    for all $X$, so by the triangle inequality
    \begin{align*}
        \E_{\bs{\phi}} \norm{\rho_j - \bs\sigma_{j,\phi}}_1
        &= \E_{\bs{\phi}} \Norm{((\chan N - \bs{\chan M_\phi}) \otimes \chan I) \cdot \xi}_1 \\
        &= \eps \E_{\bs{\phi}} \Brac{\norm{\bs\psi - \theta}_1 \Norm{\Paren{\mrm{tr}_{\din} \otimes \chan I} \cdot ((\bs\phi \otimes I) \xi)}_1} \\
        &\le 2\eps \E_{\bs{\phi}} \Norm{\Paren{\mrm{tr}_{\din} \otimes \chan I} \cdot ((\bs\phi \otimes I) \xi)}_1 \\
        &= 2\eps \E_{\bs{\phi}} \tr\Paren{\Paren{\mrm{tr}_{\din} \otimes \chan I} \cdot ((\bs\phi \otimes I) \xi)} \\
        &= 2 \eps/\din. \qedhere
    \end{align*}
\end{proof}

\section{Proof that the ACID norm is \texorpdfstring{``}{"}average-case"} \label{sec:E-Phi-def}

In this section we prove \cref{thminf:compil}, i.e.\ we give conditions under which $\norm{(\chan L \otimes \chan I_{\daux}) \bs\psi}_1$ concentrates around $\norm{\chan L}_\lsd$ for a superoperator $\chan L \in \cptp \din \dout$ and Haar random $\bs\psi \in \dens{\din \otimes \daux}$. For technical reasons it will be convenient to refer to an \emph{unnormalized} version of the maximally entangled state in this section:

\begin{dfn}
    Let $\ket{\Psi_d} = \sum_{i=1}^d \ket{ii} = \sqrt{d} \ket{\Phi_d}$, and $\Psi_d = \kb{\Psi_d} = \sum_{i,j=1}^d \ketbra{ii}{jj} = d \Phi_d$. When the dimension $d$ is implicit we will simply write $\ket\Psi$ or $\Psi$.
\end{dfn}

It will also be convenient to have a shorthand notation for the quantity $\norm{(\chan L \otimes \chan I_{\daux}) \psi}_1$ which we are relating to $\norm{\chan L}_{\lsd}$. In particular, since this quantity depends only on the reduced state $\rho$ on the first register of $\psi$, it will be convenient to have a shorthand notation in terms of $\chan L$ and $\rho$ only. One purification\footnote{I.e.\ a pure state whose reduced state on the first register equals $\rho$.} of $\rho$ is $(\sqrt \rho \otimes I_{\din}) \ket{\Psi_{\din}}$, as can be straightforwardly verified using the fact that the partial trace over the second register of $\Psi$ equals $I$. This motivates the following definition, which is equivalent to the trace norm of the operator defined by applying $\chan L \otimes \chan I$ to this purification of $\rho$, and which generalizes the ACID norm (by taking $\rho = I/d$):

\begin{dfn}[$\rho$ norm] \label{def:rho-norm}
    For a density matrix $\rho \in \dens d$ and superoperator $\chan L \in \spr d$, let
	\begin{equation*}
		\norm{\chan L}_\rho = \norm{(\chan L \otimes \chan I_d) \cdot (\sqrt{\rho} \otimes I_d) \Psi_d (\sqrt{\rho} \otimes I_d)}_1.
	\end{equation*}
    We call $\norm\cdot_\rho$ the \emph{$\rho$ norm}.
\end{dfn}

The rest of this section is organized as follows. In \cref{sec:ub-acid} we prove some useful (in)equalities involving the $\rho$ norm for fixed $\rho$. In \cref{sec:unit-inv} we prove that if a random density matrix $\bs\rho$ is \emph{unitarily invariant}, meaning $U \bs\rho \adj U$ is distributed identically to $\bs \rho$ for all fixed unitaries $U$, and if furthermore $\bs \rho$ has constant expected fidelity with the maximally mixed state, then $\E \norm{\chan L}_{\bs\rho} = \Theta(\norm{\chan L}_{\lsd})$. In \cref{sec:ev-haar} we specialize this result to the case where $\bs\rho$ is the reduction of a Haar random state, by bounding the expected fidelity of the reduction of a Haar random state with the maximally mixed state. Finally, in \cref{sec:tail} we prove tail bounds on $\norm{\chan L}_{\bs\rho}$ when $\bs\rho$ is the reduction of a Haar random state.

\subsection{The \texorpdfstring{$\rho$}{rho} norm} \label{sec:ub-acid}

The following will turn out to be a more convenient phrasing of \cref{def:rho-norm}:

\begin{lem} \label{prp:equiv-def}
For all $\rho \in \dens d$ and $\chan L \in \spr d$,
    \begin{equation*}
        \norm{\chan L}_\rho 
        = \Norm{\Paren{I \otimes \sqrt\rho^\top} \cdot (\chan L \otimes \chan I) \Psi \cdot \Paren{I \otimes \sqrt\rho^\top}}_1.
    \end{equation*}
    \end{lem}
\begin{proof}
    By \cref{eq:Phi-trans} it holds that
	\begin{equation*}
		\norm{\chan L}_\rho
		= \Norm{(\chan L \otimes \chan I) \cdot \Paren{I \otimes \sqrt\rho^\top} \Psi \Paren{I \otimes \sqrt\rho^\top}}_1
		= \Norm{\Paren{I \otimes \sqrt\rho^\top} \cdot (\chan L \otimes \chan I) \Psi \cdot \Paren{I \otimes \sqrt\rho^\top}}_1,
	\end{equation*}
	where the last equality holds because applying $\chan L$ on the first register commutes with conjugating by $\sqrt\rho^\top$ on the second register.
\end{proof}

The significance of \cref{prp:equiv-def} is that it characterizes $\norm{\chan L}_\rho$ in terms of the (unnormalized) Choi operator $(\chan L \otimes \chan I) \Psi$, which also appears (normalized) in the definition \cref{dfn:acid} of $\norm{\chan L}_\lsd$. We will use this observation to relate the $\rho$ and ACID norms, starting with the following bound:

\begin{lem} \label{lem:ip}
    For all superoperators $\chan L \in \spr d$, there exists a density matrix $\sigma \in \dens d$ such that for all density matrices $\rho \in \dens d$, it holds that $\norm{\chan L}_\rho \le \tr(\rho \sigma) d \norm{\chan L}_\lsd$.
\end{lem}
\begin{proof}
    Let $(\chan L \otimes \chan I) \Psi = \sum_j s_j \ketbra{u_j}{v_j}$ be a singular value decomposition of $(\chan L \otimes \chan I) \Psi$.  Then for all density matrices $\rho \in \dens d$,
    \begin{align*}
        \norm{\chan L}_{\rho^\top}
        &= \Norm{(I \otimes \sqrt\rho) \cdot \sum_j s_j \ketbra{u_j}{v_j} \cdot (I \otimes \sqrt\rho)}_1
        &&\text{\cref{prp:equiv-def}} \\
        &\le \sum_j s_j \Norm{(I \otimes \sqrt\rho) \ketbra{u_j}{v_j} (I \otimes \sqrt\rho)}_1
        &&\text{triangle inequality} \\
        &= \sum_j s_j \Norm{(I \otimes \sqrt\rho) \ket{u_j}}_2  \Norm{(I \otimes \sqrt\rho) \ket{v_j}}_2 \\
        &\le \frac12 \sum_j s_j \left(\Norm{(I \otimes \sqrt\rho) \ket{u_j}}_2^2 +  \Norm{(I \otimes \sqrt\rho) \ket{v_j}}_2^2\right)
        &&\text{AM-GM inequality} \\
        &= \tr\Paren{(I \otimes \rho) \cdot \frac12 \sum_j s_j (u_j + v_j)} \\
        &= \tr(\rho M),
    \end{align*}
    where in the last step we define $M = (\tr_d \otimes \chan I_d) \cdot \frac12 \sum_j s_j (u_j + v_j)$. Since $M$ is PSD and
    \begin{equation*}
        \tr(M) = \sum_j s_j = \norm{(\chan L \otimes \chan I) \Psi}_1 = d \norm{\chan L}_\lsd,
    \end{equation*}
    we may write $M = d \norm{\chan L}_\lsd \sigma^\top$ where $\sigma$ is a density matrix. Finally,
    \begin{equation*}
        \norm{\chan L}_{\rho^\top}
        \le \tr\Paren{(\rho M)^\top}
        = \tr\Paren{M^\top \rho^\top}
        = \tr\Paren{\sigma \rho^\top} d \norm{\chan L}_\lsd. \qedhere
    \end{equation*}
\end{proof}

We remark that \cref{lem:ip} implies an alternate proof of \cref{thm:diamond-acid-relat}, as $\norm{\chan L}_\diamond = \max_\rho \norm{\chan L}_\rho$ and $\tr(\rho \sigma) \le 1$.

\subsection{Bounds on the expected \texorpdfstring{$\bs\rho$}{rho} norm for unitarily invariant \texorpdfstring{$\bs\rho$}{rho}} \label{sec:unit-inv}

Call a random density matrix $\bs\rho$ \emph{unitarily invariant} if $U \bs \rho \adj U$ is distributed identically to $\bs\rho$ for all fixed unitaries $U$. In other words, the spectrum of $\bs\rho$ may be sampled arbitrarily, but conditioned on the spectrum the eigenvectors are Haar random. We prove the following, where the expectation is over both the eigenvalues and eigenvectors of $\bs\rho$:

\begin{restatable}{thm}{uinv} \label{prp:u-inv}
    Let $\bs\rho \in \dens d$ be a unitarily invariant random density matrix, where $d>1$. Then for all superoperators $\chan L \in \spr d$,
    \begin{equation*}
        \frac{d^2 \E[\fid{\bs\rho}{I/d}] - 1}{d^2(2 - \E[\fid{\bs\rho}{I/d}]) - 1} \norm{\chan L}_\lsd
        \le \E\norm{\chan L}_{\bs\rho}
        \le \norm{\chan L}_\lsd.
    \end{equation*}
\end{restatable}

In \cref{app:u-inv-tight} we give examples where the bounds in \cref{prp:u-inv} are (approximately) tight, and in \cref{app:u-inv-random} we give examples where the bounds fail to hold if $\bs\rho$ is replaced with a fixed density matrix. We remark that the lower bound in \cref{prp:u-inv} may be improved by up to a constant factor if $\norm{\chan L(I)}_1$ is given, by an easy modification of the following proof.

\begin{proof}
First we prove the upper bound on $\E \norm{\chan L}_{\bs\rho}$. By \cref{lem:ip} there exists a density matrix $\sigma \in \dens d$ such that for all density matrices $\rho \in \dens d$,
\begin{equation*}
    \norm{\chan L}_\rho \le \tr(\rho \sigma) d \norm{\chan L}_\lsd.
\end{equation*}
Since $\bs\rho$ is unitarily invariant we have $\E[\bs\rho] = I/d$, and therefore
\begin{equation*}
    \E \norm{\chan L}_{\bs\rho} \le \tr(\E[\bs\rho] \sigma) d \norm{\chan L}_\lsd = \norm{\chan L}_\lsd.
\end{equation*}

Now we prove the lower bound on $\E\norm{\chan L}_{\bs\rho}$. Write $\bs\rho = \bs U \bs D \adj{\bs U}$ where $\bs U$ is a Haar random unitary independent of the random diagonal density matrix $\bs D$, and write $\bs D = \sum_{i=1}^d \bs \lambda_i \kb{i}$. Let $F = \E[\fid{\bs\rho}{I/d}]$, and note that
\begin{align*}
    dF = \E\Brac{(\tr{\sqrt{\bs\rho}})^2}
    = \E\Brac{\Paren{\sum_i \sqrt{\bs\lambda_i}}^2}
    = \sum_{i \neq j} \E\Brac{\sqrt{\bs\lambda_i \bs\lambda_j}} + 1,
\end{align*}
where the last equality uses linearity of expectation and the fact that $\bs\rho$ has unit trace.

Let
\begin{align*}
    &A = \E[\kb{\bs\psi} \otimes \kb{\bs\psi^*}],
    &B = \E[\ketbra{\bs\psi}{\bs\phi} \otimes \ketbra{\bs\psi^*}{\bs\phi^*}],
\end{align*}
where $\ket{\bs\psi}, \ket{\bs\phi} \in \C^d$ are orthogonal Haar random states. Then the quantity $(\sqrt{\bs\rho} \otimes I) \Psi (\sqrt{\bs\rho} \otimes I)$ appearing in the definition \cref{def:rho-norm} of $\norm{\chan L}_{\bs\rho}$ satisfies
\begin{align*}
    \E\Brac{(\sqrt{\bs\rho} \otimes I) \Psi (\sqrt{\bs\rho} \otimes I)}
    &= \E\left[\Paren{\bs U \sqrt{\bs D} \adj{\bs U} \otimes I} \Psi \left(\bs U \sqrt{\bs D} \adj{\bs U} \otimes I\right)\right] \\
    &= \E\left[\left(\bs U \otimes \bs U^*\right) \left(\sqrt{\bs D} \otimes I\right) \Psi \left(\sqrt{\bs D} \otimes I\right) \left(\adj{\bs U} \otimes \bs U^\top\right)\right] \\
    &= \sum_{i,j=1}^d \E\left[\left(\bs U \otimes \bs U^*\right) \left(\sqrt{\bs D} \otimes I\right) \ketbra{ii}{jj} \left(\sqrt{\bs D} \otimes I\right) \left(\adj{\bs U} \otimes \bs U^\top\right)\right] \\
    &= \sum_{i,j=1}^d \E\Brac{\sqrt{\bs\lambda_i \bs\lambda_j}} \E\left[\left(\bs U \otimes \bs U^*\right) \ketbra{ii}{jj} \left(\adj{\bs U} \otimes \bs U^\top\right)\right] \\
    &= \sum_i \E[\bs\lambda_i] A + \sum_{i \neq j} \E \Brac{\sqrt{\bs\lambda_i \bs\lambda_j}} B \\
    &= A + (dF-1) B,
\end{align*}
where the second equality uses \cref{eq:Phi-trans}, and the fourth equality uses that $\bs U$ and $\bs D$ are independent.

We now solve for $A$ and $B$. Taking the transpose of the second register on both sides of \cref{lem:sym} gives
\begin{equation*}
    A = \frac{\Psi + I}{d(d+1)}.
\end{equation*}
Next, similarly to the above,
\begin{align*}
    \Psi
    &= \E\Brac{\Paren{\bs U \adj{\bs U} \otimes I} \Psi \Paren{\bs U {\adj{\bs U}} \otimes I}}
    = \E\Brac{\Paren{\bs U \otimes \bs U^*} \Psi \Paren{\adj{\bs U} \otimes \bs U^\top}} \\
    &= \sum_{i,j=1}^d \E\Brac{\Paren{\bs U \otimes \bs U^*} \ketbra{ii}{jj} \Paren{\adj{\bs U} \otimes \bs U^\top}}
    = \sum_i A + \sum_{i \neq j} B \\
    &= dA + d(d-1)B = \frac{\Psi + I}{d+1} + d(d-1)B,
\end{align*}
and rearranging gives
\begin{equation*}
    B = \frac\Psi{(d+1)(d-1)} - \frac{I}{(d+1)d(d-1)}.
\end{equation*}

Therefore
\begin{align*}
    \E\Brac{(\sqrt{\bs\rho} \otimes I) \Psi (\sqrt{\bs\rho} \otimes I)}
    &= A + (dF-1) B \\
    &= \Paren{\frac1{d(d+1)} + \frac{dF-1}{(d+1)(d-1)}} \Psi + \Paren{\frac1{d(d+1)} + \frac{1-dF}{(d+1)d(d-1)}} I \\
    &= \frac{d^2 F - 1}{(d+1)d(d-1)} \cdot \Psi + \frac{1-F}{(d+1)(d-1)} \cdot I.
\end{align*}
Rearranging gives
\begin{equation*}
    (d^2 F - 1) \frac\Psi{d}
    = (d^2-1) \E\Brac{(\sqrt{\bs\rho} \otimes I) \Psi (\sqrt{\bs\rho} \otimes I)} - (1-F) I,
\end{equation*} 
implying
\begin{equation*}
    (d^2 F - 1) \cdot (\chan L \otimes \chan I) (\Psi/d)
    = (d^2-1) \E\Brac{(\chan L \otimes \chan I) \cdot (\sqrt{\bs\rho} \otimes I) \Psi (\sqrt{\bs\rho} \otimes I)} - (1-F) (\chan L(I) \otimes I).
\end{equation*}

Therefore by the triangle inequality,
\begin{align*}
    (d^2 F - 1) \norm{\chan L}_\lsd
    &\le (d^2 - 1) \Norm{\E\Brac{(\chan L \otimes \chan I) \cdot (\sqrt{\bs\rho} \otimes I) \Psi (\sqrt{\bs\rho} \otimes I)}}_1 + (1-F) \Norm{\chan L(I) \otimes I}_1 \\
    &\le (d^2 - 1) \E \norm{\chan L}_{\bs\rho} + d^2(1-F) \norm{\chan L(I/d)}_1.
\end{align*}
Furthermore, since $\bs\rho$ is unitarily invariant,
\begin{equation*}
    \norm{\chan L(I/d)}_1
    = \norm{\chan L(\E[\bs\rho])}_1
    \le \E \norm{\chan L(\bs\rho)}_1,
\end{equation*}
and by \cref{eq:trace-norm-var-def} it holds for all fixed density matrices $\rho \in \dens d$ that
\begin{equation*}
    \norm{\chan L(\rho)}_1
    = \max_{\norm{B}_\infty = 1} \Mag{\tr\Paren{\chan L(\rho) B}}
    = \max_{\norm{B}_\infty = 1} \Mag{\tr\Paren{(\chan L \otimes \chan I)((\sqrt\rho \otimes I) \Psi (\sqrt\rho \otimes I)) \cdot (B \otimes I)}}
    \le \norm{\chan L}_\rho,
\end{equation*}
so
\begin{equation*}
    \Paren{d^2 F - 1} \norm{\chan L}_\lsd
    \le \Paren{d^2 - 1 + d^2(1-F)} \E\norm{\chan L}_{\bs\rho}
    = \Paren{d^2(2-F) - 1} \E\norm{\chan L}_{\bs\rho}.
\end{equation*}
Finally, since $d>1$ the quantity $d^2(2-F) - 1$ is strictly positive, so we may divide both sides of the above inequality by $d^2(2-F) - 1$, yielding
\begin{equation*}
    \frac{d^2 F - 1}{d^2(2-F) - 1} \norm{\chan L}_\lsd \le \E\norm{\chan L}_{\bs\rho}. \qedhere
\end{equation*}
\end{proof}

\subsection{Bounds on the expected \texorpdfstring{$\bs\rho$}{rho} norm when \texorpdfstring{$\bs\rho$}{rho} is the reduction of a Haar random state} \label{sec:ev-haar}

We now apply \cref{prp:u-inv} to the case where $\bs\rho$ is the reduction of a Haar random state. We will use the following two lemmas:

\begin{lem} \label{lem:fid2}
    For all density matrices $\rho \in \dens d$, it holds that $\fid\rho{I/d} \ge 1 / d\norm\rho_2^2$.
\end{lem}

We remark that \cref{lem:fid2} is tight when $\rho$ is maximally mixed on some subspace.

\begin{proof}
    For all $x \ge 0$,
    \begin{equation*}
        0 \le (\sqrt x - 1)^2 (\sqrt x + 2) \sqrt x
        = x^2 - 3x + 2 \sqrt x,
    \end{equation*}
    and rearranging gives
    \begin{equation*}
        \sqrt{x} \ge \frac32 x - \frac12 x^2.
    \end{equation*}
    Therefore it holds for all $r \ge 0$ that
	\begin{equation*}
		\sqrt{r \rho} \ge \frac32 r\rho - \frac12 (r\rho)^2
	\end{equation*}
	in the Loewner order, and therefore
	\begin{equation*}
		\sqrt{\fid\rho{I/d}}
		= \frac1{\sqrt{rd}} \tr\sqrt{r\rho}
		\ge \frac1{\sqrt{d}} \Paren{\frac32 r^{1/2} - \frac12 r^{3/2} \norm\rho_2^2}.
	\end{equation*}
	The result follows by plugging in $r = 1 / \norm\rho_2^2$, which maximizes the above bound.
\end{proof}

\begin{lem}[{Lubkin~\cite[after Eq.\ 15]{Lub78}}] \label{lem:S2-haar-red}
    Let $\bs\rho \in \dens d$ be the reduction of a Haar random state in $\C^d \otimes \C^m$. Then $\E \left[\norm{\bs\rho}_2^2\right] = (d+m)/(dm+1)$.
\end{lem}

We include a proof below for completeness:

\begin{proof}
	Let $\ket{\bs\psi} \in \C^d \otimes \C^m$ be the Haar random state of which $\bs\rho$ is the reduction. Throughout this proof, sums over the variables $i,j$ are from $1$ to $d$, and sums over the variables $s,t$ are from $1$ to $m$. Write
	\begin{equation*}
		\ket{\bs\psi} = \sum_{i,s} \bs\alpha_{is} \ket{is}.
	\end{equation*}
	Then
	\begin{align*}
		&\bs\psi = \sum_{\mathclap{i,j,s,t}} \bs\alpha_{is} \bs\alpha_{jt}^* \ketbra{is}{jt},
		&&\bs\rho = \sum_{i,j,s} \bs\alpha_{is} \bs\alpha_{js}^* \ketbra{i}j,
	\end{align*}
	so
	\begin{equation*}
		\norm{\bs\rho}_2^2
		= \sum_{i,j} \Mag{\sum_s \bs\alpha_{is} \bs\alpha_{js}^*}^2
		= \sum_{\mathclap{i,j,s,t}} \bs\alpha_{is} \bs\alpha_{js}^* \bs\alpha_{it}^* \bs\alpha_{jt}.
	\end{equation*}
	Therefore by \cref{lem:haar-moment},
	\begin{align*}
		\E \left[\norm{\bs\rho}_2^2\right]
		&= \sum_{\mathclap{i,j,s,t}}
		\begin{cases}
			2/dm(dm+1) & \text{if $i=j$ and $s=t$} \\
			1/dm(dm+1) & \text{if $i=j$ xor $s=t$} \\
			0 & \text{otherwise}.
		\end{cases} \\
		&= dm \cdot \frac{2}{dm(dm+1)} + dm(d+m-2) \cdot \frac1{dm(dm+1)} \\
		&= \frac{d+m}{dm+1}. \qedhere
	\end{align*}
\end{proof}

Now we combine the above results to prove the following:

\begin{cor} \label{cor:haar-e}
Let $\bs\rho \in \dens d$ be the reduction of a Haar random state in $\C^d \otimes \C^m$, where $d>1$. Then for all superoperators $\chan L \in \spr d$,
\begin{equation*}
    \frac{m}{2d+m} \norm{\chan L}_\lsd
    \le \E\norm{\chan L}_{\bs\rho}
    \le \norm{\chan L}_\lsd.
\end{equation*}
\end{cor}

\begin{proof}
The upper bound $\E\norm{\chan L}_{\bs\rho} \le \norm{\chan L}_\lsd$ is that in \cref{prp:u-inv}. The lower bound holds because by \cref{lem:fid2}, Jensen's inequality, and \cref{lem:S2-haar-red},
\begin{equation} \label{eq:lb-fid}
    \E[\fid{\bs\rho}{I/d}]
    \ge \frac1d \E\Brac{\frac1{\norm{\bs\rho}_2^2}}
    \ge \frac1{d \E\Brac{\norm{\bs\rho}_2^2}}
    = \frac{dm+1}{d(d+m)},
\end{equation}
so by \cref{prp:u-inv},
\begin{equation*}
    \E\norm{\chan L}_{\bs\rho}
    \ge \frac{d^2 (dm+1)/d(d+m) - 1} {d^2 \Paren{2 - (dm+1)/d(d+m)} - 1} \norm{\chan L}_\lsd
    = \frac{d^2 m - m}{2d^3 + d^2 m - 2d - m} \norm{\chan L}_\lsd
    = \frac{m}{2d+m} \norm{\chan L}_\lsd. \qedhere
\end{equation*}
\end{proof}

We remark that \cref{eq:lb-fid} is tight to within a factor of 2, as can be shown using the fact that $\rank(\bs\rho) \le m$.

\subsection{Tail bounds on the \texorpdfstring{$\bs\rho$}{rho} norm when \texorpdfstring{$\bs\rho$}{rho} is the reduction of a Haar random state} \label{sec:tail}

We now prove tail bounds on $\norm{\chan L}_{\bs\rho}$ to complement \cref{cor:haar-e}, where again $\bs\rho \in \dens d$ is the reduction of a Haar random state in $\C^d \otimes \C^m$. Since \cref{cor:haar-e} implies that $\E \norm{\chan L}_{\bs\rho} = \Theta(\norm{\chan L}_\lsd)$ assuming $m = \Omega(d)$, our goal here is to prove that $\norm{\chan L}_{\bs\rho} = \Theta(\norm{\chan L}_\lsd)$ \emph{with high probability} under the same assumption. Unfortunately we fall slightly short of this goal, and instead prove two complementary tail bounds that approach it in different ways. The first tail bound, proved using L\'evy’s lemma, implies that $\norm{\chan L}_{\bs\rho} = \Theta(\norm{\chan L}_\lsd)$ with high probability provided that either $m \ge \omega(d)$ or $\norm{\chan L}_\diamond \le o(d \norm{\chan L}_\lsd)$ (in the latter case, still assuming $m = \Omega(d)$). For comparison, recall from \cref{thm:diamond-acid-relat} that $\norm{\chan L}_\diamond \le O(d \norm{\chan L}_\lsd)$, which falls just short of the latter criterion for worst-case superoperators $\chan L$. The second tail bound, proved using \cref{lem:ip}, implies the one-sided inequality $\norm{\chan L}_{\bs\rho} \le O(\norm{\chan L}_\lsd)$ with high probability assuming only that $m \ge \omega(\log d)$.

Let $\sphere{d-1} = \{x \in \R^d: \norm{x}_2 = 1\}$ denote the $d$-dimensional unit sphere. A function $f: \sphere{d-1} \to \R$ is \emph{$L$-Lipschitz} if $|f(x) - f(y)| \le \norm{x - y}_2$ for all $x,y \in \sphere{d-1}$, and such functions obey the following concentration inequality:

\begin{lem}[{L\'evy’s lemma~\cite[Corollary 5.4]{Mec19}}] \label{thm:levy}
Let $f: \sphere{d-1} \to \R$ be $L$-Lipschitz, and let $\bs x \in \sphere{d-1}$ be uniform random. Then for all $t \ge 0$,
\begin{equation*}
    \PR{\Mag{f(\bs x) - \E f(\bs x)} \ge t} \le \exp\Paren{\pi - \frac{d t^2}{4 L^2}}.
\end{equation*}
\end{lem}

By showing that $\norm{\chan L}_\rho$ is $2 \norm{\chan L}_\diamond$-Lipschitz\footnote{When $m < d$, we may replace $\norm{\chan L}_\diamond$ with $\max_\psi \norm{(\chan L \otimes \chan I) \cdot \psi}_1$ in \cref{thm:levy-app}, where $\psi$ ranges over all pure states in $\dens{d \otimes m}$. (When $m \ge d$ this quantity equals $\norm{\chan L}_\diamond$.)} as a function of a purification of $\rho$, we prove the following:

\begin{thm} \label{thm:levy-app}
Let $\bs\rho \in \dens d$ be the reduction of a Haar random state in $\C^d \otimes \C^m$. Then for all superoperators $\chan L \in \spr d$ and all $t \ge 0$,
\begin{equation*}
    \PR{\Mag{\norm{\chan L}_{\bs\rho} - \E \norm{\chan L}_{\bs\rho}} \ge t \norm{\chan L}_\lsd}
    \le \exp\Paren{\pi - \frac{dm t^2 \norm{\chan L}_\lsd^2}{8 \norm{\chan L}_\diamond^2}}
    \le \exp\Paren{\pi - \frac{m t^2}{8d}}.
\end{equation*}
\end{thm}

\begin{proof}
For a pure state $\ket\psi \in \C^d \otimes \C^m$, let $f(\ket\psi) = \norm{\chan L}_\rho$ where $\rho$ is the reduced state on the first register of $\psi$. By identifying $\C^d \otimes \C^m$ with $\R^{2dm}$ in the natural way, we can identify the domain of $f$ with the sphere $\sphere{2dm-1} \subseteq \R^{2dm}$. Thus for all pure states $\ket\psi, \ket\phi \in \C^d \otimes \C^m$,
\begin{align*}
    f(\ket\psi) - f(\ket\phi)
    &= \norm{(\chan L \otimes \chan I) \psi}_1 - \norm{(\chan L \otimes \chan I) \phi}_1 \\
    &\le \norm{(\chan L \otimes \chan I) (\psi - \phi)}_1
    &&\text{triangle inequality} \\
    &\le \norm{\chan L}_\diamond \norm{\psi - \phi}_1 \\
    &\le 2 \norm{\chan L}_\diamond \norm{\ket\psi - \ket\phi}_2
    &&\text{\cref{eq:td2d}.}
\end{align*}
In other words $f$ is $2 \norm{\chan L}_\diamond$-Lipschitz, so by \cref{thm:levy}
\begin{equation*}
    \PR{\Mag{f(\ket{\bs\psi}) - \E f(\ket{\bs\psi})} \ge t} \le \exp\Paren{\pi - \frac{dmt^2}{8 \norm{\chan L}_\diamond^2}}
\end{equation*}
for Haar random $\ket{\bs\psi}$ and $t \ge 0$, which is equivalent to the first inequality in the theorem statement. The second inequality follows from \cref{thm:diamond-acid-relat}.
\end{proof}

Now we prove our second tail bound:

\begin{thm} \label{thm:up-tail-haar}
Let $\bs\rho \in \dens d$ be the reduction of a Haar random state in $\C^d \otimes \C^m$. Then for all superoperators $\chan L \in \spr d$ and all $t \ge 10$,
\begin{equation*}
    \PR{\norm{\chan L}_{\bs\rho} \ge t \norm{\chan L}_\lsd}
    \le 2d \exp(-tm/8).
\end{equation*}
\end{thm}

We did not attempt to optimize the constants in \cref{thm:up-tail-haar}.

\begin{proof}
By \cref{lem:ip} there exists a density matrix $\sigma \in \dens d$ such that for all density matrices $\rho \in \dens d$, it holds that $\norm{\chan L}_\rho \le \tr(\rho \sigma) d \norm{\chan L}_\lsd$. At this point we could note that by Hölder's inequality $\tr(\bs\rho \sigma) \le \norm{\bs\rho}_\infty \norm\sigma_1 = \norm{\bs\rho}_\infty$ and bound $\norm{\bs\rho}_\infty$ using a matrix Chernoff bound, but this is wasteful: we don't need to bound $\norm{\bs\rho \ket\phi}_2$ for \emph{every} pure state $\ket\phi$, but only for those $\ket\phi$ that are eigenvectors of $\sigma$. Concretely, let $\sigma = \sum_{j=1}^d \lambda_j \kb{\phi_j}$ be an eigendecomposition of $\sigma$. Then $\tr(\rho \sigma) \le \max_j \tr(\rho \phi_j)$ for all density matrices $\rho \in \dens d$, so by a union bound\footnote{A tighter but more complicated bound follows from a result of Hsu, Kakade and Zhang~\cite[Proposition 1]{HKZ12}.}
\begin{equation*}
    \PR{\norm{\chan L}_{\bs\rho} \ge t \norm{\chan L}_\lsd}
    \le \PR{\tr(\bs\rho \sigma) \ge t/d}
    \le \sum_{j=1}^d \PR{\tr(\bs\rho \phi_j) \ge t/d}
    = d \cdot \PR{\bra0 \bs\rho \ket0 \ge t/d}.
\end{equation*}

Let
\begin{equation*}
    \ket{\bs g} = \sum_{i=1}^d \sum_{j=1}^m \bs g_{ij} \ket{ij}
\end{equation*}
be a vector with independent standard complex Gaussian elements $\bs g_{ij}$, and let $\ket{\bs g} / \norm{\ket{\bs g}}_2$ be the Haar random pure state of which $\bs\rho$ is the reduction. Also write
\begin{equation*}
    \ket{\bs g} =\frac{\ket{\bs a} + i \ket{\bs b}}{\sqrt 2}
\end{equation*}
where $\ket{\bs a}, \ket{\bs b}$ are vectors with independent standard real Gaussian elements. Then
\begin{equation*}
    \bra0 \bs\rho \ket0
    = \frac{\sum_{j=1}^m \Mag{\bs g_{0j}}^2}{\sum_{i=1}^d \sum_{j=1}^m \Mag{\bs g_{ij}}^2}
    = \frac{\sum_{j=1}^m \Paren{\Mag{\bs a_{0j}}^2 + \Mag{\bs b_{0j}}^2}}{\sum_{i=1}^d \sum_{j=1}^m \Paren{\Mag{\bs a_{ij}}^2 + \Mag{\bs b_{ij}}^2}}.
\end{equation*}

The numerator $\bs N$ and denominator $\bs D$ of the latter expression are respectively $\chi^2(2m)$ and $\chi^2(2dm)$ random variables, where we write $\chi^2(k)$ to denote the $\chi^2$ distribution with $k$ degrees of freedom. A $\chi^2(k)$ random variable $\bs X$ obeys the tail bounds~\cite[Eqs.\ (4.3) and (4.4)]{LM00}
\begin{align*}
    &\PR{\bs X/k \le 1 - 2\sqrt{s}} \le \exp(-ks),
    &&\PR{\bs X/k \ge 1 + 2\sqrt{s} + 2s} \le \exp(-ks),
\end{align*}
for all $s \ge 0$, and if $s \ge 1$ then the latter bound implies
\begin{equation*}
    \PR{\bs X/k \ge 5s} \le \exp(-ks).
\end{equation*}
Therefore by a union bound,
\begin{align*}
    \PR{\bra0 \bs\rho \ket0 \ge t/d}
    &= \PR{\bs N / \bs D \ge t/d} \\
    &\le \PR{\bs D/2dm \le 1/2} + \PR{\bs N/2m \ge t/2} \\
    &\le \exp(-dm/8) + \exp(-tm/5) \\
    &\le 2 \exp(-tm/8),
\end{align*}
where the latter inequality assumes $t \le d$ (if $t > d$ then $\PR{\bra0 \bs\rho \ket0 \ge t/d} = 0$ trivially).
\end{proof}

\section{Channel certification and tomography in ACID distance}

\label{sec:acid-certification}

\subsection{Upper bounds for arbitrary channels}

We first prove the following:

\begin{thm}[Channel certification in $\ell_2$ distance between Choi states] \label{thm:Inco_l2_chan_cert}
    For all fixed channels $\chan N \in \cptp \din \dout$ and $\eps>0$, there exists an ancilla-free, non-adaptive algorithm that makes $O\Paren{\dout^{1/2} \log^3(1/\eps) \big/ \eps^2}$ queries to a channel $\chan M$, and decides whether $\chan M = \chan N$ or $\norm{\choi M - \choi N}_2 \ge \eps$ with success probability at least $2/3$.
\end{thm}

The proof uses the following result of Bao and Yao~\cite{Bao23Junta}, of which we provide a (somewhat different) proof in Appendix \ref{app:BaoYao} for completeness:

\begin{lem}[{\cite[Proposition 15]{Bao23Junta}}] \label{lem:BaoYao}
    If $\chan L \in \spr d$ is the difference between two channels, then
    \begin{equation*}
        \frac{d+1}d \E\Brac{\norm{\chan L(\bs\psi)}_2^2}
        = \norm{\choi L}_2^2 + \norm{\chan L(I/d)}_2^2
    \end{equation*}
    where $\bs\psi \in \dens d$ is Haar random.
\end{lem}

\begin{proof}[Proof of \cref{thm:Inco_l2_chan_cert}]
    Chen, Li and O'Donnell~\cite[Lemma 6.2]{CLO22} proved that for all fixed states $\sigma \in \dens d$ and $\delta, \eta > 0$, there exists an algorithm $\textsc{CertifyL2}(\sigma, \delta, \eta)$ that takes as input $O\Paren{\sqrt{d} \log(1/\delta) \big/ \eta^2}$ copies of an unknown state $\rho \in \dens d$, performs unentangled and non-adaptive measurements on the copies of $\rho$, and then accepts with probability at least $1-\delta$ if $\rho = \sigma$ and rejects with probability at least $1-\delta$ if $\norm{\rho - \sigma}_2 > \eta$. Let $t = \ceil*{\log\Paren{1/\eps^2}} + 4$ and $\delta = \eps^2 / 384 t$, and assume without loss of generality that $\eps$ is small enough so that $\delta \le 1/3$. The algorithm is \cref{alg:cap}, and its query complexity is
    \begin{equation*}
        \sum_{k=1}^t 2^{k+1} t \cdot O\Paren{\dout^{1/2} \log(1/\delta) \big/ \eps^2 2^{k-3}}
        \le O\Paren{t^2 \dout^{1/2} \log(1/\delta) \big/ \eps^2}
        \le O\Paren{\dout^{1/2} \log^3(1/\eps) \big/ \eps^2}.
    \end{equation*}
    If $\chan M = \chan N$, then by a union bound \cref{alg:cap} accepts rejects with probability at most
    \begin{equation*}
        \sum_{k=1}^t 2^{k+1} t \cdot \delta
        = \Paren{2^{t+2} - 4} \eps^2 / 384
        \le 2^{\log(1/\eps^2) + 7} \eps^2 / 384
        = 1/3.
    \end{equation*}

    \begin{algorithm}
    \caption{Channel certification in $\ell_2$ distance between Choi states} \label{alg:cap}
    \begin{algorithmic}[1]
    \For{$k \in [t]$}
        \For{$2^{k+1} t$ times}
            \State Sample a Haar random state $\bs\psi \in \dens\din$.
            \State Run $\textsc{CertifyL2}\Paren{\chan N(\bs\psi), \delta, \eps 2^{(k-3)/2}}$ on copies of $\chan M(\bs\psi)$.
        \EndFor
    \EndFor
    \If{all runs of \textsc{CertifyL2} accepted} \textbf{accept}.
    \Else{} \textbf{reject}.
    \EndIf
    \end{algorithmic}
    \end{algorithm}
    
    Now suppose $\norm{\choi M - \choi N}_2 \ge \eps$. Below we prove that there exists some fixed $k \in [t]$ such that
    \begin{equation} \label{eq:binary-search}
        \PR{\Norm{\chan M(\bs\psi) - \chan N(\bs\psi)}_2^2 > \eps^2 2^{k-3}} \ge 2^{-k}/t,
    \end{equation}
    so \cref{alg:cap} accepts with probability at most
    \begin{equation*}
        \Paren{1 - (1 - \delta) 2^{-k}/t}^{2^{k+1} t}
        \le \exp\Paren{-(1-\delta) 2^{-k}/t \cdot 2^{k+1} t}
        = \exp\Paren{-(1-\delta) 2}
        \le \exp(-4/3)
        < 1/3.
    \end{equation*}
    Toward establishing \cref{eq:binary-search} for some $k \in [t]$, let $\chan L = \chan M - \chan N$ and $\bs X = \norm{\chan L(\bs\psi)}_2^2$. By \cref{lem:BaoYao},
    \begin{equation} \label{eq:eps-X}
        \eps^2
        \le \norm{\choi L}_2^2
        \le \norm{\choi L}_2^2 + \norm{\chan L(I/\din)}_2^2
        = \frac{\din+1}{\din} \E[\bs X]
        \le 2 \E[\bs X].
    \end{equation}
    Define disjoint intervals
    \begin{align*}
        &P_0 = \left[0, \frac{\eps^2}4 \right],
        &&P_k = \left(\frac{\eps^2}4 2^{k-1}, \frac{\eps^2}4 2^k\right] \quad \text{for $k \in [t]$.}
    \end{align*}
    By the triangle inequality $\norm{\chan L}_1 \le \norm{\chan M}_1 + \norm{\chan N}_1 = 2$, so
    \begin{equation*}
        0 \le \bs X \le \norm{\chan L(\bs\psi)}_1^2 \le \norm{\chan L}_1^2 \le 4 \le \eps^2/4 \cdot 2^t
    \end{equation*}
    pointwise, so there exists a unique $k$ such that $\bs X$ is in $P_k$, and therefore
    \begin{equation*}
        \E[\bs X]
        = \sum_{k=0}^t \PR{\bs X \in P_k} \E[\bs X \mid \bs X \in P_k].
    \end{equation*}
    Since the expectation of a random variable is at most its maximum possible value, it follows that
    \begin{equation*}
        \E[\bs X]
        \le \sum_{k=0}^t \PR{\bs X \in P_k} \cdot \frac{\eps^2}4 2^k
        = \frac{\eps^2}4 \sum_{k=0}^t \PR{\bs X \in P_k} \cdot 2^k.
    \end{equation*}
    If $\PR{\bs X \in P_k} < 2^{-k}/t$ for all $k \neq 0$, then it follows from this inequality and the trivial bound $\PR{\bs X \in P_0} \le 1$ that
    \begin{equation*}
        \E[\bs X]
        < \frac{\eps^2}4 \Paren{1 + \sum_{k=1}^t 2^{-k}/t \cdot 2^k}
        = \eps^2/2,
    \end{equation*}
    which contradicts \cref{eq:eps-X}. Thus there exists $k \in [t]$ such that $\PR{\bs X \in P_k} \ge 2^{-k}/t$ as desired.
\end{proof}

As corollaries, we obtain upper bounds for channel certification in ACID and diamond distances:

\generalub*
\begin{proof}
    By Cauchy-Schwarz, $\norm{\choi M - \choi N}_2 \ge \norm{\choi M - \choi N}_1 / \sqrt{\din \dout}$, so if $\norm{\chan M - \chan N}_{\lsd} \ge \eps$ then $\norm{\choi M - \choi N}_2 \ge \eps / \sqrt{\din \dout}$. The result follows by applying \cref{thm:Inco_l2_chan_cert} with proximity parameter $\eps / \sqrt{\din \dout}$.
\end{proof}

\generalubdiamond*
\begin{proof}
    Fawzi et al.~\cite[Lemma C.1]{FFGO23} proved that $\norm{\choi M - \choi N}_2 \ge \norm{\chan M - \chan N}_\diamond \big/ \Paren{\din \dout^{1/2}}$, so if $\norm{\chan M - \chan N}_\diamond \ge \eps$ then $\norm{\choi M - \choi N}_2 \ge \eps / \din \dout^{1/2}$. The result follows by applying \cref{thm:Inco_l2_chan_cert} with proximity parameter $\eps / \din \dout^{1/2}$.
\end{proof}

Finally, we remove the log factors from \cref{thminf:ub-general} in the case where $\chan N$ is the completely depolarizing channel:

\begin{thm}[Upper bound for the completely depolarizing channel] \label{thm:ub-depol}
    Let $\chan N \in \cptp \din \dout$ be the completely depolarizing channel, i.e.\ $\chan N(X) = \tr(X) I/\dout$. Then there is an ancilla-free, non-adaptive algorithm that makes $O\Paren{\din \dout^{3/2} / \eps^2}$ queries to a channel $\chan M$, and decides whether $\chan M = \chan N$ or $\norm{\chan M - \chan N}_{\lsd} \ge \eps$ with success probability at least $2/3$.
\end{thm}

\begin{proof}
For a probability distribution $P$ with finite support, let $v(P)$ denote the total variation distance between $P$ and the uniform distribution. Paninski~\cite{Pan08} gave an algorithm $\textsc{TestMixed}\Paren{\delta, d}$ that takes as input $O(\sqrt{d} / \delta^2)$ samples from a probability distribution $P$ on $[d]$, and decides whether $P$ is the uniform distribution or $v(P) \ge \delta$ with success probability at least $2/3$. We may assume without loss of generality that the success probability of $\textsc{TestMixed}\Paren{\delta, d}$ is at least $1 - 1/30000$, by repetition and majority vote. For a universal constant $c$ and all $d \in \N$, Fawzi et al.~\cite[Lemma E.1]{FFGO23} gave a random $cd$-outcome POVM $\bs P_d$ on $\C^{d \times d}$ such that for all fixed density matrices $\rho \in \dens d$, if $\bs P_d(\rho)$ denotes the probability distribution defined by performing $\bs P_d$ on $\rho$, then
\begin{equation*}
    \PR{v(\bs P_d(\rho)) \ge \frac{\norm{\rho - I/d}_2}{20}} \ge 1/2.
\end{equation*}
The algorithm is \cref{alg:depol}. If $\chan M = \chan N$, then by a union bound \cref{alg:depol} accepts with probability at least $1 - 10^4/30000 = 2/3$.

\begin{algorithm}
\caption{Testing identity to the completely depolarizing channel}
\label{alg:depol}
\begin{algorithmic}[1]
    \For{$10^4$ times}
        \State Independently sample a Haar random state $\bs\psi \in \dens\din$ and a POVM $\bs P_{\dout}$.
        \State Run $\textsc{TestMixed}\Paren{\eps / 20\sqrt{2 (\din+1) \dout}, c \dout}$ on samples from $(\bs P_{\dout}(\chan M(\bs\psi)))$.
    \EndFor
    \If{all executions of $\textsc{TestMixed}$ accepted} \textbf{accept}.
    \Else{} \textbf{reject}.
    \EndIf
\end{algorithmic}
\end{algorithm}

Now consider the case where $\norm{\chan M - \chan N}_{\lsd} \ge \eps$. Let $\bs\psi \in \dens\din$ be the Haar random state from \cref{alg:depol}, and let $\bs X = \Norm{\chan M(\bs\psi) - I/\dout}_2^2$. Fawzi et al.~\cite[middle of Page 40]{FFGO23} proved that $\PR{\bs X \ge \E[\bs X]/2} \ge 1/1000$. If $\bs X \ge \E[\bs X]/2$ and $v(\bs P_{\dout}(\chan M(\bs\psi))) \ge \sqrt{\bs X} / 20$, an event which occurs with probability at least $1/2000$ by the definition of $\bs P_{\dout}$, then
\begin{align*}
    \eps
    &\le \norm{\choi M - I/\din\dout}_1 \\
    &\le \sqrt{\din \dout} \norm{\choi M - I/\din\dout}_2 &&\text{Cauchy-Schwarz} \\
    &\le \sqrt{(\din + 1) \dout \E[\bs X]} &&\text{\cite[Lemma C.2]{FFGO23}}\footnotemark \\
    &\le \sqrt{2 (\din+1) \dout \bs X} \\
    &\le \sqrt{2 (\din+1) \dout} \cdot 20 v(\bs P_{\dout}(\chan M(\bs\psi))),
\end{align*}
\footnotetext{This also follows from \cref{lem:BaoYao} with $\chan L = \chan M - \chan N$.}and rearranging gives $v(\bs P_{\dout}(\chan M(\bs\psi))) \ge \eps / 20\sqrt{2 (\din+1) \dout}$. By the definition of $\textsc{TestMixed}$, it follows that any given iteration of \cref{alg:depol} rejects with probability at least $(1 - 10^{-5}) \cdot 1/2000 > 1/4000$, and so overall \cref{alg:depol} accepts with probability at most $(1 - 1/4000)^{10^4} < 0.09$.
\end{proof}

\subsection{Upper bounds for erasure, unitary, and pure state replacement channels}

In this subsection we give dimension-independent upper bounds for testing identity to erasure, unitary, and pure state replacement channels in ACID distance, without ancillae or adaptivity. Along the way, we prove that testing identity to \emph{any} channel $\chan N$ in ACID distance has essentially the same complexity as that of testing identity to $\chan I \otimes \chan N$ in ACID distance. We build up to this result through a series of lemmas, starting with the following:

\begin{lem}[{Gentle measurement lemma~\cite[Lemma 9.4.1]{Wil19}}] \label{lem:gentle}
    Let $\rho$ be a density matrix and let $0 \le \Lambda \le I$. Then
    \begin{equation*}
        \Norm{\rho - \frac{\sqrt\Lambda \rho \sqrt\Lambda}{\tr(\Lambda\rho)}}_1
        \le 2\sqrt{1 - \tr(\Lambda \rho)}.
    \end{equation*}
\end{lem}

Using \cref{lem:gentle} we prove the following:

\begin{lem} \label{lem:some-lemma}
    Let $\rho \in \dens{\reg A \reg B}$ be a density matrix for some registers $\reg A$ and $\reg B$, and let $\rho_{\reg A}, \rho_{\reg B}$ be the reduced states of $\rho$ on $\reg A, \reg B$ respectively. Then for all pure states $\psi \in \dens{\reg A}$,
    \begin{equation*}
        \tr\Paren{\rho_{\reg A} \psi} \le 1 - \frac1{16} \Norm{\rho - \psi \otimes \rho_{\reg B}}_1^2.
    \end{equation*}
\end{lem}
\begin{proof}
    Let
    \begin{equation*}
        \sigma = \frac{(\psi \otimes I) \rho (\psi \otimes I)}{\tr\Paren{(\psi \otimes I) \rho}}
    \end{equation*}
    where $I$ denotes the identity on $\reg B$. By the triangle inequality,
    \begin{equation*}
        \Norm{\rho - \psi \otimes \rho_{\reg B}}_1
        \le \norm{\rho - \sigma}_1 + \Norm{\sigma - \psi \otimes \rho_{\reg B}}_1
        = \norm{\rho - \sigma}_1 + \Norm{\psi \otimes \tr_{\reg A} (\sigma - \rho)}_1
        \le 2\norm{\rho - \sigma}_1,
    \end{equation*}
    where the last inequality holds because applying a channel (in this case, tracing out $\reg A$ and then tensoring with $\psi$) to two density matrices cannot increase the trace distance between them. By \cref{lem:gentle} applied with $\Lambda = \psi \otimes I$,
    \begin{equation*}
        \norm{\rho - \sigma}_1
        \le 2\sqrt{1 - \tr\Paren{(\psi \otimes I) \rho}}
        = 2\sqrt{1 - \tr\Paren{\psi \rho_{\reg A}}},
    \end{equation*}
    and the result follows by combining the above two inequalities and rearranging.
\end{proof}

We use the following result to remove the need for ancillae in our upcoming algorithm:

\begin{lem}[{Fawzi et al.~\cite[Lemma A.1]{FFGO23}}] \label{lem:a1}
    For all channels $\chan P \in \cptp d d$,
    \begin{equation*}
	\E \Brac{\tr\Paren{\chan P(\bs\psi) \bs\psi}} = \frac{1 + d \tr\Paren{\choi P \Phi}}{1+d}
    \end{equation*}
    where $\bs\psi \in \dens d$ is Haar random.
\end{lem}
\begin{proof}
    This is the case of \cref{lem:Choi_l2_dist} where $\chan L = \chan P$ and $\chan K = \chan I$.
\end{proof}

Now we reduce the task of testing identity to $\chan I \otimes \chan N$ to that of testing identity to $\chan N$, for an arbitrary channel $\chan N$. To match the context in which we will apply this reduction, we phrase it in terms of ancilla-free, non-adaptive channel testers with perfect completeness, but the proof can easily be adopted to the other query models from \cref{sec:query-models} and to channel testers with imperfect completeness as well.

\begin{thm} \label{thm:N-red-IN}
    Let $\chan N \in \cptp \din \dout$ be a channel, and assume there exists an ancilla-free, non-adaptive algorithm that makes $n$ queries to a channel $\chan Q \in \cptp \din \dout$, accepts with probability $1$ if $\chan Q = \chan N$, and accepts with probability at most $1/2$ if $\norm{\chan Q - \chan N}_{\lsd} \ge \delta$. Then there is an ancilla-free, non-adaptive algorithm that makes $n + O(1/\eps^2)$ queries to a channel $\chan M \in \cptp {\daux \otimes \din} {\daux \otimes \dout}$, accepts with probability $1$ if $\chan M = \chan I_{\daux} \otimes \chan N$, and accepts with probability at most $1/2$ if $\norm{\chan M - \chan I_{\daux} \otimes \chan N}_{\lsd} \ge \eps + \delta$.
\end{thm}

\begin{proof}
    Let $\textsc{Certify}$ be the assumed algorithm for testing identity to $\chan N$. Define channels $\chan P \in \cptp \daux \daux$ and $\chan Q \in \cptp \din \dout$ by
    \begin{align*}
        &{\chan P(X)} = \Paren{\chan I_{\daux} \otimes \mrm{tr}_{\dout}} \chan M \Paren{X \otimes I_{\din}/\din},
        &{\chan Q(X)} = \Paren{\mrm{tr}_{\daux} \otimes \chan I_{\dout}} \chan M \Paren{I_{\daux}/\daux \otimes X}.
    \end{align*}
    The algorithm is \cref{alg:tensorI}, where queries to $\chan P$ and $\chan Q$ are implicitly simulated using queries to $\chan M$. If $\chan M = \chan I \otimes \chan N$ then $\chan P = \chan I$ and $\chan Q = \chan N$, and so the \cref{alg:tensorI} accepts with probability $1$.

    \begin{algorithm}
	\caption{Testing identity to $\chan I \otimes \chan N$}
	\label{alg:tensorI}
	\begin{algorithmic}[1]
	    \For{$\ceil{32 \ln(2) / \eps^2}$ times}
		\State Sample a Haar random state $\bs\psi \in \dens\daux$.
		\State Perform the PVM $(\bs\psi, I - \bs\psi)$ on $\chan P(\bs\psi)$.
		\EndFor
		\If{all of the measurement outcomes were $\bs\psi$ \textbf{and} $\textsc{Certify}(\chan Q)$ accepts} \textbf{accept}.
		\Else{} \textbf{reject}.
		\EndIf
	\end{algorithmic}
    \end{algorithm}

    Now suppose that $\norm{\chan M - \chan I \otimes \chan N}_{\lsd} \ge \eps + \delta$. By the triangle inequality,
    \begin{equation*}
        \eps + \delta
        \le \Norm{\choi M - \choi{I \otimes N}}_1
        \le \Norm{\choi M - \Phi_{\daux} \otimes \choi Q}_1 + \Norm{\Phi_{\daux} \otimes \choi Q - \choi{I \otimes N}}_1,
    \end{equation*}
    and since $\choi{I \otimes N} = \Phi_{\daux} \otimes \choi N$,
    \begin{equation*}
        \Norm{\Phi_{\daux} \otimes \choi Q - \choi{I \otimes N}}_1
        = \Norm{\choi Q - \choi N}_1
        = \norm{\chan Q - \chan N}_{\lsd},
    \end{equation*}
    so $\eps + \delta \le \Norm{\choi M - \Phi_{\daux} \otimes \choi Q}_1 + \norm{\chan Q - \chan N}_{\lsd}$. Therefore either $\eps \le \Norm{\choi M - \Phi_{\daux} \otimes \choi Q}_1$ or $\delta \le \norm{\chan Q - \chan N}_{\lsd}$. In the latter case, $\textsc{Certify}(\chan Q)$ accepts with probability at most $1/2$ and so \cref{alg:tensorI} accepts with probability at most $1/2$. In the former case, since $\choi P$ and $\choi Q$ are equal to the reduced states of $\choi M$ on $\dens{\daux \otimes \daux}$ and $\dens{\din \otimes \dout}$ respectively,\footnote{To see this, let $\reg A$ and $\reg B$ be $\daux$-dimensional registers and let $\reg C$ and $\reg D$ be $\din$-dimensional registers, and write $\choi M = (\chan M_{\reg{AC}}  \otimes \chan I_{\reg{BD}}) (\Phi_{\reg{AB}} \otimes \Phi_{\reg{CD}})$, where subscripts indicate which registers a superoperator acts on or a state is in. Tracing out $\reg D$ yields $(\chan M_{\reg{AC}}  \otimes \chan I_{\reg{B}}) (\Phi_{\reg{AB}} \otimes I_{\reg C} / \din)$, and then tracing out $\reg C$ (or more precisely, the $\dout$-dimensional register that $\chan M$ transforms $\reg C$ into) yields $(\chan P_{\reg A} \otimes \chan I_{\reg B}) \Phi_{\reg{AB}} = \choi P$. The argument for $\choi Q$ is similar.}
    \begin{align*}
        \E\Brac{\tr\Paren{\chan P (\bs\psi) \bs\psi}}
        &= \frac{1 + \daux \tr\Paren{\choi P \Phi_{\daux}}} {1+\daux}
        &&\text{\cref{lem:a1}} \\
        &\le \frac{1 + \daux \Paren{1 - \frac1{16} \Norm{\choi M - \Phi_{\daux} \otimes \choi Q}_1^2}}{1+\daux}
        &&\text{\cref{lem:some-lemma}} \\
        &\le \frac{1 + \daux \Paren{1 - \eps^2/16}}{1+\daux}\\
        &= 1 - \frac{\daux \eps^2}{16(1+\daux)}\\
        &\le 1 - \eps^2/32\\
        &\le \exp(-\eps^2/32),
    \end{align*}
    so again \cref{alg:tensorI} accepts with probability at most $\exp(-\eps^2/32 \cdot \ceil{32 \ln(2) / \eps^2}) \le 1/2$.
\end{proof}

Finally, we prove the main result of this subsection:
\erasureub*

\begin{proof}
    First suppose $\chan N$ is an erasure channel, i.e.\ $\chan N = \chan I_{\dout} \otimes \tr_{\din / \dout}$. Since $\tr_{\din / \dout}$ is the \emph{only} channel in $\cptp {\din / \dout} 1$, testing identity to $\tr_{\din / \dout}$ trivially requires zero queries even with perfect completeness and soundness, so the result follows from \cref{thm:N-red-IN}.

    Next suppose $\chan N$ is a unitary channel. We may assume without loss of generality that $\chan N$ is the identity channel, because if we define a channel $\chan P$ by $\chan P(X) = \adj U \chan M(X) U$ then
    \begin{align*}
        \Norm{\chan M - \chan N}_{\lsd}
        &= \Norm{\choi M - \choi N}_1 \\
        &= \Norm{\choi M - (U \otimes I) \Phi \Paren{\adj U \otimes I}} \\
        &= \Norm{\Paren{\adj U \otimes I} \choi M (U \otimes I) - \Phi}_1 \\
        &= \Norm{\choi P - \choi I}_1 \\
        &= \Norm{\chan P - \chan I}_{\lsd},
    \end{align*}
    and queries to $\chan P$ can be simulated using queries to $\chan M$. The identity channel is the erasure channel with input dimension equal to the output dimension, so the result follows by the above argument.

    Finally suppose $\chan N$ is a pure state replacement channel. The algorithm is \cref{alg:replace-pure}; clearly it accepts with probability 1 if $\chan M = \chan N$. If $\norm{\chan M - \chan N}_{\lsd} \ge \eps$, then by \cref{lem:some-lemma} (applied with $\rho = \choi M, \reg A = \C^{\dout \times \dout}, \reg B = \C^{\din \times \din}$),
    \begin{equation*}
        \tr\Paren{\chan M(I/\din) \psi}
        \le 1 - \frac1{16} \Norm{\choi M - \psi \otimes I/\din}_1^2
        = 1 - \frac1{16} \Norm{\choi M - \choi N}_1^2
        \le 1 - \eps^2/16
        \le \exp(-\eps^2/16),
    \end{equation*}
    so \cref{alg:replace-pure} accepts with probability at most $\exp(-\eps^2/16 \cdot \ceil{16 \ln(2) / \eps^2}) \le 1/2$.
\end{proof}

\begin{algorithm}
\caption{Testing identity to a pure state replacement channel}
\label{alg:replace-pure}
\begin{algorithmic}[1]
    \For{$\ceil{16 \ln(2) / \eps^2}$ times}
        \State Perform the PVM $(\psi, I - \psi)$ on $\chan M(I/\din)$.
    \EndFor
    \If{all of the measurement outcomes were $\psi$} \textbf{accept}.
    \Else{} \textbf{reject}.
    \EndIf
\end{algorithmic}
\end{algorithm}

\subsection{Lower bound for the completely depolarizing channel}\label{subsec:depol_lb}

The \emph{total variation distance} between discrete probability distributions $P$ and $Q$ is the quantity
\begin{equation*}
    d_{\mrm{TV}}(P,Q) \coloneqq \frac12 \sum_x |P(x) - Q(x)|.
\end{equation*}
We will use the following bound:

\begin{lem} \label{lem:hellinger}
Let $P_1, \dotsc, P_n$ be probability distributions on $\bits$, and let $U$ be the uniform distribution on $\bits$. Then
\begin{equation*}
    d_{\mrm{TV}}\Paren{\bigotimes_{i=1}^n P_i, U^{\otimes n}}
    \le 2 \sqrt{\sum_{i=1}^n d_{\mrm{TV}}(P_i, U)^2}.
\end{equation*}
\end{lem}

\begin{proof}
We use the well-known fact that $d_{\mrm{TV}}(P,Q) \le \sqrt2 \cdot d_{\mrm{H}}(P,Q)$ for all distributions $P,Q$, where
\begin{equation*}
    d_{\mrm{H}}(P,Q) \coloneqq \sqrt{1 - \sum_x \sqrt{P(x) Q(x)}}
\end{equation*}
denotes Hellinger distance. Let $P = \bigotimes_{i=1}^n P_i$ and $Q = U^{\otimes n}$, and write
\begin{equation*}
    P_i = \mrm{Bernoulli}\Paren{\frac{1 + x_i}2}
\end{equation*}
where $-1 \le x_i \le 1$. We use that for $-1 \le x \le 1$,
\begin{equation*}
    0 \ge (\sqrt{1+x}-1) (\sqrt{1-x}-1) = \sqrt{1 - x^2} - \sqrt{1+x} - \sqrt{1-x} + 1 \ge 2 - x^2 - \sqrt{1+x} - \sqrt{1-x},
\end{equation*}
which rearranges to $\sqrt{1-x} + \sqrt{1+x} \ge 2 - x^2$.\footnote{In fact, the stronger inequality $\sqrt{1-x} + \sqrt{1+x} \ge 2 - (2 - \sqrt2) x^2$ holds, but proving this is more time-consuming and is not necessary for our purposes.} It follows that
\begin{align*}
    \frac12 d_{\mrm{TV}}(P,Q)^2
    &\le d_{\mrm{H}}(P,Q)^2 \\
    &= 1 - \sum_{\mathclap{x \in \bits^n}} \sqrt{P(x) Q(x)} \\
    &= 1 - \prod_{i=1}^n \Paren{\sqrt{P_i(0) U(0)} + \sqrt{P_i(1) U(1)}} \\
    &= 1 - \prod_{i=1}^n \frac{\sqrt{1-x_i} + \sqrt{1+x_i}}2 \\
    &\le 1 - \prod_{i=1}^n \Paren{1 - \frac{ x_i^2}{2}} \\
    &\le \frac{1}{2} \sum_{i=1}^n x_i^2 \\
    &= 2 \sum_{i=1}^n d_{\mrm{TV}}(P_i,U)^2. \qedhere
\end{align*}
\end{proof}

Now we prove the following:

\lbdepol*

\begin{proof}
It will be convenient to identify the output space $\C^{\dout \times \dout}$ of $\chan N$ with $\C^{2 \times 2} \otimes \C^{\dout/2 \times \dout/2}$. Define density matrices $\rho_0, \rho_1 \in \dens\dout \cong\dens{2 \otimes \dout/2}$ by\footnote{We conjecture that conjugating $\rho_0$ and $\rho_1$ by a Haar random unitary would lead to an $\Omega\Paren{\din \dout^{3/2} \big/ \eps^2}$ lower bound, similarly to lower bound proofs for ancilla-free state certification~\cite{CLO22,CLHL22}.}
\begin{align*}
    &\rho_0 = \frac1{\dout} ((1+\eps) \kb0 + (1-\eps) \kb1) \otimes I_{\dout/2}, \\
    &\rho_1 = \frac1{\dout} ((1-\eps) \kb0 + (1+\eps) \kb1) \otimes I_{\dout/2}.
\end{align*}
Let $\bs\Pi \in \C^{\din \times \din}$ be the projection onto a Haar random $\din/2$-dimensional subspace of $\C^{\din}$, and define a channel $\bs{\chan M} \in \cptp \din \dout$ in terms of $\bs\Pi$ by
\begin{equation*}
    \bs{\chan M}(X) =
    \tr(X(I - \bs\Pi)) \rho_0 + \tr(X \bs\Pi) \rho_1.
\end{equation*}
Then
\begin{align*}
    &\choi{\bs{\chan M}}
    = \frac1{\din} \rho_0 \otimes \Paren{I - \bs\Pi^\top} + \frac1{\din} \rho_1 \otimes \bs\Pi^\top,
    &&J_{\chan N} = \frac1{\din \dout} I_{\din \dout},
\end{align*}
so $\Norm{\bs{\chan M} - \chan N}_{\lsd} = \Norm{\choi{\bs{\chan M}} - \choi N}_1 = \eps$ pointwise because all of the eigenvalues of $\choi{\bs{\chan M}}$ are $(1\pm\eps)/\din\dout$. Therefore by \cref{lem:derand-tester} it suffices to prove that if $T$ is a deterministic, ancilla-free, non-adaptive channel tester, and if
\begin{equation*}
    \PR{\text{$T(\chan N)$ accepts}} - \PR{\text{$T(\bs{\chan M})$ accepts}} \ge 1/3,
\end{equation*}
where the probability is over both the choice of $\bs{\chan M}$ and the randomness of the output measurements, then $T$ makes $n = \Omega(\din/\eps^2)$ queries.

Recalling \cref{dfn:tester000} of ancilla-free, non-adaptive channel testers, write
\begin{equation*}
    T = \Paren{\psi_1, \dotsc, \psi_n, P^{(1)}, \dotsc, P^{(n)}, f}.
\end{equation*}
The difference between the acceptance probabilities of $T$ on any two fixed channels is at most the trace distance between the pre-measurement states corresponding to those channels, so
\begin{equation*}
    1/3
    \le \E \frac12 \Norm{\bigotimes_{j=1}^n \chan N(\psi_j) - \bigotimes_{j=1}^n \bs{\chan M}(\psi_j)}_1 \\
    = \E \frac12 \Norm{I/\dout^n - \bigotimes_{j=1}^n \bs{\chan M}(\psi_j)}_1.
\end{equation*}
Let $\bs p_j = \tr(\bs\Pi \psi_j)$ for $1 \le j \le n$, and note that
\begin{equation*}
    \bs{\chan M}(\psi_j)
    = (1 - \bs p_j) \rho_0 + \bs p_j \rho_1
    = \frac1{\dout} \Paren{(1 + (1 - 2 \bs p_j) \eps) \kb0 + (1 + (2 \bs p_j - 1) \eps) \kb1} \otimes I_{\dout/2}.
\end{equation*}
It follows that
\begin{align*}
    1/3
    &\le \E \Brac{d_{\mrm{TV}} \Paren{\mrm{Bernoulli}(1/2)^{\otimes n}, \bigotimes_{j=1}^n \mrm{Bernoulli} \Paren{\frac12 + \Paren{\bs p_j - \frac12} \eps}}} \\
    &\le 2 \E \sqrt{\sum_{i=1}^n (\bs p_j - 1/2)^2 \eps^2}
    &&\text{\cref{lem:hellinger}} \\
    &\le 2 \eps \sqrt{\sum_{i=1}^n \E \Brac{(\bs p_j - 1/2)^2}}
    &&\text{Cauchy-Schwarz.}
\end{align*}

To compute $\E [(\bs p_j - 1/2)^2]$, write $\bs\Pi = \sum_{j=1}^{\din/2} \kb{\bs u_j}$ where $\ket{\bs u_1}, \dotsc, \ket{\bs u_{\din/2}} \in \C^{\din}$ are the first $\din/2$ columns of a Haar random unitary. Then taking sums from 1 to $\din/2$, for all states $\ket\psi$, by \cref{lem:haar-moment} we have
\begin{align*}
    \E\Brac{\tr(\bs\Pi \psi)^2}
    &= \E\Brac{\Paren{\sum_j \vert\braket{\psi}{\bs u_j}\vert^2}^2} \\
    &= \sum_{i \neq j} \E\left[\vert\braket{\psi}{\bs u_i}\vert^2 \vert\braket{\psi}{\bs u_j}\vert^2\right]
    + \sum_j \E\left[\vert\braket{\psi}{\bs u_j}\vert^4\right] \\
    &= \frac{\din}2 \Paren{\frac{\din}2 - 1} \cdot \frac1{\din(\din+1)} + \frac{\din}2 \cdot \frac2{\din(\din+1)}\\
    &= \frac14 + \frac1{4(\din+1)},
\end{align*}
and clearly $\E\Brac{\tr(\bs\Pi \psi)} = 1/2$. Therefore
\begin{equation*}
    \E\Brac{(\bs p_j - 1/2)^2} = \E\Brac{\bs p_j^2 - 1/4} = \frac1{4(\din+1)},
\end{equation*}
so $1/3 \le 2 \eps \sqrt{n/4(\din+1)}$, implying that $n = \Omega(\din / \eps^2)$ as desired.
\end{proof}

\subsection{Upper bound for arbitrary channels in an expanded query model}

For a superoperator $\chan L$, we define a superpoerator $\dual L$ by $\dual L(X) = \chan L\Paren{X^\top}^\top$. We prove the following:

\ubdual*

\begin{algorithm}
\caption{Testing identity to $\chan N$ using $\chan M, \dual M$}
\label{alg:dual}
\begin{algorithmic}[1]
    \For{$100 \dout^4 / \eps^4$ times}
        \State Perform the PVM $(\Phi_{\dout}, I-\Phi_{\dout})$ on $\Paren{\chan M \otimes \dual M} \Phi_{\din}$. \label{line:2}
    \EndFor
    \State Let $\bs p$ be the fraction of measurement outcomes from Line~\ref{line:2} that were $\Phi_{\dout}$.
    \For{$100 \dout^4 / \eps^4$ times}
        \State Perform the PVM $(\Phi_{\dout}, I-\Phi_{\dout})$ on $\Paren{\chan M \otimes \dual N} \Phi_{\din}$. \label{line:6}
    \EndFor
    \State Let $\bs q$ be the fraction of measurement outcomes from Line~\ref{line:6} that were $\Phi_{\dout}$.
    \If{$\bs p - 2 \bs q + \frac\din\dout \Norm{\choi N}_2^2 \le 0.5 \eps^2 / \dout^2$} \textbf{accept}. \label{line:9}
    \Else{} \textbf{reject}.
    \EndIf
\end{algorithmic}
\end{algorithm}

The algorithm behind our proof will be \cref{alg:dual}. The following lemma characterizes the distribution of measurement outcomes in this algorithm:

\begin{lem} \label{lem:dual-ip}
For all Hermitian-preserving superoperators $\chan K, \chan L \in \super \din \dout$,
\begin{equation*}
    \tr\Paren{\Phi_{\dout} \cdot \Paren{\chan K \otimes \dual L} \Phi_{\din}}
    = \frac\din\dout \tr\Paren{\choi L \choi K}.
\end{equation*}
\end{lem}

\begin{proof}
By linearity and \cref{eq:kraus}, we may assume without loss of generality that $\chan K(X) = A X \adj A$ and $\chan L(X) = B X \adj B$ for some matrices $A,B$. Then by the cyclic property of trace and \cref{eq:Phi-trans},
\begin{align*}
    \dout \tr\Paren{\Phi_{\dout} \cdot \Paren{\chan K \otimes \dual L} \Phi_{\din}}
    &= \dout \tr\Paren{\Phi_{\dout} (A \otimes B^*) \Phi_{\din} (\adj A \otimes B^\top)} \\
    &= \dout \tr\Paren{(I \otimes B^\top) \Phi_{\dout} (I \otimes B^*) \cdot (A \otimes I) \Phi_{\din} (\adj A \otimes I)} \\
    &= \din \tr\Paren{(B \otimes I) \Phi_{\din} (\adj B \otimes I) \cdot (A \otimes I) \Phi_{\din} (\adj A \otimes I)} \\
    &= \din \tr(\choi L \choi K). \qedhere
\end{align*}
\end{proof}

Now we prove \cref{thm:ub-dual}:

\begin{proof}
Consider an arbitrary channel $\chan M \in \cptp \din \dout$. Define $\bs p$ and $\bs q$ as in \cref{alg:dual}, and let
\begin{equation*}
    \bs X = \bs p - 2 \bs q + \frac\din\dout \norm{\choi N}_2^2
\end{equation*}
denote the quantity on the left side of the inequality in Line~\ref{line:9}. Then
\begin{align*}
    \E[\bs X]
    &= \tr \Paren{\Phi_{\dout} \cdot (\chan M \otimes \dual M) \Phi_{\din}} - 2 \tr \Paren{\Phi_{\dout} \cdot (\chan M \otimes \dual N) \Phi_{\din}} + \frac\din\dout \norm{\choi N}_2^2 \\
    &= \frac\din\dout \Paren{\tr(\choi M^2) - 2 \tr(\choi M \choi N) + \tr(\choi N^2)}
    &&\text{\cref{lem:dual-ip}} \\
    &= \frac\din\dout \Norm{\choi M - \choi N}_2^2 \\
    &\ge \frac1{\dout^2} \norm{\choi M - \choi N}_1^2
    &&\text{Cauchy-Schwarz} \\
    &= \frac1{\dout^2} \norm{\chan M - \chan N}_{\lsd}^2,
\end{align*}
with equality if $\chan M = \chan N$.

Since $\bs p$ is the average of $100 \dout^4 / \eps^4$ i.i.d.\ Bernoulli random variables, by a Chernoff bound it holds that
\begin{align*}
    \PR{\bs p - \E[\bs p] \ge 0.1 \eps^2 / \dout^2}
    \le \exp\Paren{-2 \cdot \Paren{0.1 \eps^2 / \dout^2}^2 \cdot 100 \dout^4 / \eps^4} = \exp(-2) < 0.14,
\end{align*}
and similarly
\begin{align*}
    \PR{\bs p - \E[\bs p] \le - 0.1 \eps^2 / \dout^2} \le 0.14, \\
    \PR{\bs q - \E[\bs q] \ge \phantom- 0.1 \eps^2 / \dout^2} \le 0.14, \\
    \PR{\bs q - \E[\bs q] \le - 0.1 \eps^2 / \dout^2} \le 0.14.
\end{align*}
Therefore
\begin{align*}
    \PR{\bs X - \E[\bs X] \ge 0.3 \eps^2 / \dout^2}
    &= \PR{(\bs p - \E[\bs p]) - 2(\bs q - \E[\bs q]) \ge 0.3 \eps^2 / \dout^2} \\
    &\le \PR{\bs p - \E[\bs p] \ge 0.1 \eps^2 / \dout^2 \text{ or } \bs q - \E[\bs q] \le - 0.1 \eps^2 / \dout^2} \\
    &\le \PR{\bs p - \E[\bs p] \ge 0.1 \eps^2 / \dout^2} + \PR{\bs q - \E[\bs q] \le - 0.1 \eps^2 / \dout^2} \\
    &\le 0.14 + 0.14 \\
    &\le 1/3,
\end{align*}
and similarly
\begin{equation*}
    \PR{\bs X - \E[\bs X] \le - 0.3 \eps^2 / \dout^2} \le 1/3.
\end{equation*}

Thus if $\chan M = \chan N$, then \cref{alg:dual} rejects with probability at most
\begin{equation*}
    \PR{\bs X \ge 0.3 \eps^2 / \dout^2}
    = \PR{\bs X - \E[\bs X] \ge 0.3 \eps^2 / \dout^2}
    \le 1/3.
\end{equation*}
And if $\norm{\chan M - \chan N}_{\lsd} \ge \eps$, then $\E[\bs X] \ge \eps^2 / \dout^2$, so \cref{alg:dual} accepts with probability at most
\begin{equation*}
    \PR{\bs X \le 0.7 \eps^2 / \dout^2}
    \le \PR{\bs X - \E[\bs X] \le -0.3 \eps^2 / \dout^2} \le 1/3. \qedhere
\end{equation*}
\end{proof}

\subsection{Tomography}\label{sec:tom}

\UBtom*

\begin{proof}
    O'Donnell and Wright~\cite[Theorem 1.10]{OW15} gave an algorithm $\textsc{StateTom}_\delta$ that performs an entangled measurement on $O(d^2/\delta^2)$ copies of a density matrix $\rho \in \C^{d \times d}$, and with probability at least $2/3$ outputs the description of a density matrix $\sigma \in \C^{d \times d}$ such that $\norm{\rho - \sigma}_1 \le \delta$. Our algorithm is to first perform $\textsc{StateTom}_{\eps/2}$ on $\choi M$, yielding the description of a state $\rho \in \dens{\din \otimes \dout}$, and then output the description of a channel $\chan N$ that minimizes $\norm{\rho - \choi N}_1$. If $\textsc{StateTom}_{\eps/2}$ succeeds, then by the triangle inequality
    \begin{equation*}
        \norm{\chan M - \chan N}_J
        \le \norm{\choi M - \rho}_1 + \norm{\rho - \choi N}_1
        \le 2 \norm{\choi M - \rho}_1
        \le 2 \cdot \eps/2
        \le \eps. \qedhere
    \end{equation*}
\end{proof}

\CoherentTomAcid*

\begin{proof}
Let $T_0$ be an ancilla-assisted, coherent, adaptive channel tomography algorithm such that for all channels $\chan M \in \cptp \din \dout$, with probability at least $2/3$, the output of $T(\chan M)$ is a channel $\chan N$ such that $\norm{\chan M - \chan N}_{\lsd} < 1/16$. Our goal is to prove that $T_0$ makes $\Omega\Paren{\din^2 \dout^2 / \log(\din \dout)}$ queries.

Under our assumption that $\dout \ge 4$, Oufkir~\cite{O23}\footnote{This is implicit in the proof of \cite[Lemma 2.2]{O23}, with the $1/8$ constant coming from the inequality $\eps \le 1/4$ in the paragraph preceding the lemma. (The stronger inequality $\eps \le 1/16$ in the surrounding \cite[Theorem 2.1]{O23} is not used in the proof of \cite[Lemma 2.2]{O23}.)} proved that there exists a set of channels $C \subseteq \cptp \din \dout$ of size $|C| = \exp\Paren{\Omega\Paren{\din^{2} \dout^{2}}}$ such that for all $\chan M, \chan N \in C$ it holds that $\norm{\chan M - \chan N}_\lsd > 1/8$. Let $T_1$ be the channel tomography algorithm that first executes $T_0$, yielding a measurement outcome $\chan P$, and then performs the following classical post-processing on $\chan P$:
\begin{itemize}
    \item If there exists a channel $\chan N \in C$ such that $\norm{\chan N - \chan P}_{\lsd} < 1/16$, then output that channel $\chan N$. (There cannot exist two such channels $\chan N$, by the triangle inequality and the definition of $C$.)
    \item Else, output an arbitrary channel in $C$.
\end{itemize}
On input $\chan N \in C$, if $T_0$ successfully outputs a channel $\chan P$ such that $\norm{\chan P - \chan N}_{\lsd} < 1/16$, then the above post-processing leads $T_1$ to output $\chan N$. Therefore for all $\chan N \in C$, the probability that $T_1(\chan N)$ outputs $\chan N$ is at least $2/3$. By repetition and majority vote, there exists an ancilla-assisted, coherent, adaptive channel tester $T_2$, making a number of queries proportional to that made by $T_1$ (and hence by $T_0$), such that for all $\chan N \in C$ the probability that $T_2(\chan N)$ outputs $\chan N$ is at least 0.99.

Let $n$ be the number of queries made by $T_2$, and let $\chan V_0, \dotsc, \chan V_n$ be the sequence of non-query operations performed by $T_2$, where $\chan V_0$ takes as input $\kb0$ and $\chan V_n$ outputs a measurement outcome in $C$ (formally, a diagonal density matrix in $\dens{|C|}$). Our goal is to prove that $n = \Omega\Paren{\din^2 \dout^2 / \log(\din \dout)}$. For $\chan N \in C$ and $0 \le k \le n$, let
\begin{equation*}
    \rho_{\chan N, k} = \chan V_k (\chan N \otimes \chan I) \chan V_{k-1} (\chan N \otimes \chan I) \dotsb \chan V_0 (\kb0).
\end{equation*}
In particular, $\rho_{\chan N, n}$ is the state at the end of the execution of $T_2(\chan N)$. Let $\ket{\chan N}$ denote the standard basis element indexed by the classical description of $\chan N$. Then by a Fuchs--van de Graaf inequality (\cref{eq:fvdg}),
\begin{equation*}
    \frac12 \Norm{\kb{\chan N} - \rho_{\chan N, n}}_1
    \le \sqrt{1 - \fid{\kb{\chan N}} {\rho_{\chan N, n}}}
    \le \sqrt{1 - 0.99}
    = 0.1.
\end{equation*}

We now assign names to the registers that arise throughout the execution of $T_2(\chan N)$, for a channel $\chan N \in C$. For $0 \le k \le n-1$, write $\rho_{\chan N, k} \in \dens{\reg B_k \reg C_k}$ where $\reg B_k$ is a $\din$-dimensional register, and write $(\chan N \otimes \chan I) \rho_{\chan N, k} \in \dens{\reg B_k^\prime \reg C_k}$ where $\reg B_k^\prime$ is a $\dout$-dimensional register. Thus $\chan N$ transforms $\reg B_k$ into $\reg B_k^\prime$. Also write the initial state of the system as $\kb0 \in \dens{\reg B_{-1}^\prime \reg C_{-1}}$, and write $\rho_{\chan N, n} \in \dens{\reg B_n \reg C_n}$, where $\reg B_n \reg C_n$ is a $|C|$-dimensional register. (For notational convenience we write $\reg B_n \reg C_n$ in a manner that suggests the tensor product of distinct registers, despite being a single register.) Thus $\chan V_k$ transforms $\reg{B}_{k-1}^\prime \reg{C}_{k-1}$ into $\reg{B}_k \reg{C}_k$ for all $0 \le k \le n$.

Henceforth we write $\bs{\chan N}$ to denote a uniform random channel in $C$. Let $\reg A$ be a $|C|$-dimensional register, and for $0 \le k \le n$ define a density matrix $\sigma_k \in \dens{\reg A \reg B_k \reg C_k}$ by
\begin{equation*}
    \sigma_k = \E\Brac{\kb{\bs{\chan N}} \otimes \rho_{\bs{\chan N}, k}}.
\end{equation*}
By \cref{lem:FAapp}
\begin{equation*}
    S(\reg A | \reg{B}_n \reg{C}_n)
    \le S_{\E\Brac{\kb{\bs{\chan N}}^{\otimes 2}}} (\reg A | \reg{B}_n \reg{C}_n) + \Norm{\sigma_n - \E\Brac{\kb{\bs{\chan N}}^{\otimes 2}}}_1 \cdot \frac32 \log|C| + 2,
\end{equation*}
and by \cref{dfn:CondEnt}
\begin{equation*}
    S_{\E\Brac{\kb{\bs{\chan N}}^{\otimes 2}}} (\reg A | \reg{B}_n \reg{C}_n)
    = S\Paren{\E\Brac{\kb{\bs{\chan N}}^{\otimes 2}}} - S(I/|C|)
    = \log|C| - \log|C|
    = 0,
\end{equation*}
and
\begin{equation*}
    \Norm{\sigma_n - \E\Brac{\kb{\bs{\chan N}}^{\otimes 2}}}_1
    = \Norm{\E\Brac{\kb{\bs{\chan N}} \otimes \Paren{\rho_{\bs{\chan N}, n} - \kb{\bs{\chan N}}}}}_1
    = \E \Norm{\rho_{\bs{\chan N}, n} - \kb{\bs{\chan N}}}_1
    \le 0.2,
\end{equation*}
so
\begin{equation*}
    S (\reg A | \reg B_n \reg C_n)
    \le 0.3 \log|C| + 2.
\end{equation*}
Furthermore, by \cref{thm:StrongSubadd} and \cref{dfn:CondEnt}
\begin{equation*}
    S(\reg A | \reg B_0 \reg C_0)
    \ge S(\reg A | \reg B_{-1}^\prime \reg C_{-1})
    =  S(\reg A \reg B_{-1}^\prime \reg C_{-1}) - S(\reg B_{-1}^\prime \reg C_{-1})
    = S\Paren{\frac{I}{|C|} \otimes \kb0} - S(\kb0)
    = \log|C|.
\end{equation*}
Below we will prove that $S(\reg A | \reg B_k \reg C_k) - S(\reg A | \reg B_{k+1} \reg C_{k+1}) \le 2 \log(\din \dout)$ for all $0 \le k \le n-1$. It follows that
\begin{equation*}
    0.7 \log|C| - 2
    \le S(\reg A | \reg B_0 \reg C_0) - S(\reg A | \reg B_t \reg C_t)
    = \sum_{k=0}^{n-1} \Paren{S(\reg A | \reg B_k \reg C_k) - S(\reg A | \reg B_{k+1} \reg C_{k+1})}
    \le n \cdot 2 \log(\din \dout),
\end{equation*}
and therefore $n = \Omega(\log|C| / \log(\din\dout)) = \Omega\Paren{\din^2 \dout^2 / \log(\din \dout)}$ as desired.

Fix some $0 \le k \le n-1$ and write $\reg B = \reg B_k, \reg B^\prime = \reg B_k^\prime, \reg C = \reg C_k$. Also for $\chan N \in C$ let $S_{\chan N}$ denote entropy with respect to $\rho_{\chan N, k}$ (or $(\chan N \otimes I) \rho_{\chan N, k}$). Then as promised,
\begin{align*}
    &S(\reg A | \reg B_k \reg C_k) - S(\reg A | \reg B_{k+1} \reg C_{k+1}) \\
    \le &S(\reg A | \reg{BC}) - S(\reg A | \reg{B}^\prime \reg C)
    &&\text{\cref{thm:StrongSubadd}} \\
    = &S(\reg{ABC}) - S(\reg{BC}) - S(\reg{A} \reg{B}^\prime \reg{C}) + S(\reg{B}^\prime \reg{C})
    &&\text{\cref{dfn:CondEnt}} \\
    = &\E\Brac{\Paren{\log|C| + S_{\bs{\chan N}} (\reg{BC})} - S(\reg{BC}) - \Paren{\log|C| + S_{\bs{\chan N}} (\reg{B}^\prime \reg{C})} + S(\reg{B}^\prime \reg{C})}
    &&\text{\cref{lem:EntMix}} \\
    = &\E\Brac{S_{\bs{\chan N}} (\reg{BC}) - S(\reg{BC}) - S_{\bs{\chan N}} (\reg{B}^\prime \reg{C}) + S(\reg{B}^\prime \reg{C})} \\
    \le &\E\Brac{(S_{\bs{\chan N}}(\reg B) + S_{\bs{\chan N}} (\reg C)) - (S(\reg C) - S(\reg B)) - (S_{\bs{\chan N}}(\reg C) - S_{\bs{\chan N}}(\reg B^\prime)) + (S(\reg{B}^\prime) + S(\reg{C}))}
    &&\text{\cref{thm:Subadd,thm:Triangle}} \\
    = &\E\Brac{S_{\bs{\chan N}}(\reg B) + S(\reg B) + S_{\bs{\chan N}}(\reg B^\prime) + S(\reg{B}^\prime)} \\
    \le &\E[2\log\din + 2\log\dout]
    &&\text{\cref{thm:maxEntropy}} \\
    = &2\log(\din\dout).
    &&\qedhere
\end{align*}
\end{proof}

\section*{Acknowledgments}
\addcontentsline{toc}{section}{Acknowledgments}

GR, AD, TG are supported by the EPSRC New Horizons grant EP/X018180/1. SS was supported by a Royal Commission for the Exhibition of 1851 Research Fellowship when this work was carried out. TG and HA are supported by ERC Starting Grant 101163189 and UKRI Future Leaders Fellowship MR/X023583/1. We thank Min-Hsiu Hsieh, Tony Metger, Jon Wright, Henry Yuen, and Haimeng Zhao for helpful discussions.

\section*{Author contributions}

All authors contributed to the research as well as writing and editing of the manuscript. 

\appendix

\section{The ACID norm}
\label{sec:acid-definitions}
In this appendix, we describe in greater detail the relations between the ACID norm and other well-studied norms and distance measures.

\subsection{Relation to statistical distance between Boolean functions}\label{sec:Boolclass}

The \emph{statistical distance} between Boolean functions $f, g: [d] \to \bits$ is the quantity
\begin{equation} \label{eq:stat-dist}
    |f - g| = \frac1d \sum_{j=1}^d |f(j) - g(j)|.
\end{equation}
This is the fraction of inputs on which $f$ and $g$ disagree, and is the standard notion of distance used in (classical or quantum) property testing of Boolean functions.

Let $\chan F, \chan G \in \cptp d 2$ be the channels that measure their input in the standard basis, yielding a measurement outcome $\bs j \in [d]$, and then output $f(\bs j), g(\bs j)$ respectively. Formally,
\begin{align*}
    &\chan F(X) = \sum_{j=1}^d \ketbra{f(j)}{j} X \ketbra{j}{f(j)},
    &&\chan G(X) = \sum_{j=1}^d \ketbra{g(j)}{j} X \ketbra{j}{g(j)}.
\end{align*}
This encoding of $f$ and $g$ as channels captures the setting where only \emph{classical} queries may be made to $f$ and $g$; in \cref{subsec:QBF} we will consider encodings that allow quantum queries.

It follows from definitions that
\begin{align*}
    &\choi F = \frac1d \sum_{j=1}^d \kb{f(j)} \otimes \kb{j},
    &&\choi G = \frac1d \sum_{j=1}^d \kb{g(j)} \otimes \kb{j},
\end{align*}
so
\begin{equation*}
    \frac12 \norm{\chan F - \chan G}_{\lsd}
    = \frac1{2d} \Norm{\sum_{j=1}^d \Paren{\kb{f(j)} - \kb{g(j)}} \otimes \kb{j}}_1
    = |f-g|,
\end{equation*}
i.e.\ ACID distance generalizes statistical distance between Boolean functions.

\subsection{Relation to average-case distance between unitaries} \label{subsec:avg-unitary}

Throughout this subsection let $U, V \in \C^{d \times d}$ be arbitrary unitaries. Low~\cite[Definition 10 and Eq.\ 7]{Low2009Avg} used the distance
\begin{equation*}
    D(U,V)
    \coloneqq \sqrt{1 - \frac1{d^2} \Mag{\tr\Paren{\adj U V}}^2}
    = \frac1{\sqrt{2} d} \norm{U \otimes \adj U - V \otimes \adj V}_2
\end{equation*}
in the context of unitary testing and tomography. Montanaro and de Wolf~\cite[Proposition 21]{MdW13} proved that
\begin{equation} \label{eq:acid-unitary-l2}
    D(U,V)
    = \sqrt{\frac{d+1}{4d} \E\Brac{\Norm{U \bs\psi \adj U - V \bs\psi \adj V}_1^2}}
\end{equation}
where $\bs\psi \in \dens d$ is Haar random, giving an interpretation of $D$ as an ``average-case distance". ACID distance generalizes $D$ because if channels $\chan U, \chan V \in \cptp d d$ conjugate by $U$ and $V$ respectively, then by \cref{eq:pure-state-td,lem:tr-choi-prod},
\begin{equation} \label{eq:MdW-dist}
    \frac12 \norm{\chan U - \chan V}_{\lsd}
    = \frac12 \norm{\choi U - \choi V}_1
    = \sqrt{1 - \tr(\choi U \choi V)}
    = D(U,V).
\end{equation}
Zhao et al.~\cite[Lemma 22 and its proof]{Zhao+23} independently observed \cref{eq:acid-unitary-l2,eq:MdW-dist} as well. They used the distance $\sqrt{\frac14 \E\Brac{\Norm{U \bs\psi \adj U - V \bs\psi \adj V}_1^2}}$ in the context of unitary tomography, a quantity which is within a universal constant factor of $D(U,V)$ by \cref{eq:acid-unitary-l2} and the fact that $1 \le (d+1)/d \le 2$.

Wang~\cite[Eq.\ 4]{Wang2011Unitary} and Chen, Nadimpalli and Yuen~\cite[Definition 7]{Chen2023Junta} used the distance
\begin{equation*}
    D^\prime (U,V) \coloneqq \frac1{\sqrt{2d}} \min_{\substack{\phi \in \C \\ |\phi|=1}} \norm{\phi U - V}_2
\end{equation*}
in the context of unitary testing and tomography. Wang~\cite[Eq.\ 6]{Wang2011Unitary} and Zhao et al.~\cite[Lemma 4(1) and its proof]{Zhao+23} independently observed that
\begin{equation*}
    D^\prime (U,V)^2
    = 1 - \frac1d \Mag{\tr\Paren{\adj U V}},
\end{equation*}
and since $\frac1d \Mag{\tr\Paren{\adj U V}} \le 1$ by Cauchy-Schwarz, it follows that
\begin{equation*}
    D(U,V)^2
    = D^\prime(U,V)^2 \cdot \Paren{1 + \frac1d \Mag{\tr\Paren{\adj U V}}}
    \le 2 D^\prime(U,V)^2
    \le 2 D(U,V)^2,
\end{equation*}
i.e.\ $D^\prime$ is within a constant factor of $D$ (and hence of ACID distance).

\subsection{Relation to distance between quantum Boolean functions} \label{subsec:QBF}

A \emph{quantum Boolean function} is a Hermitian unitary transformation. This definition was introduced by Montanaro and Osborne, and generalizes the standard encodings of (classical) Boolean functions as unitaries~\cite[Section 3]{MO10}. Montanaro and Osborne defined the distance between quantum Boolean functions $F,G \in \C^{d \times d}$ as $\Delta(F,G) \coloneqq \norm{F-G}_2^2 / 4d$~\cite[Definition 11 and Eq.\ 5]{MO10} in the context of property testing and tomography~\cite[Sections 6 and 7]{MO10}. This notion of distance generalizes statistical distance between (classical) Boolean functions (i.e.\ \cref{eq:stat-dist}), in the sense that if $f,g: [d] \to \bits$ are Boolean functions and
\begin{align*}
    &F = \sum_{j=1}^d (-1)^{f(j)} \kb{j},
    &&G = \sum_{j=1}^d (-1)^{g(j)} \kb{j},
\end{align*}
then a straightforward calculation shows that $\Delta(F,G) = |f-g|$. If we use the alternative encoding
\begin{align*}
    &F^\prime = \sum_{j=1}^d \kb{j} \otimes \Paren{\ketbra01 + \ketbra10}^{f(j)},
    &&G^\prime = \sum_{j=1}^d \kb{j} \otimes \Paren{\ketbra01 + \ketbra10}^{g(j)},
\end{align*}
then similarly $\Delta(F^\prime, G^\prime) = 2|f-g|$.

Now consider arbitrary quantum Boolean functions $F, G \in \C^{d \times d}$. Since $F,G$ are Hermitian it holds that $\tr(FG)$ is real. Up to a $\pm1$ global phase, we may further assume that $\tr(FG)$ is nonnegative, and then $\Delta(F,G) = D^\prime(F,G)^2 / 2$ for $D^\prime$ defined as in \cref{subsec:avg-unitary}. Recalling that $D^\prime$ is proportional to ACID distance, it follows that $\Delta$ is proportional to squared ACID distance.

\subsection{Relation to the diamond norm}

\begin{thm}[{Brand\~ao, Piani and Horodecki~\cite[Lemma 6]{BPH15}}] \label{thm:diamond-acid-relat}
    For all Hermitian-preserving superoperators $\chan L \in \spr d$, it holds that
        $\frac1d \norm{\chan L}_\diamond
        \le \norm{\chan L}_{\lsd}
        \le \norm{\chan L}_\diamond$.
\end{thm}

We reproduce their proof below:

\begin{proof}
    The second inequality follows directly from the definitions of the ACID and diamond norms. For the first inequality, let $\ket\psi \in \C^d \otimes \C^d$ be a pure state such that $\norm{\chan L}_\diamond = \norm{(\chan L \otimes \chan I_d) \psi}_1$. Let $\ket\psi = \sum_i \sqrt{p_i} \ket{u_i} \ket{v_i}$ be a Schmidt decomposition of $\ket\psi$, i.e.\ the $p_i$ form a probability distribution, the $\ket{u_i}$ form an orthonormal basis for $\C^d$, and the $\ket{v_i}$ also form an orthonormal basis for $\C^d$. Let $A = \sum_i \sqrt{p_i} \ketbra{v_i}{u_i^*}$. Then
    \begin{equation*}
        \ket\psi
        = (I \otimes A) \sum_i \ket{u_i} \ket{u_i^*}
        = \sqrt{d} (I \otimes A) \ket\Phi,
    \end{equation*}
    so by the definition of $\ket\psi$,
    \begin{equation*}
        \norm{\chan L}_\diamond
        = \norm{(\chan L \otimes \chan I)\psi}_1
        = d \norm{(\chan L \otimes \chan I) \cdot (I \otimes A) \Phi (I \otimes \adj A)}_1
        = d \norm{(I \otimes A) \choi L (I \otimes \adj A)}_1,
    \end{equation*}
    where the last equality holds because conjugating by $A$ on the second register commutes with applying $\chan L$ on the first register. The expression $A = \sum_i \sqrt{p_i} \ketbra{v_i}{u_i^*}$ is a singular value decomposition of $A$, and therefore $\norm{A}_\infty = \max_i \sqrt{p_i} \le 1$, so by H\"older's inequality $\norm{\chan L}_\diamond \le d \norm{\choi L}_1 = d \norm{\chan L}_{\lsd}$.
\end{proof}

The first inequality in \cref{thm:diamond-acid-relat} may be tight, for example if $\chan L(X) = \bra0 X \ket0$. The second inequality in \cref{thm:diamond-acid-relat} may also be tight, for example if $\chan L$ is a channel, or if $\chan L(X) = \tr(X) A$ for some fixed matrix $A$, or if $\chan L$ is the transpose superoperator $\chan L(X) = X^\top$.

Jen{\v{c}}ov{\'a} and Pl{\'a}vala~\cite{JP16} proved the following inequality, where $|A| \coloneqq \sqrt{A^2}$ denotes the matrix absolute value of a Hermitian matrix $A$:

\begin{thm}[{\cite[Eq.\ 11]{JP16}}] \label{thm:JP}
    Let $\chan L = \lambda \chan M - (1-\lambda) \chan N$ for some $\lambda \in (0,1)$ and channels $\chan M, \chan N \in \cptp \din \dout$. Then
    \begin{equation*}
    \Norm{\chan L}_\diamond
    \le \Paren{1 + \Norm{\frac{\din} {\norm{\chan L}_{\lsd}} \Paren{\mrm{tr}_{\dout} \otimes \chan I_{\din}} |\choi L| - I_{\din}}_\infty} \norm{\chan L}_{\lsd},
\end{equation*}
\end{thm}

Jen{\v{c}}ov{\'a} and Pl{\'a}vala~\cite{JP16} also observed that \cref{thm:JP} gives a stronger bound than \cref{thm:diamond-acid-relat} does. To see this, note that $\tr(|\choi L|) = \norm{\choi L}_1 = \norm{\chan L}_{\lsd}$, so $|\choi L| / \norm{\chan L}_{\lsd}$ is a density matrix, and therefore its partial trace $(\tr_{\dout} \otimes \chan I_{\din}) |\choi L| / \norm{\chan L|}_{\lsd}$ is also a density matrix. Since the eigenvalues of a density matrix are between 0 and 1, the infinity norm appearing in \cref{thm:JP} is at most $\din-1$ (assuming $\din \ge 2$), and therefore the upper bound from \cref{thm:JP} is at most $\din \norm{\chan L}_{\lsd}$.

\subsection{Relation to the induced trace norm and its average-case analogue} \label{subsec:avg-case-induced-trace}

Recall that the \emph{induced trace norm} of a superoperator $\chan L$ is the quantity $\norm{\chan L}_1 \coloneqq \max_{\norm{X}_1 = 1} \norm{\chan L(X)}_1$. The following example shows that the induced trace distance between two channels can be much less than their diamond distance:

\begin{xmp}[{Watrous~\cite[Example 3.36]{Wat18}}] \label{xmp:watrous}
    Define channels $\chan M, \chan N \in \cptp d d$ by
    \begin{align*}
        &\chan M(X) = \frac{\tr(X) I_d + X^\top}{d+1},
        &&\chan N(X) = \frac{\tr(X) I_d - X^\top}{d-1}.
    \end{align*}
    These are in fact channels because they are clearly trace-preserving, and because their Choi states
    \begin{align*}
        &\choi M = \frac2{d(d+1)} \sym d,
        &&\choi N = \frac2{d(d-1)} \Paren{I-\sym d}
    \end{align*}
    are PSD. Observe that
    \begin{equation*}
        \norm{\chan M - \chan N}_1
        = \max_\psi \norm{\chan M\psi - \chan N\psi}_1
        = \max_\psi \Norm{\frac{I + \psi^\top}{d+1} - \frac{I - \psi^\top}{d-1}}_1
        = \max_\psi \Norm{\frac{2(d\psi^\top - I)}{(d+1)(d-1)}}_1
        = \frac4{d+1},
    \end{equation*}
    where the last equality holds because $d\psi^\top - I$ has one eigenvalue equal to $d-1$ and $d-1$ eigenvalues equal to $1$. On the other hand,
    \begin{equation*}
        \norm{\chan M - \chan N}_\diamond
        \ge \norm{(\chan M \otimes \chan I) \Phi - (\chan N \otimes \chan I) \Phi}_1
        = \norm{\choi M - \choi N}_1
        = 2,
    \end{equation*}
    where the last equality holds because $\choi M$ and $\choi N$ are supported on orthogonal subspaces.
\end{xmp}

We note that \cref{xmp:watrous} holds equally well with the ACID norm in place of the diamond norm, and with an ``average-case induced trace norm" in place of the induced trace norm:

\begin{obs} \label{obs:watrous}
    It is implicit in \cref{xmp:watrous} that $\norm{\chan M - \chan N}_\lsd = 2$, and that $\E \norm{(\chan M - \chan N) \bs\psi}_1 = 4/(d+1)$ for Haar random $\bs\psi$.
\end{obs}

Thus, ancillae can be information-theoretically useful for distinguishing between two channels in the average-case setting as well as in the worst-case setting. Recall that \cref{thminf:compil} says that if $\chan L \in \spr d$ is a superoperator and $m = \Omega(d)$, then $\norm{(\chan L \otimes \chan I_m) \bs\psi}_1$ concentrates around $\norm{\chan L}_J$; \cref{obs:watrous} implies that this statement does not generalize to arbitrary values of $m$.

In \cref{sec:E-Phi-def} (specifically \cref{prp:u-inv}) we will prove a generalization of the fact that $\E\norm{\chan L(\bs\psi)}_1 \le \norm{\chan L}_J$ for all superoperators $\chan L$, where $\bs\psi$ is Haar random. Here we give an alternate proof of this fact in the case where $\chan L$ is the difference between two unitary channels, i.e.\ $\chan L(X) = U X \adj U - V X \adj V$ for some unitaries $U, V \in \C^{d \times d}$. If $D$ denotes the average-case distance between unitaries from \cref{subsec:avg-unitary}, then by Cauchy-Schwarz and \cref{eq:acid-unitary-l2,eq:MdW-dist},
\begin{equation*}
    \E\norm{\chan L(\bs\psi)}_1
    \le \sqrt{\E\Brac{\norm{\chan L(\bs\psi)}_1^2}}
    = \sqrt{\frac{4d}{d+1}} D(U,V)
    \le 2 D(U,V)
    = \norm{\chan L}_J.
\end{equation*}

Finally we note that unlike ACID distance, average-case induced trace distance fails to generalize statistical distance between Boolean functions, at least according to the encoding of functions $f$ and $g$ as channels $\chan F$ and $\chan G$ used in \cref{sec:Boolclass}. Specifically, if $\bs\psi \in \dens d$ is Haar random and $\bs p_j = \bra{j} \bs\psi \ket{j}$, then
\begin{equation*}
    \E\norm{(\chan F - \chan G) \bs\psi}_1
    = \E \Norm{\sum_{j=1}^d (\kb{f(j)} - \kb{g(j)}) \bs p_j}_1
    = 2 \E\Mag{\sum_{j=1}^d (f(j) - g(j)) \bs p_j},
\end{equation*}
where the last equality holds because $\kb{f(j)} - \kb{g(j)} = (f(j) - g(j)) \cdot (\kb1 - \kb0)$. If $f(j)=0, g(j)=1$ and for half of the inputs $j$ and $f(j)=1, g(j)=0$ for the other half, then $|f-g|=1$, but $\sum_j (f(j)-g(j)) \bs p_j$ concentrates around 0 and so $\E\norm{(\chan F - \chan G) \bs\psi}_1$ is close to 0.

\subsection{Relation to quantum fault-tolerance and experiments} \label{subsec:fault-tolerance}

We continue the discssion from the end of \cref{subsec:discussion}. Gilchrist, Langford and Nielsen~\cite{GLN05} proposed six properties that any distance $\Delta(\chan M, \chan N)$ between channels $\chan M$ and $\chan N$ should have in order to be suitable for measuring the error of a quantum computation: it should be a metric, be easy to calculate, be easy to experimentally measure, have a well-motivated physical interpretation, satisfy \emph{stability} (i.e.\ $\Delta(\chan I \otimes \chan M, \chan I \otimes \chan N) = \Delta(\chan M, \chan N)$), and satisfy \emph{chaining} (i.e.\ $\Delta(\chan M_2 \chan M_1, \chan N_2 \chan N_1) \le \Delta(\chan M_1, \chan N_1) + \Delta(\chan M_2, \chan N_2)$). Kueng, Long, Doherty and Flammia~\cite[Eqs.\ 2 and 3]{KLDF16} noted the significance of stability and chaining as well. Out of many candidate distances, Gilchrist, Langford and Nielsen~\cite{GLN05} identified four that satisfy these criteria: ACID distance (which they call Jamiołkowski process distance or $J$ distance), related distances arising from fidelity (i.e.\ Jamiolkowski process fidelity or $J$ fidelity), diamond distance (i.e.\ stabilized process distance or $S$ distance), and related distances arising from fidelity (i.e.\ stabilized process fidelity or $S$ fidelity).\footnote{In particular, they rejected the ``average-case induced trace distance" from \cref{subsec:avg-case-induced-trace} as a candidate distance~\cite[Eq.\ 13]{GLN05}, for only seeming to satisfy the metric and chaining criteria out of the six.} They also gave an operational interpretation of the ACID norm as a bound on the ``average probability of error experienced during quantum computation of a function, or as a bound on the distance between the real and ideal joint distributions of the quantum computer in a sampling computation"~\cite[Section VI.(i)]{GLN05}.

This suggests that one may hope to prove a fault-tolerance theorem with respect to the ACID norm. Aharonov, Kitaev and Nisan~\cite[Lemma 12]{Aharonov98} listed five\footnote{As well as a sixth, $\norm{\chan L(X)}_1 \le \norm{\chan L}_\diamond \norm{X}_1$, which follows immediately from the first property and the definition of the induced trace norm.} properties of the diamond norm that are used in the proof of the fault-tolerance theorem: for all superoperators $\chan K, \chan L$:
\begin{enumerate}
    \item $\norm{\chan L}_\diamond = \norm{\chan L \otimes \chan I_{\daux}}_1 \ge \norm{\chan L}_1$ for all $\daux \ge \din$, where $\din$ is the input dimension of $\chan L$.
    \item $\norm{\chan K \chan L}_\diamond \le \norm{\chan K}_\diamond \norm{\chan L}_\diamond$, i.e.\ the diamond norm is submultiplicative.
    \item $\norm{\chan K \otimes \chan L}_\diamond = \norm{\chan K}_\diamond \norm{\chan L}_\diamond$.
    \item If $\chan L$ is a channel then $\norm{\chan L}_\diamond = 1$.
    \item If $A,B$ are matrices of the same dimensions with $\norm{A}_\infty, \norm{B}_\infty \le 1$, and if $\chan L(X) = A X \adj A - B X \adj B$, then $\norm{\chan L}_\diamond \le 2\norm{A - B}_\infty$.
\end{enumerate}

Unfortunately not all of these properties hold with the ACID norm in place of the diamond norm. However, analogous properties may hold if we \emph{also} replace other worst-case quantities besides just the diamond norm with their average-case analogues:
\begin{enumerate}
    \item If $\chan L(X) = \bra0 X \ket0$ for example, then $\norm{\chan L}_J = 1/\din$ while $\norm{\chan L \otimes \chan I_{\daux}}_1 = \norm{\chan L}_1 = \chan L(\kb0) = 1$ for all $\daux$. However, if we also replace the induced trace norm with the ``average-case induced trace norm" from \cref{subsec:avg-case-induced-trace}, then we recall that $\norm{\chan L}_J \ge \E\norm{\chan L(\bs\psi)}_1$ for Haar random $\bs\psi$, and furthermore $\norm{\chan L}_J$ is proportional to $\E\norm{(\chan L \otimes \chan I_{\daux}) \bs\psi}_1$ for $\daux \ge \din$ by \cref{thminf:compil}.
    \item The ACID norm is not submultiplicative: for example, if $\chan K, \chan L \in \super d d$ with $d>1$ and $\chan K(X) = \chan L(X) = \kb0 X \kb0$, then $\norm{\chan K \chan L}_J = \norm{\chan L}_J = 1/d > 1/d^2 = \norm{\chan K}_J \norm{\chan L}_J$. The problem is that the ACID norm of $\chan K$ describes its behavior on average-case inputs, whereas the output of $\chan L$ is proportional to the ``worst-case input" $\kb0$. However, this issue may conceivably be circumvented if we only consider circuits where the input is average-case, and where individual gates map average-case inputs to average-case outputs. Specifically, we propose a model of computation using only unitary gates, gates that initialize new qubits in the maximally mixed state (as opposed to the all-zeros state), and gates that trace out qubits; this generalizes ancilla-free computation and is related to the ``one clean qubit" model~\cite{KL98}.
    \item It holds that $\norm{\chan K \otimes \chan L}_J = \Norm{\choi K \otimes \choi L}_1 = \norm{\chan K}_J \norm{\chan L}_J$.
    \item If $\chan L$ is a channel then $\norm{\chan L}_J = 1$ because $\choi L$ is a density matrix.
    \item Since $\norm{\chan L}_J \le \norm{\chan L}_\diamond$, the analogous property with the ACID norm in place of the diamond norm follows immediately.
\end{enumerate}
This leaves open the possibility of a ``fully average-case" version of the fault-tolerance theorem.

A related question is how to efficiently test whether a quantum gate (or module of many gates) achieves a level of error below the threshold required for such a fault-tolerance theorem. Even when we can achieve dimension-independent upper bounds for this task, a remaining problem is that what is typically measurable are fidelities, which are only quadratically related to trace norm-based quantities via the Fuchs--van de Graaf inequalities (see the discussion after \cref{thminf:erasure-etc-upper-bounds}). This presents a serious problem~\cite{KLDF16} because as the quality of quantum hardware improves and fidelities rise, the square root leaves a significant gap between the experimentally measured and the theoretically prescribed quantities.

\section{Barriers to strengthening the results from \texorpdfstring{\cref{sec:E-Phi-def}}{Section 5}}
\label{app:E-Phi-companion-appendix}

\subsection{Examples where \texorpdfstring{\cref{prp:u-inv}}{Theorem 5.5} is tight} \label{app:u-inv-tight}

Recall the statement of \cref{prp:u-inv}:

\uinv*

The following example shows that the first inequality in \cref{prp:u-inv} is sometimes tight to within a constant factor, for a wide range of values of $\E[\fid{\bs\rho}{I/d}]$:

\begin{xmp}
Let $\chan L \in \super d d$ be the transpose superoperator, i.e.\ $\chan L(X) = X^\top$, and let $\bs\rho \in \dens d$ be maximally mixed on a Haar random $r$-dimensional subspace of $\C^d$. Then
\begin{equation*}
    \fid{\bs\rho}{I/d}
    = \tr(\sqrt{\bs\rho})^2 / d
    = r/d
\end{equation*}
pointwise, and
\begin{equation} \label{eq:acid-norm-transpose}
    \norm{\chan L}_\lsd
    = \frac1d \norm{(\chan L \otimes \chan I) \Psi}_1
    = \frac1d \norm{\swap d}_1
    = d,
\end{equation}
 so
\begin{equation*}
\frac{d^2 \E[\fid{\bs\rho}{I/d}] - 1}{d^2(2 - \E[\fid{\bs\rho}{I/d}]) - 1} \norm{\chan L}_\lsd
= \frac{rd - 1}{2d^2 - rd - 1} \cdot d
\ge r/2 - o(1)
\end{equation*}
as $r,d \to \infty$. On the other hand, the state $(\sqrt{\bs\rho} \otimes I) \Psi (\sqrt{\bs\rho} \otimes I)$ is maximally entangled across two $r$-dimensional systems, so
\begin{equation*}
    \norm{\chan L}_{\bs\rho}
    = \norm{(\chan K \otimes \chan I) \cdot (\sqrt{\bs\rho} \otimes I) \Psi (\sqrt{\bs\rho} \otimes I)}_1
    = r
\end{equation*}
pointwise by reasoning similar to that in \cref{eq:acid-norm-transpose}.
\end{xmp}

And the following example shows that the second inequality in \cref{prp:u-inv} is sometimes tight:

\begin{xmp} \label{xmp:comp-pos}
Let $\chan L \in \spr d$ be any completely positive superoperator. Then its Choi operator is PSD, so
\begin{align*}
    \E \norm{\chan L}_{\bs\rho} 
    &= \E \Norm{\Paren{I \otimes \sqrt{\bs\rho}^\top} \cdot (\chan L \otimes \chan I) \Psi \cdot \Paren{I \otimes \sqrt{\bs\rho}^\top}}_1
    &&\text{\cref{prp:equiv-def}} \\
    &= \E \tr\Paren{{\Paren{I \otimes \sqrt{\bs\rho}^\top} \cdot (\chan L \otimes \chan I) \Psi \cdot \Paren{I \otimes \sqrt{\bs\rho}^\top}}}
    &&\text{PSD} \\
    &= \E \tr\Paren{{\Paren{I \otimes {\bs\rho}^\top} \cdot (\chan L \otimes \chan I) \Psi}}
    &&\text{cyclic property of trace} \\
    &= \tr\Paren{{I/d \cdot (\chan L \otimes \chan I) \Psi}}
    &&\text{unitarily invariant} \\
    &= \norm{\chan L}_\lsd
    &&\text{PSD.}
\end{align*}
\end{xmp}

One may object that \cref{xmp:comp-pos} is irrelevant to our ultimate motivation of channel testing, since the difference between two distinct channels cannot be completely positive. The following example also shows that the second inequality in \cref{prp:u-inv} is sometimes tight, and arises for example when $\chan L$ is the difference between two replacement channels, i.e.\ channels of the form $X \mapsto \tr(X) \sigma$ for a fixed density matrix $\sigma$:

\begin{xmp}
    Let $\chan L(X) = \tr(X) A$ for some fixed matrix $A$. Then by \cref{prp:equiv-def}, $\norm{\chan L}_\rho = \norm{A \otimes \rho^\top}_1 = \norm{A}_1$ for all density matrices $\rho$, and in particular $\E \norm{\chan L}_{\bs\rho} = \norm{\chan L}_\lsd$ regardless of the distribution from which $\bs\rho$ is sampled.
\end{xmp}

\subsection{Examples where \texorpdfstring{\cref{prp:u-inv}}{Theorem 5.5} relies on \texorpdfstring{$\bs\rho$}{rho} being random} \label{app:u-inv-random}

The following example shows that the first inequality in \cref{prp:u-inv} may fail to hold if $\bs\rho$ is replaced with a fixed density matrix $\rho$:

\begin{xmp}
    Let $\Pi \in \C^{d \times d}$ be the projection onto an arbitrary $d/2$-dimensional subspace of $\C^d$, and define a state
    $\rho \in \dens d$, reflection $U \in \C^{d \times d}$, and superoperator $\chan L \in \super d d$ by
    \begin{align*}
        &\rho = \frac2d \Pi,
        &&U = I - 2 \Pi,
        &&\chan L(X) = X - UXU
        = 2X\Pi + 2\Pi{X} - 4\Pi{X}\Pi.
    \end{align*}
    Then
    \begin{equation*}
        \norm{\chan L}_\rho
        = \norm{(\chan L \otimes \chan I) \cdot (\sqrt\rho \otimes I) \Psi (\sqrt\rho \otimes I)}_1
        = \frac2d \norm{(\chan L \otimes \chan I) \cdot (\Pi \otimes I) \Psi (\Pi \otimes I)}_1
        = 0.
    \end{equation*}
    On the other hand, by \cref{eq:pure-state-td}
    \begin{equation*}
        \norm{\chan L}_{\lsd}
        = \Norm{\frac1d \Psi - \frac1d (U \otimes I) \Psi (U \otimes I)}_1
        = 2\sqrt{1 - \frac1{d^2} \Mag{\bra\Psi (U \otimes I) \ket\Psi}^2}
        = 2\sqrt{1 - \frac1{d^2} |\tr(U)|^2}
        = 2,
    \end{equation*}
    and $\fid \rho {I/d} = 1/2$, so 
    \begin{equation*}
        \frac{d^2 \fid\rho{I/d} - 1}{d^2(2 - \fid\rho{I/d}) - 1} \norm{\chan L}_{\lsd}
        = \frac{d^2/2 - 1}{d^2 \cdot 3/2 - 1} \cdot 2
        = 2/3 - o(1).
    \end{equation*}
\end{xmp}

And the following example shows that the second inequality in \cref{prp:u-inv} may fail to hold if $\bs\rho$ is replaced with a fixed density matrix $\rho$:

\begin{xmp}
    Let $\chan L$ be any Hermitian-preserving superoperator such that $\norm{\chan L}_\diamond > \norm{\chan L}_{\lsd}$. Let $\psi$ be a pure state such that $\norm{\chan L}_\diamond = \norm{(\chan L \otimes \chan I) \psi}_1$, and let $\rho$ be the reduced state on the first register of $\psi$. Then
        $\norm{\chan L}_\rho
        = \norm{(\chan L \otimes \chan I) \psi}_1
        = \norm{\chan L}_\diamond
        > \norm{\chan L}_{\lsd}$.
\end{xmp} 

\subsection{Concentration of \texorpdfstring{$\norm{\chan L}_{\bs\rho}$}{Lrho} does not directly follow from the triangle inequality} \label{app:S5-nontrivial}

Recall that in \cref{sec:E-Phi-def} we proved that if $\chan L \in \spr d$ is a superoperator, and $\bs\psi \in \dens{d \otimes m}$ is a Haar random state where $m \ge \omega(d)$, then $\norm{(\chan L \otimes \chan I) \bs\psi}_1 = \Theta(\norm{\chan L}_{\lsd})$ with high probability. The reader may wonder, would it not be simpler to prove this by showing that $\ket{\bs\psi}$ is close to maximally entangled across the two registers, and then applying the triangle inequality to show that $\norm{(\chan L \otimes \chan I) \bs\psi}_1$ is close to $\norm{(\chan L \otimes \chan I) \Phi}_1 = \norm{\chan L}_{\lsd}$? Below we carry out this argument and show that it only seems to imply concentration when $m \ge \omega(d^3)$, not $m \ge \omega(d)$.

We will use the case of the following lemma where either $\rho$ or $\sigma$ is maximally mixed:

\begin{lem} \label{prp:lipschitz}
For all superoperators $\chan L \in \spr d$ and density matrices $\rho, \sigma \in \dens d$,
\begin{equation*}
    \norm{\chan L}_\rho - \norm{\chan L}_\sigma
    \le 2 \norm{\chan L}_\diamond \sqrt{1-  \fid \rho \sigma}.
\end{equation*}
\end{lem}
\begin{proof}
By Uhlmann's theorem there exist purifications $\psi, \phi$ of $\rho, \sigma$ respectively such that $\fid \psi \phi = \fid \rho \sigma$. So by the triangle inequality and \cref{eq:pure-state-td},
\begin{align*}
    \norm{\chan L}_\rho - \norm{\chan L}_\sigma
    &= \norm{(\chan L \otimes \chan I) \psi}_1 - \norm{(\chan L \otimes \chan I) \phi}_1 \\
    &\le \norm{(\chan L \otimes \chan I) \cdot (\psi - \phi)}_1 \\
    &\le \norm{\chan L}_\diamond \norm{\psi - \phi}_1 \\
    &= 2 \norm{\chan L}_\diamond \sqrt{1 - \fid \psi \phi} \\
    &= 2 \norm{\chan L}_\diamond \sqrt{1-  \fid \rho \sigma}. \qedhere
\end{align*}
\end{proof}

So if $\bs\rho \in \dens d$ is the reduction of a Haar random state in $\C^d \otimes \C^m$, then
\begin{align*}
    \E\Mag{\norm{\chan L}_{\bs\rho} - \norm{\chan L}_{\lsd}}
    &\le 2 \norm{\chan L}_\diamond \E\sqrt{1 - \fid{\bs\rho}{I/d}}
    &&\text{\cref{prp:lipschitz}} \\
    &\le 2d\norm{\chan L}_{\lsd} \E\sqrt{1 - \fid{\bs\rho}{I/d}}
    &&\text{\cref{thm:diamond-acid-relat}} \\
    &\le 2d\norm{\chan L}_{\lsd} \sqrt{1 - \E \fid{\bs\rho}{I/d}}
    &&\text{Cauchy-Schwarz} \\
    &= 2d\norm{\chan L}_{\lsd} \sqrt{1 - \frac{dm+1}{d(d+m)}}
    &&\text{\cref{eq:lb-fid}} \\
    &= 2 \norm{\chan L}_{\lsd} \sqrt{\frac{d(d^2-1)}{d + m}},
\end{align*}
and the latter expression is $o(\norm{\chan L}_{\lsd})$ when $m \ge \omega(d^3)$.

\section{Proof of \texorpdfstring{\cref{lem:BaoYao}}{Lemma 6.2}} \label{app:BaoYao}

We use the following equality:

\begin{lem} \label{lem:tr-Haar-prod}
    For all matrices $X, Y \in \C^{d \times d}$,
    \begin{equation*}
        \E[\tr(\bs\psi X \bs\psi Y)] = \frac1{d(d+1)} (\tr(X) \tr(Y) + \tr(XY)),
    \end{equation*}
    where $\bs\psi \in \dens d$ is Haar random.
\end{lem}

Gu~\cite{Gu13} proved a significant generalization of \cref{lem:tr-Haar-prod} using Weingarten calculus. For completeness and simplicity, below we present a self-contained proof of \cref{lem:tr-Haar-prod} (essentially due to Montanaro and de Wolf~\cite[proof of Proposition 21]{MdW13}) using only \cref{lem:sym}:

\begin{proof}
    We have
    \begin{align*}
        \E[\tr(\bs\psi X \bs\psi Y)]
        &= \E[\bra{\bs\psi} X \kb{\bs\psi} Y \ket{\bs\psi}] \\
        &= \E[\tr(X \bs\psi) \tr(Y \bs\psi)] \\
        &= \tr\Paren{(X \otimes Y) \E[\bs\psi^{\otimes 2}]} \\
        &= \frac1{d(d+1)} \tr\Paren{(X \otimes Y)(I + \swap d)},
    \end{align*}
    where the last equality is by \cref{lem:sym}. Clearly
    \begin{equation*}
        \tr(X \otimes Y) = \tr(X) \tr(Y),
    \end{equation*}
    and furthermore
    \begin{equation*}
        \tr((X \otimes Y) \swap d)
        = \sum_{j,k=1}^d \bra{jk} (X \otimes Y) \swap d \ket{jk}
        = \sum_{j,k} \bra{j} X \kb{k} Y \ket{j}
        = \sum_j \bra{j} X Y \ket{j} = \tr(XY).
    \end{equation*}
    The result follows by combining the above three equations.
\end{proof}

\cref{lem:BaoYao} is the case of the following where $\chan L = \chan K$ is the difference between two channels:

\begin{lem} \label{lem:Choi_l2_dist}
    For all Hermitian-preserving superoperators $\chan L, \chan K \in \spr d$,
    \begin{equation*}
        \frac{d+1}d \E[\tr(\chan L(\bs\psi) \chan K(\bs\psi))]
        = \tr(\choi L \choi K) + \tr(\chan L(I/d) \chan K(I/d)),
    \end{equation*}
    where $\bs\psi \in \dens d$ is Haar random.
\end{lem}

\begin{proof}
    By linearity and \cref{eq:kraus}, we may assume without loss of generality that $\chan L$ and $\chan K$ are defined by $\chan L(X) = A X \adj A$ and $\chan K(X) = B X \adj B$ respectively for some matrices $A$ and $B$. Then
    \begin{align*}
        \E[\tr(\chan L(\bs\psi) \chan K(\bs\psi))]
        &= \E\Brac{\tr\Paren{A \bs\psi \adj A B \bs\psi \adj B}} \\
        &= \E\Brac{\tr\Paren{\bs\psi \adj A B \bs\psi \adj B A}} \\
        &= \frac1{d(d+1)} \Paren{\tr\Paren{\adj A B} \tr\Paren{\adj B A} + \tr\Paren{\adj A B \adj B A}}
        &&\text{\cref{lem:tr-Haar-prod}} \\
        &= \frac1{d(d+1)} \Paren{\Mag{\tr\Paren{\adj A B}}^2 + \tr\Paren{A \adj A B \adj B}} \\
        &= \frac1{d(d+1)} \Paren{d^2 \tr(\choi L \choi K) + \tr(\chan L(I) \chan K(I))}
        &&\text{\cref{lem:tr-choi-prod}} \\
        &= \frac{d}{d+1} \Paren{\tr(\choi L \choi K) + \tr(\chan L(I/d) \chan K(I/d))}.
        &&\qedhere
    \end{align*}
\end{proof}

\printbibliography[heading=bibintoc]
\end{document}